\documentclass[a4paper,11pt]{book}
\usepackage[T1]{fontenc}
\usepackage[utf8]{inputenc}
\usepackage{lmodern}
\usepackage{comment}
\usepackage{standalone, gensymb, tensor, amsmath,amsfonts,amssymb,amsthm,bbm,bm, braket}
\usepackage[T1]{fontenc}
\usepackage{hyperref}
\usepackage[inline]{enumitem}
\usepackage{scalerel, physics}
\usepackage{wasysym}
\usepackage{comment}
\usepackage{enumerate}
\usepackage{graphicx}
\usepackage{mathtools}
\usepackage{soul}
\usepackage{xparse}
\usepackage{dsfont}
\usepackage{ifthen}
\usepackage{tensor}
\usepackage{tikz}
\usepackage{pgfplots}
\usepackage{tikz-network}
\usepackage{simplewick} 
\usepackage{mathrsfs}

\usetikzlibrary{patterns,decorations.pathreplacing}
\theoremstyle{break}        

\theoremstyle{break}

\usepackage{color}
\definecolor{OliveGreen}{RGB}{85,107,47}
\definecolor{NavyBlue}{RGB}{0,0,128}
\usepackage{tensor}
\usepackage{xifthen}
\usepackage{xargs}

\definecolor{blue1}{RGB}{0, 0, 255}
\definecolor{blue2}{RGB}{0, 150, 255}
\definecolor{blue3}{RGB}{0, 71, 171}
\definecolor{blue4}{RGB}{100, 149, 237}
\definecolor{blue5}{RGB}{93, 63, 211}
\definecolor{blue6}{RGB}{101,147,245}
\definecolor{blue7}{RGB}{176,223,229}
\definecolor{blue8}{RGB}{0,0,128}
\definecolor{blue9}{RGB}{0,108,255}
\definecolor{blue10}{RGB}{101,147,245}
\definecolor{blue11}{RGB}{115,194,251}
\definecolor{blue12}{RGB}{87,160,211}
\definecolor{blue13}{RGB}{137,207,240}
\definecolor{blue14}{RGB}{29,41,81}
\definecolor{blue15}{RGB}{14,77,146}
\definecolor{blue16}{RGB}{15,82,186}

\definecolor{red1}{RGB}{238, 75, 43}
\definecolor{red2}{RGB}{233, 116, 81}
\definecolor{red3}{RGB}{222, 49, 99}
\definecolor{red4}{RGB}{250, 160, 160}
\definecolor{red5}{RGB}{236, 88, 0}
\definecolor{red6}{RGB}{232,102,102}
\definecolor{red7}{RGB}{202,52,51}
\definecolor{red8}{RGB}{205,92,92}
\definecolor{red9}{RGB}{178,34,34}
\definecolor{red10}{RGB}{164,90,82}
\definecolor{red11}{RGB}{255,8,0}
\definecolor{red12}{RGB}{202,52,51}
\definecolor{red13}{RGB}{66,13,9}
\definecolor{red14}{RGB}{141,2,31}
\definecolor{red15}{RGB}{250,128,114}
\definecolor{red16}{RGB}{237,41,57}

\definecolor{yellow1}{RGB}{254,220,86}
\definecolor{yellow2}{RGB}{255,229,180}
\definecolor{yellow3}{RGB}{238,220,130}
\definecolor{yellow4}{RGB}{253,165,15}
\definecolor{yellow5}{RGB}{255,195,11}
\definecolor{yellow6}{RGB}{218,165,32}
\definecolor{yellow7}{RGB}{255,211,0}
\definecolor{yellow8}{RGB}{248,222,126}
\definecolor{yellow9}{RGB}{245,245,220}
\definecolor{yellow10}{RGB}{248,228,115}

\definecolor{grey1}{RGB}{98,98,98}
\definecolor{grey2}{RGB}{211,211,211}
\definecolor{grey3}{RGB}{192,192,192}
\definecolor{grey4}{RGB}{169,169,169}
\definecolor{grey5}{RGB}{246,246,246}
\definecolor{grey6}{RGB}{32,32,32}
\definecolor{grey7}{RGB}{64,64,64}
\definecolor{grey8}{RGB}{96,96,96}
\definecolor{grey9}{RGB}{128,128,128}
\definecolor{grey10}{RGB}{160,160,160}
\definecolor{grey11}{RGB}{224,224,224}
\definecolor{grey12}{RGB}{180,180,180}

\definecolor{green1}{RGB}{80, 180, 152}

\definecolor{orange}{RGB}{255, 116, 23}
\definecolor{orange2}{RGB}{244, 174, 114}

\newcommandx{\fineq}[5][1=-.8ex,2=1,3=1,5=0]{
	\begin{tikzpicture}[baseline={([yshift=#1]current  bounding  box.center)}, scale = #2, every node/.style={scale = #3},rotate around={#5:(0,0)},every node/.style={transform shape}]
		#4
	\end{tikzpicture}
}

\newcommandx{\tikzdiagup}{
	\tikz {\draw[thick] (0,0)--(0.15,0.15); \draw (0,0) rectangle (0.15,0.15);}
}

\newcommandx{\gatecross}[1][1=0.5]{
	\pgfmathparse{#1/2.0}
	\let\x\pgfmathresult
	\draw[thick] (-\x,-\x) -- (\x,\x);
	\draw[thick] (\x,-\x) -- (-\x,\x);
}
\newcommandx{\gatesqu}[2][1=0.25,2=]{
	\pgfmathparse{#1/2.0}
	\let\x\pgfmathresult
	\ifthenelse{\equal{#2}{}}{
		\draw[thick, fill=white, rounded corners=2pt] (-\x,\x) rectangle (\x,-\x);
	}{
		\draw[thick, fill=#2, rounded corners=2pt] (-\x,\x) rectangle (\x,-\x);
	}
}

\newcommandx{\gatemark}[2][1=0.075,2=tr]{
	\pgfmathparse{#1}
	\let\l\pgfmathresult
	\ifthenelse{\equal{#2}{topleft}}{
		\draw[thick] (0,\l) -- ++(-\l,0) --++ (0,-\l);
	}{}
	\ifthenelse{\equal{#2}{topright}}{
		\draw[thick] (0,\l) -- ++(\l,0) --++ (0,-\l);
	}{}
	
	\ifthenelse{\equal{#2}{bottomleft}}{
		\draw[thick] (0,-\l) -- ++(-\l,0) --++ (0,\l);
	}{}
	\ifthenelse{\equal{#2}{bottomright}}{
		\draw[thick] (0,-\l) -- ++(\l,0) --++ (0,\l);
	}{}
	\ifthenelse{\equal{#2}{right}}{
		\draw[thick] (-\l/2,\l/2) -- ++(\l,0) --++ (0,-\l);
	}{}
	\ifthenelse{\equal{#2}{left}}{
		\draw[thick] (\l/2,\l/2) -- ++(-\l,0) --++ (0,-\l);
	}{}
	\ifthenelse{\equal{#2}{rightf}}{
		\draw[thick] (-\l/2,-\l/2) -- ++(\l,0) --++ (0,\l);
	}{}
	\ifthenelse{\equal{#2}{leftf}}{
		\draw[thick] (\l/2,-\l/2) -- ++(-\l,0) --++ (0,\l);
	}{}
	
}

\newcommandx{\squaregate}[3][1=0,2=0,3=white]
{
	\begin{scope}[shift={(#1,#2)},rounded corners= 2pt]
		\draw[thick,fill=#3] (-.13,-.13) rectangle (.15,.15);
	\end{scope}
}

\newcommandx{\roundgate}[6][1=0,2=0,3=1,4=topright,5=white,6=-1]{
	\pgfmathparse{#3}
	\let\l\pgfmathresult
	\begin{scope}[shift={(#1,#2)}]
		\gatecross[\l]
			\pgfmathparse{\l/2.0}
		\let\s\pgfmathresult
		\gatesqu[\s][#5]
		\pgfmathparse{\l*0.15}
		\let\m\pgfmathresult
	\ifthenelse{\equal{#6}{-1}}{		\gatemark[\m][#4]
	}{	\node at ({0},{0}) {\scalebox{1}{$#6$}};}
\end{scope}
}

\newcommandx{\wcirc}[2]{\begin{scope}
		\draw[fill=white] (#1,#2) circle (0.15);	\end{scope}} 
\newcommandx{\wcircc}[2]{\begin{scope}
		\draw[fill=white] (#1,#2) circle (0.13);	\end{scope}} 
\newcommandx{\wsqr}[2]{\begin{scope}
		\draw[fill=white,shift={(#1,#2)}] (-.13,.13) rectangle (.13,-.13);	\end{scope}} 
\newcommandx{\wsqrr}[2]{\begin{scope}
		\draw[fill=white,shift={(#1,#2)}] (-.11,.11) rectangle (.11,-.11);	\end{scope}}
\newcommandx{\bcirc}[2]{\begin{scope}
		\draw[fill=black] (#1,#2) circle (0.15);	\end{scope}} 
\newcommandx{\thetastate}[4][1=0,2=0,3=1,4=]{
	\pgfmathparse{#3/2}
	\let\l\pgfmathresult
	\pgfmathparse{\l*0.15}
	\let\m\pgfmathresult
	\begin{scope}[shift={(#1,#2)}]
		\draw[thick] (0,0)--(\l,\l);
		\draw[thick] (0,0)--(-\l,\l);
		\ifthenelse{\equal{#4}{}}{
			\draw[fill=white] (0,0) circle (0.15);
		}{
			\draw[thick, fill=#4] (0,0) circle (0.15);
		}
	\end{scope}
}

\newcommandx{\thetastateflipped}[4][1=0,2=0,3=1,4=]{
	\pgfmathparse{#3/2}
	\let\l\pgfmathresult
	\pgfmathparse{\l*0.15}
	\let\m\pgfmathresult
	\begin{scope}[shift={(#1,#2)}]
		\draw[thick] (0,0)--(\l,-\l);
		\draw[thick] (0,0)--(-\l,-\l);
		\ifthenelse{\equal{#4}{}}{
			\draw[fill=white] (0,0) circle (0.15);
		}{
			\draw[thick, fill=#4] (0,0) circle (0.15);
		}
	\end{scope}
}

\newcommandx{\vertgate}[5][1=0,2=0,3=4,4=orange,5=topright]
{
	\begin{scope}[shift={(#1,#2)}]
		\ifthenelse{\equal{#3}{1}}{
			\roundgate[0][0][1][#5][#4]
		}{
			\foreach \n[evaluate=\n as \y using {2*\n-2}] in {1,...,#3}{
				\roundgate[0][\y][1][#5][#4]
			}
		}
	\end{scope}
}

\newcommandx{\circgate}[4][1=0,2=0,3=white,4]
{

         \let\m\pgfmathresult
	\begin{scope}[shift={(#1,#2)}]
	\draw[thick,fill=#3] (0,0) circle (.125);
	\gatemark[.125*\m][#4]
	\end{scope}
	
}

\newcommandx{\tsfmatV}[8][1=0,2=0,3=l,4=4,5=tr,6=init,7=orange,8=topright]{
	\begin{scope}[shift={(#1,#2)}]
		\ifthenelse{\equal{#3}{l}}{
			\pgfmathsetmacro{\flag}{0}
		}{
			\pgfmathsetmacro{\flag}{1}
		}
		
		\foreach \y[evaluate=\y as \x using {mod(\y+\flag,2)}] in {1,...,#4}{
			\roundgate[\x][\y][1][#8][#7]
		}
		\ifthenelse{\equal{#5}{tr}}{
			\foreach \y[evaluate=\y as \x using {mod(\y+\flag,2)}] in {#4}{
				\draw [fill=white] (\x-0.5,\y+0.5) circle (0.15);
				\draw [fill=white] (\x+0.5,\y+0.5) circle (0.15);
			}
		}{}
		\ifthenelse{\equal{#6}{init}}{
			\thetastate[\flag][0][1][#7]
		}{}
	\end{scope}
}

\newcommandx{\leftriangle}[5][1=0,2=0,3=4,4=orange,5=topright]{
	\begin{scope}[shift={(#1,#2)}]
		\pgfmathsetmacro{\t}{#3}
		\pgfmathsetmacro{\steps}{ceil(\t/2)}
		\foreach \i[evaluate=\i as \x using -\t+2*\i-1, evaluate=\i as \ylim using \t-2*\i+2] in {1,...,\steps}{
			\foreach \y[evaluate=\y as \thisx using {\x+\y-1}] in {1,...,\ylim}{
				\roundgate[\thisx][\y][1][#5][#4]
			}
		}
	\end{scope}
}

\newcommandx{\rightriangle}[5][1=0,2=0,3=4,4=orange,5=topright]{
	\begin{scope}[shift={(#1,#2)}]
		\pgfmathsetmacro{\t}{#3}
		\pgfmathsetmacro{\steps}{ceil(\t/2)}
		\foreach \i[evaluate=\i as \x using -\t+2*\i-1, evaluate=\i as \ylim using \t-2*\i+2] in {1,...,\steps}{
			\foreach \y[evaluate=\y as \thisx using {-\x-\y+1}] in {1,...,\ylim}{
				\roundgate[\thisx][\y][1][#5][#4]
			}
		}
	\end{scope}
}

\newcommandx{\eigenVL}[8][1=0,2=0,3=l,4=5,5=tr,6=init,7=orange,8=topright]{
	\begin{scope}[shift={(#1,#2)}]
		\pgfmathsetmacro{\t}{#4}
		\leftriangle[0][0][\t][#7][#8]
		
		\ifthenelse{\equal{#6}{init}}{
			\drawinitstate[0][0][l][\t][#7]
		}{}
		
		\ifthenelse{\equal{#5}{tr}}{
			\draw[fill=white] \foreach \x in {0,...,\t} {(\x-0.5-\t,0.5+\x) circle (0.15)};
			\ifthenelse{\equal{#3}{r}}{
				\draw[fill=white] (0.5,\t+0.5) circle (0.15);
			}{}
		}{}
		\ifthenelse{\equal{#5}{parttr}}{
			\draw[fill=white] \foreach \x in {0,...,\t} {(\x-0.5-\t,0.5+\x) circle (0.15)};
		}{}
	\end{scope}
}

\newcommandx{\eigenVR}[8][1=0,2=0,3=l,4=5,5=tr,6=init,7=orange,8=topright]{
	\begin{scope}[shift={(#1,#2)}]
		\pgfmathsetmacro{\t}{#4}
		\rightriangle[0][0][\t][#7][#8]
		
		\ifthenelse{\equal{#6}{init}}{
			\drawinitstate[0][0][r][\t][#7]
		}{}
		
		\ifthenelse{\equal{#5}{tr}}{
			\draw[fill=white] \foreach \x in {0,...,\t}{(-\x+0.5+\t,0.5+\x) circle (0.15)};
			\ifthenelse{\equal{#3}{l}}{
				\draw[fill=white] (-0.5,\t+0.5) circle (0.15);
			}{}
		}{}
	\end{scope}
}

\newcommandx{\tra}[2][1]{\underset{#1}{\text{tr}}\left[#2\right]}

\newcommandx{\tsfmatDgate}[7][1=0,2=0,3=l,4=4,5=tr,6=orange,7=topright]
{
	\begin{scope}[shift={(#1,#2)}]
		\ifthenelse{\equal{#3}{l}}{
			\pgfmathsetmacro{\flag}{-1}
		}{
			\pgfmathsetmacro{\flag}{1}
		}
		\pgfmathsetmacro{\t}{#4}
		\foreach \i[evaluate=\i as \x using {\flag*\i}, evaluate=\i as \y using \i] in {1,...,\t}{
			\roundgate[\x][\y][1][#7][#6]
		}
		
		\ifthenelse{\equal{#5}{tr}}{
			\foreach \i[evaluate=\i as \x using {\flag*\i}, evaluate=\i as \y using \i] in {\t}{
				\draw [fill=white] (\x-0.5,\y+0.5) circle (0.15);
				\draw [fill=white] (\x+0.5,\y+0.5) circle (0.15);
			}  
		}{}
	\end{scope}
	
}

\newcommandx{\tsfmatD}[8][1=0,2=0,3=l,4=4,5=tr,6=init,7=orange,8=topright]{
	\begin{scope}[shift={(#1,#2)}]
		\ifthenelse{\equal{#6}{init}}{
			\thetastate[0][0][1][#7]
		}{}
		
		\ifthenelse{\equal{#3}{l}}{
			\pgfmathsetmacro{\flag}{-1}
		}{
			\pgfmathsetmacro{\flag}{1}
		}
		
		\pgfmathsetmacro{\t}{#4}
		\foreach \i[evaluate=\i as \x using {\flag*\i}, evaluate=\i as \y using \i] in {1,...,\t}{
			\roundgate[\x][\y][1][#8][#7]
		}
		
		\ifthenelse{\equal{#5}{tr}}{
			\foreach \i[evaluate=\i as \x using {\flag*\i}, evaluate=\i as \y using \i] in {\t}{
				\draw [fill=white] (\x-0.5,\y+0.5) circle (0.15);
				\draw [fill=white] (\x+0.5,\y+0.5) circle (0.15);
			}  
		}
		\ifthenelse{\equal{#5}{parttr}}{
			\foreach \i[evaluate=\i as \x using {\flag*\i}, evaluate=\i as \y using \i] in {\t}{
				\draw [fill=white] (\x+0.5*\flag,\y+0.5) circle (0.15);
			}  
		}
		{}
	\end{scope}
}

\newcommandx{\drawinitstate}[5][1=0,2=0,3=l,4=4,5=orange]{
	\pgfmathsetmacro{\t}{#4}
	\begin{scope}[shift={(#1,#2)}]
		\pgfmathsetmacro{\steps}{ceil((\t-1)/2)}
		\ifthenelse{\equal{#3}{l}}{
			\foreach \i[evaluate=\i as \x using -\t+2*\i] in {0,...,\steps}{
				\thetastate[\x][0][1][#5]
			}
		}{
			\foreach \i[evaluate=\i as \x using -\t+2*\i] in {0,...,\steps}{      
				\thetastate[-\x][0][1][#5]
			}
		}
	\end{scope}
}

\newcommandx{\drawinitstateflipped}[5][1=0,2=0,3=l,4=4,5=orange]{
	\pgfmathsetmacro{\t}{#4}
	\begin{scope}[shift={(#1,#2)}]
		\pgfmathsetmacro{\steps}{ceil((\t-1)/2)}
		\ifthenelse{\equal{#3}{l}}{
			\foreach \i[evaluate=\i as \x using -\t+2*\i] in {0,...,\steps}{
				\thetastateflipped[\x][0][1][#5]
			}
		}{
			\foreach \i[evaluate=\i as \x using -\t+2*\i] in {0,...,\steps}{      
				\thetastateflipped[-\x][0][1][#5]
			}
		}
	\end{scope}
}

\newcommandx{\eigenDL}[6][1=0,2=0,3=l,4=4,5=orange,6=topright]{
	\begin{scope}[shift={(#1,#2)}]
		\pgfmathsetmacro{\t}{#4}
		\ifthenelse{\equal{#3}{l}}{
			\eigenVL[0][0][l][\t][tr][init][#5][#6]
			\pgfmathsetmacro{\t}{#4-1}
			\rightriangle[1][0][\t][#5][#6]
			\drawinitstate[1][0][r][\t][#5]
		}{
			\begin{scope}[shift={(-0.5,0.5)}]
				\foreach \i[evaluate=\i as \x using \i, evaluate=\i as \y using \i] in {0,...,\t}{      
					\draw (\x,\y)--++(0.5,0);
					\draw[fill=white] (\x,\y) circle (0.15);
				}
			\end{scope}
		}
	\end{scope}
}

\newcommandx{\eigenDR}[6][1=0,2=0,3=l,4=4,5=orange,6=topright]{
	\begin{scope}[shift={(#1,#2)}]
		\pgfmathsetmacro{\t}{#4}
		\ifthenelse{\equal{#3}{r}}{
			\eigenVR[0][0][r][\t][tr][init][#5][#6]
			\pgfmathsetmacro{\t}{#4-1}
			\leftriangle[-1][0][\t][#5][#6]
			\drawinitstate[-1][0][l][\t][#5]
		}{
			\begin{scope}[shift={(0.5,0.5)}]
				\foreach \i[evaluate=\i as \x using \i, evaluate=\i as \y using \t-\i] in {0,...,\t}{      
					\draw (\x,\y)--++(0.5,0);
					\draw[fill=white] (\x+0.5,\y) circle (0.15);
				}
			\end{scope}
		}
	\end{scope}
}

\newcommandx{\idonpurity}[2][1=0,2=0]
{
	\begin{scope}[shift={(#1,#2)}]
		\draw[thick] (-0.5,0)--++(-0.1,0.1)--++(0,0.2)--++(0.1,-0.1);
		\draw[thick] (-0.5,0.4)--++(-0.1,0.1)--++(0,0.2)--++(0.1,-0.1);
		\draw[thick] (0.5,0)--++(0.1,0.1)--++(0,0.2)--++(-0.1,-0.1);
		\draw[thick] (0.5,0.4)--++(0.1,0.1)--++(0,0.2)--++(-0.1,-0.1);
	\end{scope}
}

\newcommandx{\swaponpurity}[2][1=0,2=0]
{
	\begin{scope}[shift={(#1,#2)}]
		\draw[thick] (-0.5,0)--++(-0.2,0.2)--++(0,0.6)--++(0.2,-0.2);
		\draw[thick] (-0.5,0.2)--++(-0.075,0.075)--++(0,0.2)--++(0.075,-0.075);
		\draw[thick] (+0.5,0)--++(+0.2,0.2)--++(0,0.6)--++(-0.2,-0.2);
		\draw[thick] (+0.5,0.2)--++(+0.075,0.075)--++(0,0.2)--++(-0.075,-0.075);
	\end{scope}
}

\newcommandx{\hook}[4][1=0,2=0,3=t,4=l]{
	\begin{scope}[shift={(#1,#2)}]
		\ifthenelse{\equal{#3}{t}}{
			\ifthenelse{\equal{#4}{l}}{\draw[thick] (0.5,-0.5) arc (45:-90:0.15);}{\draw[thick] (0.5,-0.5) arc (45:270:0.15);}
		}{\ifthenelse{\equal{#4}{l}}{\draw[ thick] (0.5,-0.5) arc (-45:90:0.15);}{\draw[ thick] (0.5,-0.5) arc (315:90:0.15);}
		}
	\end{scope}
}

\newcommandx{\hhook}[4][1=0,2=0,3=t,4=l]{
	\begin{scope}[shift={(#1,#2)}]
		\ifthenelse{\equal{#3}{t}}{
			\ifthenelse{\equal{#4}{l}}{\draw[thick] (0.5,-0.5) arc (-45:175:0.15);}{\draw[thick] (0.5,-0.5) arc (225:0:0.15);}
		}{\ifthenelse{\equal{#4}{l}}{\draw[ thick] (0.5,-0.5) arc (-45:180:-0.15);}{\draw[ thick] (0.5,-0.5) arc (45:-180:0.15);}
		}
	\end{scope}
}

\newcommandx{\Pproj}[3][3=$P_\Lambda$]{
\begin{scope}[shift={(#1-.5,#2-1)}]
\draw[thick,fill=white] (0,0)rectangle (1,2);
\draw[thick] (0,1.5)--(-.5,1.5);
\draw[thick] (1,1.5)--(1.5,1.5);
\draw[thick] (0,.5)--(-.5,.5);
\draw[thick] (1,.5)--(1.5,.5);
\node[scale=2] at (.5,1) {#3};
\end{scope}}

\definecolor{FcolU}{rgb}{0.71,0.78,0.91}
\definecolor{colLines}{rgb}{0.31,0.31,0.31}
\definecolor{colVMPSLines}{rgb}{0.11,0.11,0.11}
\definecolor{IcolUc}{rgb}{0.71,0.41,0.42}
\definecolor{IcolU}{rgb}{0.71,0.8,0.76}
\definecolor{IcolVMPSc}{rgb}{0.73,0.69,0.7}
\definecolor{IcolVMPS}{rgb}{0.81,0.77,0.78}
\definecolor{colObs}{rgb}{1.,1.,1.}

\newcommandx{\eightlegs}[2][1=0,2=0]{
	\begin{scope}[shift={(#1,#2)}]
		\foreach \x in {1,...,8}{
			\draw (\x, 0)--++(0,0.25);
			\draw[fill] (\x,0) circle (0.05);
		}
		\foreach \x in {1,3}{
			\pgfmathsetmacro\result{2*\x-1} 
			\node () at (\result,-0.5) {$i_{\x}$};
			\pgfmathsetmacro\result{2*\x}
			\node () at (\result,-0.5) {$j_{\x}$};	
		}
		\foreach \x in {2,4}{
			\pgfmathsetmacro\result{2*\x} 
			\node () at (\result,-0.5) {$i_{\x}$};
			\pgfmathsetmacro\result{2*\x-1}
			\node () at (\result,-0.5) {$j_{\x}$};	
		}
	\end{scope}
}

\newcommandx{\MPSinitialstate}[6][1=0,2=0,3=orange,4=topright,5=-1,6=1]{
\begin{scope}[shift={(#1,#2)},rounded corners=1.5pt]
	\draw[black,thick,fill=#3] 
	(-0.25*#6,.25*#6)--++(.5*#6,0)--++(0,-.3*#6)--++(-.5*#6,0)--cycle;
	\draw[thick] (-.25*#6,.25*#6)--++(-.25*#6,.25*#6);
	\draw[thick] (.25*#6,.25*#6)--++(.25*#6,.25*#6);
	\draw[very thick] (-.5*#6,.-.05*#6)--++(1*#6,0);
\ifthenelse{\equal{#5}{-1}}{
	\ifthenelse{\equal{#4}{topright}}{\draw[thick,rounded corners=0.3]
	(-.1*#6,.15*#6)--++(.2*#6,0)--++(0,-0.1*#6); }{}
	\ifthenelse{\equal{#4}{topleft}}{\draw[thick,rounded corners=0.3]
	(.1*#6,.15*#6)--++(-.2*#6,0)--++(0,-0.1*#6); }{}
	\ifthenelse{\equal{#4}{bottomleft}}{\draw[thick,rounded corners=0.3]
	(.1*#6,-.15*#6)--++(-.2*#6,0)--++(0,0.1*#6); }{}	\ifthenelse{\equal{#4}{bottomright}}{\draw[thick,rounded corners=0.3]
	(-.1*#6,-.15*#6)--++(.2*#6,0)--++(0,0.1*#6); }{}}{			\node at ({0},{0.085*#6}) {\scalebox{1.}{{$#5$}}};}
\end{scope}
}

\newcommandx{\Cmatrix}[6][1=0,2=0,3=2,4=orange,5=,6=topright]{
	\pgfmathsetmacro\result{#3-1} 
	\begin{scope}[shift={(#1,#2)}]
		\foreach \i in {0,...,\result}
		{\foreach \j in {0,...,\i}
			{\roundgate[\i+\j][\i-\j][1][#6][#4]}
		}
		\ifthenelse{\equal{#5}{init}}{
			\foreach \i in {0,...,#3}
			{
				\MPSinitialstate[-1+2*\i][-1][#4]
			}
		}{}
	\end{scope}
}

\newcommand{\msqr}{\fineq[-0.6ex][0.8][1]{\sqrstate[0][0][]}}

\newcommand{\mcirc}{\fineq[-0.6ex][0.8][1]{\cstate[0][0][]}}
\newcommand{\mtrig}{\fineq[-0.6ex][0.8][1]{\trigstate[0][0][]}}
\newcommand{\mcross}{\fineq[-0.6ex][0.8][1]{\crossstate[0][0][]}}
\renewcommand{\bcirc}{\fineq[-0.6ex][0.8][1]{\cstate[0][0][][black]}}

\newcommandx{\cstate}[4][1=0,2=0,3= ,4=white]{
	\begin{scope}[shift={(0,0)}]
				\draw[fill=#4,thick] (#1,#2) circle (0.13);
				\node[scale=1.1] at (#1,#2) {$#3$};
\end{scope}
}
\newcommandx{\sqrstate}[4][1=0,2=0,3= ,4=white]{
	\begin{scope}[shift={(#1,#2)}]
		\draw[thick,fill=#4] (-0.12,-0.12) rectangle (0.12,0.12) ;
		\node[scale=1.1] at (0,0) {$#3$};
	\end{scope}
}
\newcommandx{\cstatep}[4][1=0,2=0,3= ,4=white,]{
	\begin{scope}[shift={(0,0)}]
				\draw[fill=#4,thick] (#1,#2) circle (0.15);
				\node[scale=1.1] at (#1,#2) {\scalebox{.85}{$#3$}};
\end{scope}
}
\newcommandx{\trigstatel}[4][1=0,2=0,3= ,4=white]{
	\begin{scope}[shift={(#1,#2)}]
	\begin{scope}[rotate around={-15:(0,0)}]
		\draw[thick,fill=#4] (-0.12,-0.07) -- (0.12,-0.07) -- (0,0.13) --cycle ;
		\node[scale=1.1] at (0,0) {$#3$};
	\end{scope}
	\end{scope}
}

\newcommandx{\trigstater}[4][1=0,2=0,3= ,4=white]{
	\begin{scope}[shift={(#1,#2)}]
	\begin{scope}[rotate around={15:(0,0)}]
		\draw[thick,fill=#4] (-0.12,-0.07) -- (0.12,-0.07) -- (0,0.13) --cycle ;
		\node[scale=1.1] at (0,0) {$#3$};
	\end{scope}
	\end{scope}
}

\newcommandx{\trigstate}[4][1=0,2=0,3= ,4=white]{
	\begin{scope}[shift={(#1,#2)}]
		\draw[thick,fill=#4] (-0.12,-0.07) -- (0.12,-0.07) -- (0,0.13) --cycle ;
		\node[scale=1.1] at (0,0) {$#3$};
	\end{scope}
}

\newcommandx{\crossstate}[4][1=0,2=0,3= ,4=white]{
	\begin{scope}[shift={(#1,#2)}]
		\draw[thick] (-0.12,-0.12) -- (0.12,0.12);
		\draw[thick] (0.12,-0.12) -- (-0.12,0.12);
		\node[scale=1.1] at (0,0) {$#3$};
	\end{scope}
}

\newcommandx{\pairproduct}[2][1=0,2=0]
{\begin{scope}[shift={(#1 ,#2)}]
\draw[thick] (-.5,.5) arc(-135:-45:1/1.414);
\sqrstate[0][.5-.1414][][black]	
\end{scope}
}
\newcommandx{\bellpair}[2][1=0,2=0]
{\begin{scope}[shift={(#1 ,#2)}]
		\draw[thick] (-.5,.5) arc(-135:-45:1/1.414);
	\end{scope}
}
\newcommandx{\charge}[3][1=0,2=0,3=black]
{
	\ifthenelse{\equal{#3}{blue}}{\def \chargecolor{grey6}
	}{
	\ifthenelse{\equal{#3}{red}}{\def \chargecolor{grey6}}{\def \chargecolor {#3}}}
\begin{scope}[shift={(#1 ,#2)}]
	\draw[ fill=\chargecolor] circle (0.13);        
\end{scope}
}
\newcommandx{\trianglediag}[7][1=0,2=0,3=1,4=orange,5= ,6=-1,7=topright]
{\begin{scope}[shift={(#1 ,#2)}]
	\foreach \i in {0,...,#3}
	{	\foreach \j in {0,...,\i}
		{	\roundgate[-\j+2*\i][\j][1][topright][#4][#6]
		}
	}
	\foreach \i in {-1,...,#3}
	{\ifthenelse{\equal{#5}{bellpair}}{\bellpair[\i*2+1][-1]}{\ifthenelse{\equal{#5}{pairproduct}}{\pairproduct[\i*2+1][-1]}{\MPSinitialstate[\i*2+1][-1][#4][#7][#6]}}}
\end{scope}
}
\newcommandx{\projectorleg}[4][1=0,2=0,3=R,4=left]
{
	\begin{scope}[shift={(#1 ,#2)}]
		{\ifthenelse{\equal{#3}{R}}{		\draw[thick] (-.25,-.25)--++(.5,.5);}{\ifthenelse{\equal{#3}{L}}{		\draw[thick] (-.25,.25)--++(.5,-.5);}{\draw[thick] (-.25,0)--++(.5,0);}}
		}
	\draw[thick, fill=white] circle (0.13);
	\ifthenelse{\equal{#4}{right}}{
		\draw[thick] (.0,.07)--++(.07,0)--++(0,-.07);
		\node[scale=0.5] at (0,-.02) {$\alpha$};
	}{}
	\ifthenelse{\equal{#4}{left}}{	
	\draw[thick] (0,.07)--++(-.07,0)--++(0,-.07);
	\node[scale=0.5] at (0,-.03) {$\beta$};
	}{}
	\end{scope}
}
\usetikzlibrary{patterns,decorations.pathreplacing}

\usepackage{hyperref}
\usepackage{graphicx}
\usepackage{physics}
\usepackage{amsmath}
\usepackage{amsfonts}
\usepackage{amssymb}
\usepackage{doi}
\usepackage{apacite}
\usepackage[british]{babel}
\usepackage{pgfplots}
\usepackage{pgfplotstable}
\usepackage{color,ifthen,amsthm,amsmath,amsxtra,amsfonts,fancyhdr}
\usepackage{dsfont,graphicx,bm,tikz,scalerel,wasysym,bbm,amsthm,empheq,amssymb,changepage}
\newcommand{\be}{\begin{equation}}
\newcommand{\ee}{\end{equation}}

\newtheorem*{theorem}{Theorem}
\newtheorem{lemma}{Lemma}
\newtheorem{definition}{Definition}

\makeatletter
\renewcommand{\@chapapp}{}

\makeatother

\title{\Huge \textbf{Non-Equilibrium Quantum Many-Body Physics with Quantum Circuits}  
}
\author{\textsc{Bruno Bertini}\thanks{\url{b.bertini@bham.ac.uk}}}

\begin{document}

\frontmatter
\maketitle


\tableofcontents

\mainmatter

\chapter*{Preface}

These are the notes for the 4.5-hour course with the same title that I delivered in August 2025 at the Les Houches summer school \href{https://www.lptms.universite-paris-saclay.fr/leshouches2025/}{``Exact Solvability and Quantum Information''}. In these notes I pedagogically introduce the setting of brickwork quantum circuits and show that it provides a useful framework to study non-equilibrium quantum many-body dynamics in the presence of local interactions. I first show that brickwork quantum circuits evolve quantum correlations in a way that is fundamentally similar to local Hamiltonians, and then present examples of brickwork quantum circuits where, surprisingly, one can compute exactly several relevant dynamical and spectral properties in the presence of non-trivial interactions.\\ 

\noindent If you find an error please report to \url{b.bertini@bham.ac.uk}. 

\section*{Recommended material}

Most of the material discussed in this notes is reviewed in the recent article
\begin{itemize}
\item ``Exactly solvable many-body dynamics from space-time duality''~\cite{bertini2025exactly}
\end{itemize}

Other relevant reviews on quantum circuits and their use in quantum many-body physics --- more focussed on the case of random circuits --- are 
\begin{itemize}
\item ``Entanglement Dynamics in Hybrid Quantum Circuits''~\cite{potter2022entanglement}.  
\item ``Random Quantum Circuits''~\cite{fisher2023random}. 
\end{itemize}
For a general review on quantum quenches I suggest 
\begin{itemize}
\item ``Quench dynamics and relaxation in isolated integrable quantum spin chains''~\cite{essler2016quench}. 
\end{itemize}

Other references, relevant to the specific topics I will touch upon, will be given in the course of the discussion.

\section*{Acknowledgements} 

I would like to thank all the great researchers with whom I had the pleasure to collaborate on the topics discussed in these notes. A particular mention goes to Pavel Kos, Lorenzo Piroli, Katja Klobas, Alessandro Foligno, and especially Toma\v z Prosen, from whom I learnt most of the material presented here. I am also very grateful to the organisers of the school ``Exact Solvability and Quantum Information'' for inviting me to lecture on these exciting topics. Finally, I would like to thank the Royal Society for supporting my research through a University Research Fellowship (No.\ 201101).

\chapter{}

\section{Introduction}


When studying non-equilibrium quantum physics our goal is to understand\slash characterise\slash describe the dynamics of quantum many-body systems. One might ask what is there to understand since we already know the laws of quantum mechanics. 
The answer, of course, is a resounding ``everything!''. The quantum mechanical description of a many-body system is so complicated and inefficient that it does not offer any conceptual or practical understanding: we need an ``emergent'' (in the sense of~\cite{anderson1972more}) macroscopic theory to really make sense of the physics of quantum systems at mascroscopic scales. For systems at equilibrium an example of such an emergent theory is classical thermodynamics --- which we can derive from the microscopic laws via equilibrium statistical mechanics. For systems out of equilibrium, however, such an emergent theory is currently unknown. 

The goal of these lectures is to discuss a setting --- brickwork quantum circuits ---  where quantum many-body dynamics can be studied efficiently and its main properties can be thoroughly characterised. In analogy with (few body) classical systems, we will also try to define whether the dynamics of a quantum many-body system are ``integrable'' or ``chaotic''.

\subsection{Concrete Setting: Quantum Quench}

A concrete protocol that we can use to generate quantum many-body dynamics is the so-called ``Quantum Quench''~\cite{calabrese2006time}. It can be implemented as follows: 
\begin{itemize}
\item[(i)] Prepare the system in the ground state $\ket{\Psi_0}$;
\item[(ii)] Suddenly change (quench) a parameter in the Hamiltonian (e.g.\ switch on an interaction or an external field) $H_0 \mapsto H$;
\item[(iii)] Let the system evolve without any other external intervention, i.e., $\ket{\Psi_t}=e^{-i H t} \ket{\Psi_0}$. 
\end{itemize}
In fact, one can generally simplify the protocol by just considering a state $\ket{\Psi_0}$ and evolving it with a Hamiltonian $H$ of which it is not an eigenstate. In general, however, to observe non-trivial phenomenology one has to require that $\ket{\Psi_0}$ has non-zero overlap with a number of eigenstates of the Hamiltonian that is that is exponentially large in the volume. Moreover we will generally consider states $\ket{\Psi_0}$ with low entanglement. These properties are automatically verified in a Quantum Quench (i)--(iii) where both $H_0$ and $H$ feature local (sufficiently short-ranged) interactions. 

The questions that one can ask in this situation can be grouped in two main classes: those concerning the \emph{finite-time behaviour} of the system or those concerning the \emph{long-time behaviour}. Examples of questions in the first class are\\

Q: Are there some ``universal features'' that are common in the evolution of different systems?' If so what controls them? Are there some new out-of-equilibrium phenomena occurring at finite times?\\ 

\noindent Instead, examples of questions in the second class are:\\
 
Q: Is there a form of relaxation/stationary behaviour emerging despite the fact that the dynamics is purely unitary? If so what are its mechanisms? Is it possible to develop a statistical mechanical description of the putative steady state? \\

All these questions are \emph{extremely difficult} to answer because we have no general theory to describe quantum many-body systems in the presence of non-trivial interactions. For instance, \emph{no low-energy description is applicable} and \emph{perturbative approaches break down at large enough times}. Also our computational approaches are severely limited: exact approaches scale exponentially with the number of particles in the system (see, e.g.,~\cite{feynman2018simulating}'s take on the problem), while even our best approaches to treat numerically large one-dimensional systems (based on \emph{tensor networks}, cf.~Frank Verstraete's lectures~\cite{cuiper2026leshouches}) are hampered by the rapid growth of entanglement (almost unavoidable after a quantum quench). 

So, what can we do? We need to come up with new ideas. In the next section I will discuss a setting that, it turns out, allows us to make interesting progress.  

\section{Brickwork quantum circuits}

Let us begin by focussing on a one dimensional setting (we will see that the discussion can be directly generalised to higher dimensions). Specifically, we consider $2L$ ``qudits'', i.e., quantum systems with $d\geq 2$ states (qubits for $d=2$), arranged along a one-dimensional line and evolved by discrete applications of the unitary operator  
\be
\mathbb{U} =\mathbb{U}_e\mathbb{U}_o,\qquad \mathbb{U}_e=\bigotimes_{x=0}^{L-1} U^{(x)}_{x,x+1/2},\,\,\mathbb{U}_o=\bigotimes_{x=1}^{L} U^{(x-1/2)}_{x-1/2,x}\,.
\label{eq:floquetoperator}
\ee
Here we labelled the chain's sites using half integers (in $\mathbb Z_{2L}/2$) and took periodic boundary conditions by setting $L\equiv 0$. The operator $U^{(x)}_{y,z}\in U(d^{2L})$ acts as $U^{(x)}\in U(d^2)$ on the qudits at position $y$ and $z$ and as the identity on all the others. We see that $\mathbb{U}_{e/o}$ couples nearest neighbours (either integer sites to half-odd-integer ones (e) or half-odd-integer to integer ones (o)) with a coupling that can depend on the position ($U^{(x)}$ depends explicitly on $x$). 

More explicitly, by  ``evolved by discrete applications of $\mathbb{U}$'' I mean that $\ket{\Psi_t}$ --- the state of the system after $t\in \mathbb N$ steps of evolution --- is given by 
\be
\ket{\Psi_t} = \mathbb{U} \ket{\Psi_{t-1}}\,. 
\ee

Using the language of quantum computation, a system evolving according to a $\mathbb{U}$ of the form in Eq.~\eqref{eq:floquetoperator} is often referred to as \emph{brickwork quantum circuit} (BQC) (the generic term ``quantum circuit'' is used for any system of qudits with discrete time evolution and the specific way according to which the gates are applied is called ``architecture'' --- in this language we can say that Eq.~\eqref{eq:floquetoperator} has a brickwork architecture). In the same vein one can refer to $U^{(x)}$ as the ``local gate'', to $2L$ as the ``width'' of the circuit, and $2t$ as its ``depth''. On the other hand, one can think of $\mathbb{U}$ as the operator implementing one period of evolution --- Floquet operator --- in a periodically driven system in continuous time. In general one can consider BQCs where $U^{(x)}$ has a larger support (i.e., acts on more than two qubits) but for these lectures we will exclusively focus on nearest-neighbour cases.

Before continuing with the discussion of the basic physical properties of this discrete-time evolution, it is useful to represent it in a more intuitive, graphical form.

\subsection{Diagrammatic representation}	
\label{sec:diagrams}

States and observables in quantum circuits can be represented diagrammatically using a graphical notation very similar to the one used in tensor network theory (again cf.~Frank's lectures). In particular, we have three main rules:\\

\noindent (a) Operators are represented as shapes with legs indicating the qudits they act upon. For instance 
\begin{align}
&U^{(x)} & &\longmapsto & &\fineq[-0.8ex][1][1]{
		\tsfmatV[0][-0.5][r][1][][][red6]
}\\
&(U^{(x)})^\dag & &\longmapsto & &\fineq[-0.8ex][1][1]{
		\tsfmatV[0][-0.5][r][1][][][blue6][bottomright]
}\\
&\mathds{1} & &\longmapsto & &\fineq[-0.8ex][1][1]{
                 \node at (-0.35,0) {};
		\draw[ thick] (0,.5) -- (0,-.5);
}\\
& \ket{\psi}\in \mathbb C^d  & &\longmapsto & &\fineq[-0.8ex][1][1]{
                 \node at (-0.35,0) {};
		\draw[ thick] (0,.5) -- (0,0);
		\draw[thick, fill = black]  (0,0) circle (2pt);
}
\end{align}
When representing different operators of the same kind at different positions, e.g., $U^{(y)}\neq U^{(x)}$, I will use different shades of the same colour. 

\noindent (b) Joining legs corresponds to summing over indices. For instance  
\be
\fineq[-0.8ex][1][1]{
\tsfmatV[0][1.25][r][1][][][red6][topright]
\tsfmatV[0][0][r][1][][][blue6][bottomright]
\draw[ thick] (-0.5,1.5) -- (-0.5,1.75);
\draw[ thick] (0.5,1.5) -- (0.5,1.75);}
 = 
\fineq[-0.8ex][1][1]{
\tsfmatV[0][1.25][r][1][][][blue6][bottomright]
\tsfmatV[0][0][r][1][][][red6][topright]
\draw[ thick] (-0.5,1.5) -- (-0.5,1.75);
\draw[ thick] (0.5,1.5) -- (0.5,1.75);}=
\fineq[-0.8ex][1][1]{	
\draw[ thick] (-0.5,.5) -- (-0.5,2.75);
\draw[ thick] (0.5,.5) -- (0.5,2.75);}\,, 
\label{eq:unitarity}
\ee
represents the unitarity of  $U^{(x)}$. \\

\noindent (c) Time-evolution goes bottom to top.   
\be
\ket{\Psi_t} = \mathbb{U}^t\ket{\Psi_0}=\fineq[-0.8ex][.75][1]{	
		\foreach \j in {0,...,3}{
		\foreach \i in {0,...,3}
		{
	         \roundgate[2*\i][2*\j][1][topright][red6][-1]
	         \roundgate[2*\i+1][1+2*\j][1][topright][red6][-1]}}
	         \foreach \i in {0,...,3}{
	         \draw[thick, fill = black]  (2*\i-0.5,-0.5) circle (2pt);
	         \draw[thick, fill = black]  (2*\i+0.5,-0.5) circle (2pt);}
	         \draw [decorate, thick, decoration = {brace}]   (6.5,-.75)--++(-7,.0);
                  \node[scale=1.25] at (3,-1.25) {``width''};
                   \draw [decorate, thick, decoration = {brace}]   (8,7.5)--++(0,-8);
                  \node[scale=1.25] at (9.5,3.5) {``depth''};
	         }\!\!\!\!\!.
\label{eq:timeevolution_graphical}
\ee
Where we took $(L=4, t=3)$ and, for simplicity, we represented the case where $U^{(x)}=U$ for all $x$ and 
\be
\ket{\Psi_0} = \bigotimes_{j=1}^{2L} \ket{\psi}. 
\ee 

\subsection{Basic physical properties of BQC}

Using the diagrammatic representation introduced above we can easily see that: 
\begin{itemize}
\item[1.]  We see from Eq.~\eqref{eq:timeevolution_graphical} that the evolution in a BQC really looks as a quantum version of a boolean circuit, where the boxes --- local gates --- implement unitary transformations. Also, the local gates are applied in brickwork pattern, motivating its name. 
\item[2.]  Quantum circuits have a \emph{strict} causal light cone for the propagation of correlations. For instance, considering a local operator $\mathcal O_{L/2}$ and evolving it in time in the Heisenberg picture we have 
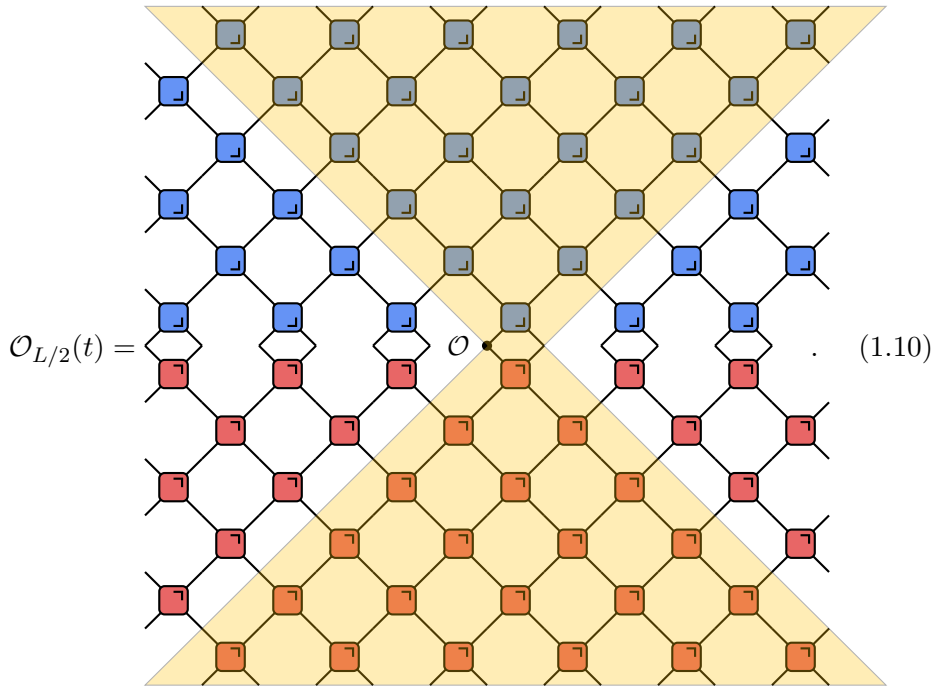
\be
\hspace{-1cm}\mathcal O_{L/2}(t) =\fineq[-0.8ex][.75][1]{
\foreach \j in {0,...,2}{
\foreach \i in {0,...,5}{
\roundgate[2*\i][2*\j][1][bottomright][blue6][-1]
\roundgate[2*\i+1][1+2*\j][1][bottomright][blue6][-1]}}
\foreach \j in {0,...,2}{
\foreach \i in {0,...,5}{
\roundgate[2*\i][-2*\j-1][1][topright][red6][-1]
\roundgate[2*\i+1][-1-2*\j-1][1][topright][red6][-1]}}
 \draw[thick, fill = black]  (5.5,-.5) circle (2pt);
 \draw[fill=yellow5, opacity=0.3] (5.5,-0.5) -- (-0.5,-6.5) -- (12.5,-6.5) -- (6.5,-0.5) -- (12.5, 5.5) -- (-0.5,5.5) -- (5.5,-0.5);
\node[scale=1.25] at (5,-0.5) {$\mathcal O$};
}\hspace{-1cm}. 
\label{eq:localoperatorQC}
\ee 
All gates out of the yellow light cones can be simplified using the unitarity conditions in Eq.~\eqref{eq:unitarity}. All BQC described by the Floquet operator in Eq.~\eqref{eq:floquetoperator} have a speed of light $v_{\rm max}=1$ in our units.  
\item[3.] Everything discussed above can be immediately generalised to dimension $D\geq 1$ by replacing 2d shapes with $(D+1)$d shapes. For example, squares $\mapsto$ hypercubes. 
\end{itemize}

\section{Relation between BQC and Hamiltonian systems}

The dynamics of BQC are closely related to those of Hamiltonian systems with local interactions in that both these classes of systems can propagate correlations at most \emph{ballistically}. To see this, consider a local Hamiltonian with nearest-neighbour interactions~\footnote{I am again considering the one-dimensional setting for simplicity but all statements directly carry over to any dimension.}
\be
\label{eq:ham}
H = \sum_{x\in \mathbb Z_{2L}/2} h_{x,x+1/2}, \qquad (L\equiv 0), 
\ee 
and also assume 
\be
\bar h \equiv \max_x \|h_{x,x+1/2}\| = C < \infty,
\ee
where $\|\cdot\|$ is the operator norm and $C$ does not depend on $L$ (this is the case, e.g., for systems of fermions or spins, not for bosons). In this case one characterise the spreading of correlations by means of the famous Lieb Robinson bound~\cite{lieb1972finite}. We can express it as follows 
\begin{theorem}[Lieb and Robinson (1972)]
\label{th:LR}
Let $\mathcal O_X$ be an operator acting non-trivially in the block $X$ and let $A$ and $B$ be two blocks at distance $\ell$. Then 
\be
\label{eq:LRbound}
\| [\mathcal O_A(t),\mathcal O_B(0)]\| \leq D \exp[-\frac{(\ell- v t)}{\xi}],
\ee
where $D$, $v$, and $\xi$ are system-dependent constants ($D$ also depends on $A$, $B$, and the operators).  
\end{theorem}

In words, the above result states that the commutator of the two operators $\mathcal O_A(t)$ and $\mathcal O_B(0)$ is exponentially suppressed for times smaller than $\ell/v$, with $\ell$ being the distance between the regions they act upon at time 0 and $v$ an intrinsic maximal velocity for the spreading of correlations in the system~\footnote{This statement has later been generalised both to systems interacting over a longer range and to unbounded $\bar h$ (see, e.g.,~\cite{kuwahara2021lieb} for a recent account of the possible extensions) but for the purposes of these lectures the original statement above is sufficient.}. This does indeed show that --- up to exponentially small corrections --- correlations propagate at most ballistically in the Hamiltonian system characterised by Eq.~\eqref{eq:ham}. To make more stringent the analogy with BQC, it is useful to make a restatement of the Lieb-Robinson bound due to Bravyi, Hastings, and Vestraete~\cite{bravyi2006lieb}. Namely, one can show that Eq.~\eqref{eq:LRbound} implies 
\be
\label{eq:bravyi}
\| \mathcal O_A(t)-\mathcal O^{(\ell)}_A(t) \| \leq D' \exp[-\frac{(\ell- v t)}{\xi}],
\ee   
where $D'$ is again a constant, we defined the ``truncated operator''  
\be
\mathcal O^{(\ell)}_A(t) = \frac{\tr_S[\mathcal O_A(t)] \otimes \mathds{1}_S}{\tr_S[\mathds{1}_S]}\,,
\ee
and $S$ is a region whose distance from $A$ is at least $\ell$, see Fig.~\ref{fig:regions}. 

\begin{figure}
\centering
    \begin{tikzpicture}[baseline={([yshift=-0.6ex]current bounding box.center)},scale=0.65]
     \draw[thick,black,dashed] (-8,0) -- (8,0);
     \draw[thick,black] (-7,0) -- (7,0);
      \draw[thick,black] (3.5,0) -- (3.5,1);
      \draw[thick,black] (-3.5,0) -- (-3.5,1);
      \draw[thick,black] (1.5,0) -- (1.5,1);
      \draw[thick,black] (-1.5,0) -- (-1.5,1);
     \node at (4.5,.5) {$S$};
     \node at (-4.5,.5) {$S$};
     \node at (0,.5) {$A$};
     \node[scale = 1] at (-2.5,1)  {$\ell$};
     \node[scale = 1] at (2.5,1)  {$\ell$};
     \draw[thick,black,<->] (-3.5,0.5) -- (-1.5,0.5);
     \draw[thick,black,<->] (3.5,0.5) -- (1.5,0.5);
     \draw[fill=gray, opacity = 0.5] (1.5,0) -- (-1.5,0) -- (-1.5,1) -- (1.5,1) -- cycle;
     \draw[fill=gray, opacity = 0.5] (3.5,0) -- (6.5,0) -- (6.5,1) -- (3.5,1) -- cycle;
     \draw[fill=gray, opacity = 0.5] (-3.5,0) -- (-6.5,0) -- (-6.5,1) -- (-3.5,1) -- cycle;
  \end{tikzpicture}
  \caption{Setting in the one dimensional case.}
  \label{fig:regions}
\end{figure}
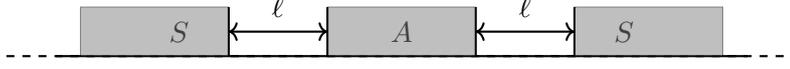

The meaning of Eq.~\eqref{eq:bravyi} is that at any time $t$, the time evolving operator $\mathcal O_A(t)$ can be approximated arbitrary well by an operator acting non trivially only within the effective light cone $[-v' t + a, b+v' t]$, where we took $v'>v$ and $A=[a,b]$. Specialising this statement to the case where $A$ is the single point $x=0$ this means that one can approximate $\mathcal O_0(t)$ with an operator acting non-trivially only in $[-v' t,v' t]$. This is very similar to Eq.~\eqref{eq:localoperatorQC}! The difference is that Eq.~\eqref{eq:localoperatorQC} is an exact statement while Eq.~\eqref{eq:bravyi} involves an (exponentially small in time) error, see Fig.~\ref{fig:LRvsBQC}

To make this intuition more quantitative I will now show how the time-evolution operator generated by the local Hamiltonian $H$ can be written in terms of an appropriate BQC. I will present two methods ordered by simplicity. 

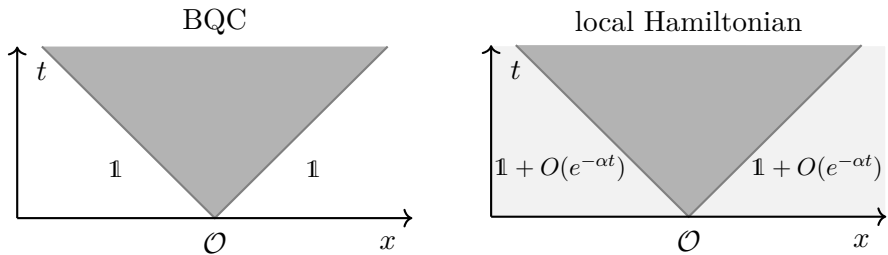
\begin{figure} [b]
    \begin{tikzpicture}[baseline={([yshift=-0.6ex]current bounding box.center)},scale=0.65]
     \draw[thick,black,->] (-4,0) -- (-4,3.5) node at (-3.5,3) {$t$};
     \draw[color=gray, fill = gray!60, thick, domain = -3.5: 3.5]    plot (\x,{abs(\x)});
     \draw[thick,black,->] (-4,0) -- (4,0)  node at (3.5,-.5) {$x$};]     
     \node at (0,-.5) {$\mathcal O$};
     \node[scale = 0.85] at (-2,1)  {$\mathds{1}$};
     \node[scale = 0.85] at (2,1)  {$\mathds{1}$};
     \node[scale = 1] at (0,4)  {BQC};
  \end{tikzpicture}
  \qquad 
  \begin{tikzpicture}[baseline={([yshift=-0.6ex]current bounding box.center)},scale=0.65]
     \draw[color=white, fill = gray!10, thick, domain = -3: 3]   (-4,0) --  (-4,3.5) --  (4,3.5) --  (4,0) -- cycle;
     \draw[thick,black,->] (-4,0) -- (-4,3.5) node at (-3.5,3) {$t$};
     \draw[color=gray, fill = gray!60, thick, domain = -3.5: 3.5]    plot (\x,{abs(\x)});
     \draw[thick,black,->] (-4,0) -- (4,0)  node at (3.5,-.5) {$x$};
     \node at (0,-.5) {$\mathcal O$};
       \node[scale = 0.85] at (-2.6,1)  {$\mathds{1}+O(e^{-\alpha t})$};
     \node[scale = 0.85] at (2.6,1)  {$\mathds{1}+O(e^{-\alpha t})$};
      \node[scale = 1] at (0,4)  {local Hamiltonian};
  \end{tikzpicture}
  \caption{Schematic depiction of the light cone of a local operator in a BQC (left) and local Hamiltonian system (right).}
  \label{fig:LRvsBQC}
\end{figure}

\subsection{Suzuki--Trotter decomposition}

The simplest way to write the time evolution operator of $H$ in terms of a BQC is the ``Suzuki--Trotter'' decomposition~\cite{trotter1959product, suzuki1990fractal}. To express it, we begin by with the following rewriting
\be
\label{eq:suzukitrotter}
\begin{aligned}
e^{-i t H} &= \left( e^{-i \frac{t}{n} \sum_{x=0}^{2L-1} h_{x/2,(x+1)/2}}\right)^n\\
& =  \left( e^{-i \frac{t}{n} \sum_{x=0}^{L-1} h_{x,x+1/2}} e^{-i \frac{t}{n} \sum_{x=1}^{L} h_{x-1/2,x}}\right)^n + \epsilon, 
\end{aligned}
\ee
where $n\in \mathbb N$ is an arbitrary natural number and in the second step we introduced 
\be
\!\!\epsilon \!=\! \left( e^{-i \frac{t}{n} \sum_{x=0}^{2L-1} h_{x/2,(x+1)/2}}\right)^n \!\!\!\!- \left( e^{-i \frac{t}{n} \sum_{x=0}^{L-1} h_{x,x+1/2}} e^{-i \frac{t}{n} \sum_{x=1}^{L} h_{x-1/2,x}}\right)^n\!\!\!\!. 
\ee
Note that the first term on the second line of Eq.~\eqref{eq:suzukitrotter} is precisely the $n$-th power of the BQC in Eq.~\eqref{eq:floquetoperator} written in terms of the local gates 
\be
\label{eq:localgatesST}
U^{(x)} =e^{- i \frac{t}{n} h_{x,x+1/2} },\quad x\in \mathbb Z_{2L}/2. 
\ee
The basic observation now is that the Baker--Campbell--Hausdorff formula guarantees that for large $n$ the operator norm of $\epsilon$ --- the ``error'' --- is $O(1/n)$. Therefore, the evolution generated by a local Hamiltonian can be approximated arbitrary well by a carefully chosen BQC, provided the latter is sufficiently deep, i.e., 
\be
\| e^{-i t H}  - \left(\mathbb{U}_e\mathbb{U}_o\right)^n \|= O\left(\frac{1}{n}\right)\,.
\ee
As shown by~\cite{suzuki1990fractal}, the approximation can be improved by replacing $\mathbb{U}_e\mathbb{U}_o$ in the above expression by a more complicated function of $\mathbb{U}_e$ and $\mathbb{U}_o$. For example, considering 
\be
\mathbb{U}_e\mathbb{U}_o \mapsto \mathbb{U}_e^{1/2}\mathbb{U}_o \mathbb{U}_e^{1/2}, 
\ee
gives $\|\epsilon\|=O(1/n^2)$. 

The problem of this approximation is that, for it to become accurate, one really needs a very deep BQC (made of gates that are infinitesimally close to the identity). More precisely, to achieve a fixed error $\|\epsilon\|$ one needs $n=O(L)$ even when $t=O(L^0)$. Also note that, in this construction, the ``trivial'' light cone of the BQC --- i.e.\ the light cone appearing as a consequence of unitarity as in Eq.~\eqref{eq:localoperatorQC} --- is not related to the Lieb--Robinson light cone of the Hamiltonian system. While the latter is fixed, the former becomes increasingly more ``open'' as $n$ increases, see Fig.~\ref{fig:trottercircuit}. This is because in the special BQC described by the gates in Eq.~\eqref{eq:localgatesST} correlations are effectively restricted within a much smaller light cone (with speed $O(1/n)$). This, and not the trivial one, is the light cone approaching the Lieb--Robinson light cone in the large $n$ limit.

 \begin{figure}
\begin{tikzpicture}[baseline={([yshift=-0.6ex]current bounding box.center)},scale=0.6]
\foreach \j in {0,...,1}{
\foreach \i in {0,...,4}{
\roundgate[2*\i][2*\j][1][topright][red6][-1]}}
\foreach \j in {0,...,0}{
\foreach \i in {0,...,3}{
\roundgate[2*\i+1][1+2*\j][1][topright][red6][-1]}}
	         \foreach \i in {0,...,4}{
	         \draw[thick, fill = black]  (2*\i-0.5,-0.5) circle (2pt);
	         \draw[thick, fill = black]  (2*\i+0.5,-0.5) circle (2pt);}
 \draw[fill=yellow5, opacity=0.3] (4.5,-.5) -- (4.5+3,0+3-.5) -- (4.5-4,0+3-.5) -- (4.5-1,-.5) -- cycle;
 \draw [decorate, thick, decoration = {brace}]   (9,2.5)--++(0,-3);
                  \node[scale=1.25] at (9.5,1.25) {$t$};
\end{tikzpicture} 
\,\,
\begin{tikzpicture}[baseline={([yshift=-0.6ex]current bounding box.center)},xscale=0.6,yscale=0.3]
\foreach \j in {0,...,2}{
\foreach \i in {0,...,4}{
\roundgate[2*\i][2*\j][1][topright][red6][-1]}}
\foreach \j in {0,...,2}{
\foreach \i in {0,...,3}{
\roundgate[2*\i+1][1+2*\j][1][topright][red6][-1]}}
	         \foreach \i in {0,...,4}{
	         \draw[thick, fill = black]  (2*\i-0.5,-0.5) circle (2pt);
	         \draw[thick, fill = black]  (2*\i+0.5,-0.5) circle (2pt);}
	         \draw[fill=yellow5, opacity=0.3] (4.5,-.5) -- (4.5+4,0+4-.5) -- (4.5+4,0+5-.5) -- (4.5-5,0+5-.5) -- (4.5-5,0+4-.5) -- (4.5-1,-.5) -- cycle;
	          \draw [decorate, thick, decoration = {brace}]   (9,5)--++(0,-6);
                  \node[scale=1.25] at (9.5,2.5) {$t$};
\end{tikzpicture}
\caption{The light cone opens up when increasing $n$ for fixed $t$.}
\label{fig:trottercircuit}
\end{figure}
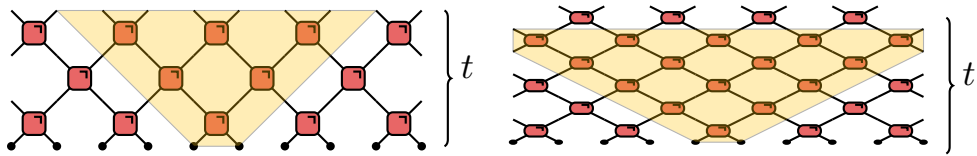

\subsection{Osborne Construction}

A more refined construction is the one proposed by~\cite{osborne2006efficient} for 1D systems (see also~\cite{haah2023quantum} for a generalisation to higher dimensions). In essence, the main idea is to proceed in three steps
\begin{itemize}
\item[a.] ``Block together'' $\Omega$ sites 
\be
\fineq[-0.8ex][.75][1]{
\foreach \i in {0,...,5}{ 
 \draw[thick, fill = black]  (\i,0) circle (1pt);}
\draw[->, thick] (6,0) -- (7,0); 
\draw[|->, thick] (6,-1.5) -- (7,-1.5); 
 \draw[thick, fill = black]  (8,0) circle (4pt);
  \draw [decorate, thick, decoration = {brace}]   (5,-.25)--++(-5,.0);
                  \node[scale=1.25] at (2.5,-.75) {$\Omega$};
\node[scale=1.35] at (2.5,-1.5) {$\mathbb C^d \otimes \cdots \otimes \mathbb C^d$};           
\node[scale=1.35] at (8.25,-1.5) {$\mathbb C^{d^\Omega}$};           
}
\ee
\item[b.] Define operators $W^{(\ell)}\in U(d^{2\Omega})$ acting as local gates on the blocked sites. 
\item[c.] Write the time evolution operator for a time $\tau$ as
\be
e^{- i \tau H} = \bigotimes_{x=0}^{L/\Omega-1} W^{(x)}_{x,x+1/2}\bigotimes_{x=1}^{L/\Omega} W^{(x-1/2)}_{x-1/2,x}+\epsilon. 
\ee
\end{itemize}
Using the Lieb--Robinson bound one can show that the error can be bounded by 
\be
\|\epsilon\| \leq \frac{a L}{\Omega} e^{b \tau -c \Omega}, \qquad a,b,c>0. 
\ee
This means that considering $\Omega = O(\log L)$ we can control the error for any fixed $\tau$. Therefore we can write 
\be
e^{- i t H} = \left(e^{- i \tau H}\right)^{t/\tau} =  \left(\bigotimes_{x=0}^{L/\Omega-1} W^{(x)}_{x,x+1/2}\bigotimes_{x=1}^{L/\Omega} W^{(x-1/2)}_{x-1/2,x}\right)^{t/\tau}+\epsilon,
\ee
where the error is bounded by 
\be
\|\epsilon\| \leq \exp\left(\frac{a t L}{\Omega \tau} e^{b \tau -c \Omega}\right)-1, \qquad a,b,c>0. 
\ee
Namely, it can be controlled for any fixed $\tau$ without the need of taking the limit $\tau\to0$. This means that in this case the local gates do not need to be close to the identity and the light cone of the quantum circuit approaches the Lieb-Robinson light cone.

\paragraph{General strategy of these lectures.} In this section we discussed how the dynamics of BQCs are fundamentally similar to those of local Hamiltonians in the way they spread correlations. We also showed two explicit constructions to obtain the latter from the former. In the rest of the lectures, however, we will study BQCs as independent quantum dynamical systems --- with general local gates in $U(d^2)$ --- without thinking about them as approximations of continuous time Hamiltonian dynamics.

\section{Dynamics in BQC: a first look}

Let us now begin to study the dynamics generated by BQCs --- we will look at a very simple example that will help clarifying what we mean by relaxation in systems evolving unitarily.

Consider the ``Floquet XX model'': the BQC described by the evolution operator in Eq.~\eqref{eq:floquetoperator} with $d=2$ and local gates   
\be
\label{eq:trotterXX}
U^{(x)} = e^{i \frac{\omega}{4} (XX+YY)},\qquad \forall x, 
\ee 
where I use $X, Y$, and $Z$ to denote the Pauli matrices, and initialise the circuit in the ``N\'eel state''  
\be
\label{eq:neel}
\ket{\Psi_0}=\bigotimes_{x=1}^{L}\ket{01},
\ee
with $\ket{0}$ and $\ket{1}$ the two eigenstates of $Z$ (corresponding respectively to eigenvalues $+1$ and $-1$). 

Let us now ask whether this BQC shows any form of relaxation. Consider first the case $L=1$ (two qubits) --- in this case we can immediately find the state at time $t$ to be
\be
\ket{\Psi_t}=\frac{e^{i \omega t }\ket{+}+e^{-i \omega t}\ket{-}}{\sqrt 2}, 
\ee
where we set $\ket{\pm}=(\ket{01}\pm\ket{10})/\sqrt 2$. This means that the expectation value of  any observable $\mathcal O$ can be written as 
\be
\!\!\!\expval{\mathcal O}{\Psi_t} \!=\! \frac{1}{2}\! \left\{ \!\!\expval{\mathcal O}{+}\!+\!\! \expval{\mathcal O}{-}\!+\!e^{i 2 \omega t} \mel{-}{\mathcal O}{+} \!+\! e^{- i 2 \omega t}\mel{+}{\mathcal O}{-}\!\right\}\!.
\ee
Therefore, either the expectation value is time independent or it oscillates indefinitely: no form of relaxation can be observed. The oscillations are due to the off-diagonal terms. 

Let us now move to consider $L>1$. Crucially, we can still find a closed-form expression for the expectation values of relevant observables. This is because for our special choice of local gates (cf.~Eq.~\eqref{eq:trotterXX}) the BQC evolution operator can be mapped into a quadratic form of free fermions and the N\'eel state in Eq.~\eqref{eq:neel} is Gaussian with respect to them (see, e.g., Sec.~D of the supplemental material of~\cite{vernier2023integrable} for details). For instance, after a simple (but lengthy) calculation we obtain 
\be
\label{eq:expvalZL}
\!\!\!\expval{Z_x}{\Psi_t} 
\!=\! \frac{(-1)^x}{2 L} \!\sum_{n=0}^{L} \!\left\{g\!\left(\frac{2\pi}{L} n, \omega\right)\!\!+\!\!f\!\left(\frac{2\pi}{L} n, \omega\right) \!\cos(\!2t \varepsilon\!\left(\frac{2\pi}{L} n,\omega\right)\!\!)\!\!\right\}\!,
\ee
where we introduced 
\be
\begin{aligned}
f(k,\omega)&= \cos(\Phi(k,\omega))^2, \qquad
g(k,\omega)=\sin(\Phi(k,\omega))^2,\\
\Phi(k,\omega) &= \arctan(\tan(\omega/2) \sin(k))+\pi\theta_{\rm H}\left(k-\frac{\pi}{2}\right),\\
\varepsilon(k,\omega) &= \arctan\left(\frac{\sqrt{\sin ^4\left(\frac{{\omega}}{2}\right) \sin ^2(2 k)+\sin ^2({\omega}) \cos ^2(k)}}{\cos ^2\left(\frac{{\omega}}{2}\right)-\sin ^2\left(\frac{{\omega}}{2}\right) \cos (2 k)}\right)\!\!.
\end{aligned}
\ee
The expression in Eq.~\eqref{eq:expvalZL} does still not relax at large times --- its $t\to\infty$ does not exist ---, however, if we first consider its ``thermodynamic limit'' ($L\to \infty$) we obtain
\be
\label{eq:expvalZTL}
\!\!\!\lim_{L\to\infty}\expval{Z_x}{\Psi_t} 
\!=\! \frac{(-1)^x}{4\pi} \int_0^\pi {\rm d}k \left\{g(k, \omega)+f(k, \omega)\cos(2t \varepsilon(k,\omega))\right\}\!.
\ee
This expression attains a well defined limit for $t \to \infty$! 

This is a general fact. In order to attain a well defined infinite-time limit one needs a perfect destructive interference between the phases of the off diagonal terms. This can only occur in the limit where the spectrum becomes continuous \emph{and} if the observable has non zero overlap with almost all the eigenstates within a given energy window. This means that local in-space observables will typically relax while non-local ones will often not.  A simple example is 
\be
\mathcal O = \ketbra{n}{m}+\ketbra{m}{n}, 
\ee
where $\ket{n}$ and $\ket{m}$ are two eigenstates such that $E_n-E_m = O(L^0)$. In this case one has
\be
\expval{\mathcal O}{\Psi_t} = 2 \Re[\braket{\Psi_0}{n} \braket{\Psi_0}{m}^* e^{i (E_n-E_m)t}].
\ee
This is the same general phenomenology that one observes after quenches in short-range interacting Hamiltonian systems~\cite{essler2016quench}.

\chapter{}

In the last lecture we showed that BQC encode the physics of local interactions and show relaxation in the thermodynamic limit ($L\to\infty$). Namely, one generally has 
\be
\label{eq:localrelaxation}
\lim_{t\to\infty}\lim_{L\to \infty}\tr[\rho_A(t)\mathcal O_A] \to \lim_{L\to \infty} \tr[\rho_{\rm st}\mathcal O_A],
\ee
where $A$ is a finite region and $\rho_A(t)=\tr_{\bar A}[\ketbra{\Psi_t}]$ is its reduced state. Analogously to local Hamiltonian systems~\cite{essler2016quench}, the stationary state $\rho_{\rm st}$ can be described \`a la Gibbs, based on the number of local conserved operators (or ``conservation laws'') in the system. The latter are a special kind of conserved operators written as 
\be
\label{eq:quasilocal}
Q = \sum_x q_x,
\ee
with $q_x$ acting non-trivially only around $x$, and are the only ones effectively constraining the dynamics of local subsystems~\footnote{In fact, it has been shown that also operators written as Eq.~\eqref{eq:quasilocal} but where the support $q_x$ decays exponentially away from $x$  --- quasi local --- need to be taken into account in order for Eq.~\eqref{eq:localrelaxation} to hold~\cite{ilievski2016quasi}.}.

As opposed to Hamiltonian systems that have always at least $H$ as local conservation laws, generic BQCs have no local conservation laws. Therefore we expect them to relax to the infinite temperature state
\be
\label{eq:infinitetemperature}
\rho_{\rm st} = \frac{\mathds{1}}{\tr[\mathds{1}]},
\ee
which is typically disdained by physicists as featureless and ``boring''. This is not always the case of course: you saw the lectures of Fabian Essler and Bal\'azs Pozsgay about \emph{integrable models}, where the number of local conservation laws is $\propto L$ and the manifold of possible states is much richer~\cite{essler2016quench}. Importantly, integrable systems exist also in BQC form (a simple example is the Floquet XX model that we encountered in the previous lecture). Here, however, I want to focus on the \emph{generic case} --- where $\rho_{\rm st}$ is given by Eq.~\eqref{eq:infinitetemperature} --- and ask what happens for \emph{finite times}. Specifically, I will discuss some general features of the dynamics of many-body systems that can be quantitatively understood in BQCs.

\section{Finite-time evolution}

When trying to describe ``general features'' of the many-body dynamics the first question that we are facing is ``where should we look?'', i.e., what observable should we study to identify them? This is a difficult question as it is clear that by looking at the dynamics of the expectation value of a specific observable it will always be hard to disentangle the specific features of the observable from the general properties of the system. One way to solve this problem would be to study the dynamics of all possible local observables, however, it would be very impractical. A more convenient way to proceed is to probe the evolution of $\rho_A(t)$ itself, e.g., characterising the evolution of its spectrum by computing its von-Neumann entropy
\be
\label{eq:entanglemententropy}
S_A(t) = -\tr[\rho_A(t)\log \rho_A(t)].
\ee
This quantity has also an important quantum information theoretical meaning: Whenever the state of the full system is pure --- which is the one of interest for us --- $S_A(t)$ measures the entanglement between $A$ and the rest of the system. For this reason, it is typically called the \emph{entanglement entropy}.

\begin{figure}
\centering
    \begin{tikzpicture}[baseline={([yshift=-0.6ex]current bounding box.center)},scale=1.5]
    \draw[dashed, thick]  (2.5, 0) -- (2.5, 2.5); 
     \draw[thick,black,->] (0,0) -- (0,3.5) node at (-.35,3) {$S_A(t)$};
     \draw[thick, domain = 0: 2.5]    plot (\x,{abs(\x)});
      \draw[thick, domain = 2.5: 6]    plot (\x,2.5);
     \draw[thick,black,->] (-0.25,0) -- (7,0)  node at (6.75,-.25) {$t$};] 
     \node at (2.5,-.25) {$\propto |A|$};
     \node at (4.5,3) {thermodynamic entropy of $\rho_{{\rm st}, A}$};
     \draw[thick,black,->] (4.5,2.85) -- (4.5,2.55);  
  \end{tikzpicture}
  \caption{Leading-order evolution of the entanglement entropy in generic systems with local interactions.}
  \label{fig:entropyevolution}
\end{figure}
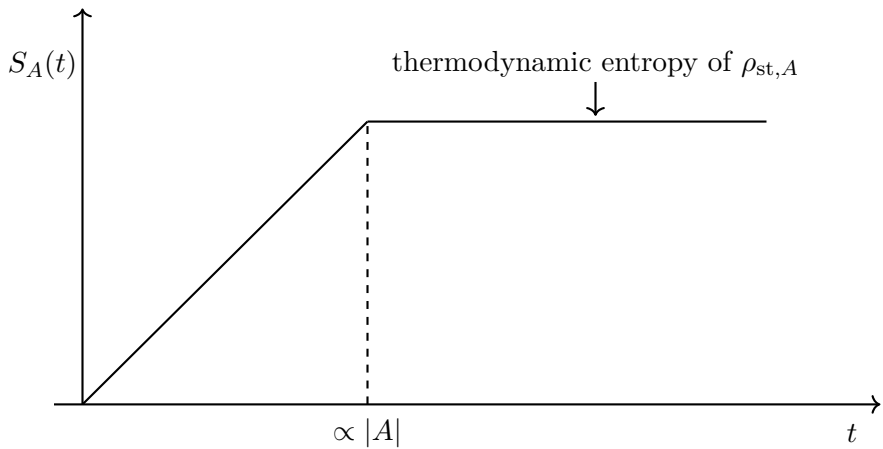

It turns out that $S_A(t)$ has a strikingly similar evolution in essentially all systems characterised by local interactions, regardless of dimensionality, number of conservation laws, the discrete or continuous nature of space or time, or any other specific properties of the system such as masses or coupling strengths (perhaps cases of strong disorder causing real-space localisation or similar extreme situations) --- borrowing a term from the theory of critical phenomena we can say that the dynamics of $S_A(t)$ is \emph{universal}. The typical evolution of the entanglement entropy following a quantum quench is sketched in Fig.~\ref{fig:entropyevolution}: there is an initial linear growth characterised by a system-dependent slope (or ``entanglement velocity'') followed by a relatively sharp relaxation to a value that is extensive in the size of the subsystem and corresponds to the thermodynamic entropy of the stationary state $\rho_{{\rm st}}$ reduced to the subsystem~\footnote{The first account of such ``universal'' behaviour was given in the context of $(1 +1)$-dimensional conformal field theory~\cite{calabrese2005evolution}, where it was explained postulating that the entanglement is carried though the system by pairs of ``correlated quasiparticles'' produced at each point in space by the non-equilibrium initial state at the beginning of the evolution. This semiclassical description became known as the ``quasiparticle picture'' and has later been extended to characterise the dynamics of entanglement in other models with stable quasi-particles such as free~\cite{fagotti2008evolution} and interacting integrable models~\cite{alba2017entanglement}. A few years later, however, the same behaviour was also found in systems with no quasiparticles such as holographic conformal field theories~\cite{liu2014entanglement}, and generic interacting systems~\cite{kim2013ballistic, laeuchli2008spreading}.}. Therefore, one can give a compact description of the equilibration process of a subsystem as the one turning entanglement entropy into thermodynamic entropy (cf.~~\cite{calabrese2020entanglement}).

This universal phenomenon inevitably attracted the attention of many theoretical physicists with particular attention being devoted to the linear growth of entanglement observed at early times. Indeed, linear-in-time behaviours are typically associated with ballistic transport which is clearly not expected to take place in generic systems. However, computing $S_A(t)$ in the presence of interactions is far from easy and the progress was stagnant. The achievement of exact results in BQSs was instrumental to provide a better characterisation of this phenomenon and (partially) explain it~\cite{nahum2017quantum, zhou2020entanglement, bertini2022growth} (see also Romain Vasseur's lectures~\cite{vasseur2026leshouches}). 

Before illustrating how these exact results can be obtained, however, it is useful to simplify the object of interest. Indeed, one feature that makes the calculation of $S_A(t)$ difficult is the presence in Eq.~\eqref{eq:entanglemententropy} of the logarithm of the reduced density matrix. To avoid this issue one can employ a ``replica trick'', i.e., obtain $S_A(t)$ as the limit of quantities that are easier to compute. Specifically, we introduce the so called  ``R\'enyi entropies'' given by 
\be
\label{eq:renyientropies}
S^{(n)}_A(t) = \frac{1}{1-n}\log\tr[\rho^n_A(t)],
\ee
and fulfilling 
\be
\lim_{n\to 1} S^{(n)}_A(t)  = S_A(t). 
\ee
R\'enyi entropies are easier to compute as they only involve logarithms of real numbers rather than of an operator. Moreover, they also turn out to behave ``universally'' and in fact their evolution is typically very similar to the one of $S_A(t)$ (there are however interesting exceptions, see, e.g.~\cite{rakovszky2019sub, huang2020dynamics}). 

In the following I will discuss how to compute $S^{(n)}_A(t)$ in the particularly simple case of $n\in \mathbb N$ and $n\ge2$. Whenever we can obtain the analytic dependence on $n$ the results for these quantities can be used to compute $S_A(t)$ via analytic continuation~\footnote{In fact, due to the discreteness of space, the results for $\{S^{(n)}_A(t)\}_{n\in \mathbb N >1}$ characterise the full spectrum of $\rho_A(t)$. This means that, if one can compute them exactly, R\'enyi entropies with integer index give access to $S_A(t)$ without the analytic continuation.}. 

\subsection{R\'enyi entropies in BQCs}
\label{sec:folding}

Employing the diagrammatic representation discussed in the previous lecture we can represent the reduced density matrix of a block of qudits $A$~\footnote{we denote the number of qudits in $A$ by $2|A|$ because our lattice spacing is $1/2$} as follows
\begin{align}
\hspace{-1cm} \rho_A(t) = 
\fineq[-0.8ex][.6][1]{
\foreach \j in {3,...,4}{	
\foreach \i in {-1,...,5}{
\roundgate[2*\i][2*\j+1.5][1][topright][red6][-1]
\roundgate[2*\i+1][2*\j+.5][1][topright][red6][-1]}}
\foreach \j in {3,...,4}{	
\foreach \i in {-1,...,5}{
\roundgate[2*\i][2*\j+.5-5][1][bottomright][blue6][-1]
\roundgate[2*\i+1][2*\j+1.5-5][1][bottomright][blue6][-1]}}
\draw[thick, fill=gray, rounded corners =0.5pt] (-1.8,5) rectangle (11.8,5.2);
\draw[thick, fill=white, rounded corners =0.5pt] (-1.8,6) rectangle (11.8,5.8);
\draw [decorate, thick, decoration = {brace}]   (2.25,10.125)--++(3.5,.0);
\node[scale=1.75] at (4,10.6) {$2 |A|$};
\foreach \i in {0,...,4}{	
\draw[thick] (-2.5+\i,10) -- (-2.5+\i,10.5+\i*0.25) -- (-3-\i*0.25,10.5+\i*0.25) -- (-3-\i*0.25, .5-\i*0.25) -- (-2.5+\i,.5-\i*0.25) -- (-2.5+\i, 1);}
\foreach \i in {0,...,4}{	
\draw[thick] (10.5-\i,10) -- (10.5-\i,10.5+\i*0.25) -- (12+\i*0.25,10.5+\i*0.25) -- (12+\i*0.25, .5-\i*0.25) -- (10.5-\i,.5-\i*0.25) -- (10.5-\i, 1);}
}, \hspace{-0.5cm}
\end{align}
where we made no assumption on the structure of the initial state representing it as a white box with $2L$ legs (the gray box denotes its complex conjugate). Therefore, the trace of its $n$-th power can be written as  
\begin{align}
\hspace{-.35cm} \tr[\rho^n_A(t)] = 
\fineq[-0.8ex][.5][1]{
\foreach \i in {2,...,3}{
\draw[thick] (12+2*\i+2-2.5-12,10) arc (225:0:0.15);
\draw[thick] (8+2*\i+2-11.5,10) arc (-45:175:0.15);
}
\foreach \i in {2,...,3}{
\draw[thick] (12+2*\i+2-2.5-12,-15) arc (-225:0:0.15);
\draw[thick] (8+2*\i+2-11.5,-15) arc (45:-175:0.15);
}
\foreach \j in {3,...,4}{	
\foreach \i in {-1,...,5}{
\roundgate[2*\i][2*\j+1.5][1][topright][red6][-1]
\roundgate[2*\i+1][2*\j+.5][1][topright][red6][-1]}}
\foreach \j in {3,...,4}{	
\foreach \i in {-1,...,5}{
\roundgate[2*\i][2*\j+.5-5][1][bottomright][blue6][-1]
\roundgate[2*\i+1][2*\j+1.5-5][1][bottomright][blue6][-1]}}
\draw[thick, fill=gray, rounded corners =0.5pt] (-1.8,5) rectangle (11.8,5.2);
\draw[thick, fill=white, rounded corners =0.5pt] (-1.8,6) rectangle (11.8,5.8);
\draw [decorate, thick, decoration = {brace}]   (13.5,5)--++(0,-15);
\node[scale=1.75] at (14,-2.5) {$n$};
\foreach \i in {0,...,4}{	
\draw[thick] (-2.5+\i,10) -- (-2.5+\i,10.5+\i*0.25) -- (-3.5-\i*0.25,10.5+\i*0.25) -- (-3.5-\i*0.25, .5-\i*0.25) -- (-2.5+\i,.5-\i*0.25) -- (-2.5+\i, 1);}
\foreach \i in {0,...,4}{	
\draw[thick] (10.5-\i,10) -- (10.5-\i,10.5+\i*0.25) -- (12+\i*0.25,10.5+\i*0.25) -- (12+\i*0.25, .5-\i*0.25) -- (10.5-\i,.5-\i*0.25) -- (10.5-\i, 1);}
\foreach \i in {2,...,7}{	
\draw[thick] (\i-0.5,1) -- (\i-0.5,-.5);
\draw[thick] (\i-0.5,-4.5) -- (\i-0.5,-6);
\node[scale=1.75] at (\i-0.5,-2.5) {$\vdots$};}
\def\shifty{-16}
\foreach \j in {3,...,4}{	
\foreach \i in {-1,...,5}{
\roundgate[2*\i][2*\j+1.5+\shifty][1][topright][red6][-1]
\roundgate[2*\i+1][2*\j+.5+\shifty][1][topright][red6][-1]}}
\foreach \j in {3,...,4}{	
\foreach \i in {-1,...,5}{
\roundgate[2*\i][2*\j+.5-5+\shifty][1][bottomright][blue6][-1]
\roundgate[2*\i+1][2*\j+1.5-5+\shifty][1][bottomright][blue6][-1]}}
\draw[thick, fill=gray, rounded corners =0.5pt] (-1.8,5+\shifty) rectangle (11.8,5.2+\shifty);
\draw[thick, fill=white, rounded corners =0.5pt] (-1.8,6+\shifty) rectangle (11.8,5.8+\shifty);
\foreach \i in {0,...,4}{	
\draw[thick] (-2.5+\i,10+\shifty) -- (-2.5+\i,10.5+\i*0.25+\shifty) -- (-3.5-\i*0.25,10.5+\i*0.25+\shifty) -- (-3.5-\i*0.25, .5-\i*0.25+\shifty) -- (-2.5+\i,.5-\i*0.25+\shifty) -- (-2.5+\i, 1+\shifty);}
\foreach \i in {0,...,4}{	
\draw[thick] (10.5-\i,10+\shifty) -- (10.5-\i,10.5+\i*0.25+\shifty) -- (12+\i*0.25,10.5+\i*0.25+\shifty) -- (12+\i*0.25, .5-\i*0.25+\shifty) -- (10.5-\i,.5-\i*0.25+\shifty) -- (10.5-\i, 1+\shifty);}
}. 
\label{eq:tracerhon}
\end{align}
This representation can be made more compact by folding the $2n$ replicas of the time evolution on top of each other in such a way that each copy of the backward evolution (blue) is underneath the nearest forward one (red). We then introduce a symbol for stacks of $n$ local gates alternated with its complex conjugate as follows
\begin{align}
\left(U\otimes U^*\right)^{\otimes n}=\fineq{
\foreach \i in {0,...,-2}
{\roundgate[\i*0.3][\i*0.15][1][topright][blue6]
\roundgate[\i*0.3-.15][\i*0.15-.075][1][topright][red6]
}
\draw [decorate, decoration = {brace}]   (-1.35,0.1)--++(1*0.9,.5*0.9);
\node[scale=1.25] at (-1.135,.5) {${}_{2n}$};
}=\fineq{\roundgate[0][0][1][topright][orange][n]}.
\label{eq:foldedgatepicture}
\end{align}
Analogously, we introduce a symbol for piles of $n$ initial states alternated with its complex conjugate
\be
\hspace{-.35cm} \left(\ket{\Psi_0} \otimes \ket{\Psi_0}^*\right)^{\otimes n} = \fineq{
\foreach \j in {0,...,-2}{
\foreach \i in {1,...,3}{
\draw[thick] (\i-0.5+\j*0.3,0+\j*0.15) -- (\i-0.25+\j*0.3,.25+\j*0.15);
\draw[thick] (\i-0.5+\j*0.3-.15,0+\j*0.15-.075) -- (\i-0.25+\j*0.3-.15,.25+\j*0.15-.075);
\draw[thick] (\i+0.25+\j*0.3,0+\j*0.15) -- (\i+\j*0.3,.25+\j*0.15);
\draw[thick] (\i+0.25+\j*0.3-.15,0+\j*0.15-.075) -- (\i+\j*0.3-.15,.25+\j*0.15-.075);}
\draw[thick, fill=gray, rounded corners =0.5pt] (.35+\j*0.3,.1+\j*0.15) rectangle (3.35+\j*0.3,-.1+\j*0.15);
\draw[thick, fill=white, rounded corners =0.5pt] (.35+\j*0.3-.15,.1+\j*0.15-.075) rectangle (3.35+\j*0.3-.15,-.1+\j*0.15-.075);}
\draw [decorate, decoration = {brace}]   (-.5,-0.2)--++(1*0.9,.5*0.9);
\node[scale=1.25] at (-.75,0) {${}_{2n}$};
}=\fineq{
\foreach \i in {1,...,3}{
\draw[thick] (\i-0.5,0) -- (\i-0.25,.25);
\draw[thick] (\i+0.25,0) -- (\i,.25);}
\draw[fill=black, rounded corners =0.5pt] (.35,.1) rectangle (3.35,-.1);}.
\ee
Finally, we introduce special symbols to represent special states in the replicated space that correspond to contractions between the replicas. Specifically, labelling the $2n$ replicas as $(1,1^*, 2, 2^*,\ldots, n, n^*)$ we define~\footnote{The number of replicas of the space where these states are defined is clear for the context. I will not report it on the states for brevity.} 
\begin{align}
\ket{\mcirc} &= \sum_{i_1,\ldots,i_n =1}^d \ket{i_1,i_{1},i_2,i_{2},\ldots, i_n, i_{n}} \equiv \fineq{\draw[thick] (0,0)--++(0,-.5);\cstate[0][0]}, \\
\ket{\msqr} &= \sum_{i_1,\ldots,i_n =1}^d \ket{i_1,i_{2},i_{2},i_{3},\ldots, i_n, i_{1}} \equiv \fineq{\draw[thick] (0,0)--++(0,-.5);\sqrstate[0][0]}. 
\end{align}
In this way Eq.~\eqref{eq:tracerhon} is represented as 
\begin{align}
\label{eq:tracerhonfolded}
\hspace{-1cm} \tr[\rho^n_A(t)] \!=\!\!
\fineq[-0.8ex][.7][1]{
\foreach \j in {3,...,4}{	
\foreach \i in {-1,...,5}{
\roundgate[2*\i][2*\j+1.5][1][bottomright][orange][n]
\roundgate[2*\i+1][2*\j+.5][1][bottomright][orange][n]}}
\draw[fill=black, rounded corners =0.5pt] (-1.8,6) rectangle (11.8,5.8);
\foreach \i in {7,..., 11}{
\cstate[\i-0.5][10]}
\foreach \i in {-2,..., 2}{
\cstate[\i-0.5][10]}
\foreach \i in {3,..., 6}{
\sqrstate[\i-0.5][10]}
}\!\!\!\!. \hspace{-1cm}
\end{align}
It is also useful to recall that, in this notation, the unitarity condition for the local gate (cf.~Eq.~\eqref{eq:unitarity}) is expressed as 
\begin{align}
	\!\!\!\!\!\!\!\!\!\fineq[-0.8ex][.9][1]{
	\roundgate[1][1][1][topright][orange][n]
	\cstatep[1.5][1.5][]
	\cstatep[.5][1.5][]
}&=
\fineq[-0.8ex][.9][1]{
	\draw[thick] (.5,1.5)--++(0,-1);
	\draw[thick] (1.5,1.5)--++(0,-1);
	\cstatep[1.5][1.5][] 
	\cstatep[.5][1.5][] },\quad\!\!
\fineq[-0.8ex][.9][1]{
	\roundgate[1][1][1][topright][orange][n]
	\cstatep[1.5][.5][] 
	\cstatep[.5][.5][]
	 }=\fineq[-0.8ex][.9][1]{
	\draw[thick] (.5,1.5)--++(0,-1);
	\draw[thick] (1.5,1.5)--++(0,-1);
	\cstatep[.5][.5][] 
	\cstatep[1.5][.5][]
	 },\label{eq:unitaritynfoldeddiagramcirc}\\
	 \\
	\!\!\!\!\!\!\!\!\!\fineq[-0.8ex][.9][1]{
	\roundgate[1][1][1][topright][orange][n]
	\sqrstate[1.5][1.5]
	\sqrstate[.5][1.5]
}&=
\fineq[-0.8ex][.9][1]{
	\draw[thick] (.5,1.5)--++(0,-1);
	\draw[thick] (1.5,1.5)--++(0,-1);
	\sqrstate[1.5][1.5]
	\sqrstate[.5][1.5]},\quad\!\!
\fineq[-0.8ex][.9][1]{
	\roundgate[1][1][1][topright][orange][n]
	\sqrstate[1.5][.5] 
	\sqrstate[.5][.5]
	 }=\fineq[-0.8ex][.9][1]{
	\draw[thick] (.5,1.5)--++(0,-1);
	\draw[thick] (1.5,1.5)--++(0,-1);
	\sqrstate[.5][.5]
	\sqrstate[1.5][.5]
	 }.\label{eq:unitaritynfoldeddiagramsqr}
\end{align}
In fact, we have many more conditions coming from unitarity. For any ``pairing state''  indexed by a permutation $\sigma$ of $n$ elements as
\begin{align}\label{eq:permutation_states}
\ket{\sigma}= \sum_{i_1,\ldots,i_n =1}^d \ket{i_1,i_{\sigma(1)},\ldots, i_n, i_{\sigma(n)}} \equiv\fineq[-0.4ex]{\draw[thick] (0,0)--++(0,-.5);
\cstatep[0][0][\sigma]}\,,
\end{align}
we have 
\begin{align}
	\!\!\!\!\!\!\!\!\!\fineq[-0.8ex][.9][1]{
	\roundgate[1][1][1][topright][orange][n]
	\cstatep[1.5][1.5][\sigma]
	\cstatep[.5][1.5][\sigma]
}&=
\fineq[-0.8ex][.9][1]{
	\draw[thick] (.5,1.5)--++(0,-1);
	\draw[thick] (1.5,1.5)--++(0,-1);
	\cstatep[1.5][1.5][\sigma] 
	\cstatep[.5][1.5][\sigma] },\quad\!\!
\fineq[-0.8ex][.9][1]{
	\roundgate[1][1][1][topright][orange][n]
	\cstatep[1.5][.5][\sigma] 
	\cstatep[.5][.5][\sigma]
	 }=\fineq[-0.8ex][.9][1]{
	\draw[thick] (.5,1.5)--++(0,-1);
	\draw[thick] (1.5,1.5)--++(0,-1);
	\cstatep[.5][.5][\sigma] 
	\cstatep[1.5][.5][\sigma]
	 }.\label{eq:unitaritynfoldeddiagramperm}
\end{align}
The states $\ket{\mcirc}$ and $\ket{\msqr}$ are just $2$ of these $n!$ states. 

Going back to Eq.~\eqref{eq:tracerhonfolded} and using multiple times the condition Eq.~\eqref{eq:unitaritynfoldeddiagramsqr} we find 
\be
\label{eq:tracerhonsimplified}
\hspace{-1cm} \tr[\rho^n_A(t)] \!=\hspace{-0.75cm}
\fineq[-0.8ex][.7][1]{
\draw[thick] (6.5,-1.25)--++(0,-.25);
\draw[thick] (-6.5,-1.25)--++(0,-.25);
\cstate[6.5][-1.25]
\cstate[-6.5][-1.25]
\foreach \j [evaluate=\j as \jplus using {\j+1}] in {1,...,2}
{	
\foreach \i in {-\j,...,\j}
{
\roundgate[2*\i][4-2*\j][1][bottomright][orange][n]
}
\foreach \i  [evaluate=\i as \ieval using {\i+.5}] in {-\jplus,...,\j}
{
\roundgate[2*\ieval][4-2*\j-1][1][bottomright][orange][n]
}
}
\foreach \i in {0,...,3}{
\cstate[\i-5.5][\i-.5]
\cstate[-\i+5.5][\i-.5]
}
\foreach \i in {-1,...,2}{
\sqrstate[\i-.5][2.5]
}

\draw[fill=black, rounded corners =0.5pt] (-1.8-5,-1.7) rectangle (11.8-5,-1.5);	
\draw [decorate, thick, decoration = {brace}]   (-1.5,2.75)--++(3,.0);
\node[scale=1.25] at (.5,3.25) {$2 |A|$};		
\draw [decorate, thick, decoration = {brace}]   (2.65,2.75)--++(3,-3);
\node[scale=1.25] at (4.5,1.5) {$2 t$};
}.
\ee
This diagram is a tensor network whose size scales linearly with $t$ and $|A|$ (and the local Hilbert space is $d^n$ dimensional): contracting it for generic choices of the unitary gates is very hard. As I discuss in the next section, however, there are some interesting choices of gates for which this is possible. Plugging back into Eq.~\eqref{eq:renyientropies} this gives the desired R\'enyi entropies.

\section{Circuits with solvable entanglement dynamics}

In last decade two different strategies have emerged to find solvable (but generically non-integrable) instances of BQCs:
\begin{itemize}
\item[(i)] Averaging over random gates.
\item[(ii)] Imposing particular conditions (or symmetries) on the gates.  
\end{itemize}
These two strategies are somewhat antithetical --- one is based upon reducing the constraints on (``the structure of'') the dynamics to a minimum and the other upon applying more constraints. Nevertheless, as we shall see, they both lead to interesting solvable examples.

\subsection{Random unitary circuits}

The idea of considering quantum circuits made of random unitary gates has appeared several times in the quantum information theoretic literature~(see, e.g.,~\cite{znidaric2008exact, harrow2009random}), however, to the best of my knowledge the first reference considering them as a bona fide quantum many-body system has been~\cite{nahum2017quantum}. In essence, the logic underlying this choice is the same that motivates random matrix theory (see, e.g.,~\cite{mehta2004random}): when the properties of an individual system are out of reach one can make progress by averaging over ensembles of similar systems. From this point of view one can think of random unitary circuits as a special kind of random matrix theory encoding spatial locality. Here I will only discuss the microscopic calculation of Eq.~\eqref{eq:tracerhonsimplified} in these systems, while a more extensive discussion is provided in Romain Vasseur's lectures (see also the reviews~\cite{potter2022entanglement, fisher2023random}).  

The simplest random unitary circuit with brickwork architecture is obtained considering the evolution operator in Eq.~\eqref{eq:floquetoperator} where, at each time, each $\{U^{(x)}\}_{x\in\mathbb Z_{2L}/2}$ is drawn independently from $U(d^2)$ with a probability proportional to the Haar measure. The focus is then on computing averaged quantities. For example in the case of $\tr[\rho^n_A(t)]$ we have 
\be
\tr[\rho^n_A(t)] \longrightarrow \mathbb E[\tr[\rho^n_A(t)]] = \int_{U(d^2)} \tr[\rho^n_A(t)] \prod_{x=0}^{2L-1} \prod_{\tau=1}^{2t} {\rm d}\mu_{\rm H}(U^{(x/2,\tau/2)}),
\label{eq:globalaverage}
\ee 
where ${\rm d}\mu_{\rm H}$ is the normalised Haar measure. Because of the independence of random unitaries at different space-time points $\mathbb E[\tr[\rho^n_A(t)]] $ can still be expressed as a tensor network like Eq.~\eqref{eq:tracerhonsimplified} where the local gate is changed as 
\begin{align}
\label{eq:weingarten}
\fineq{\roundgate[0][0][1][topright][orange][n]} \longrightarrow  \fineq{\roundgate[0][0][1][topright][gray][n]} &= \int_{U(d^2)}   \left(U\otimes U^*\right)^{\otimes n} {\rm d}\mu_{\rm H}(U)\notag\\
&= \sum_{i,j=1}^{n!} [{\rm Wg}]_{i,j} \ketbra{\sigma_i\sigma_i}{\sigma_j\sigma_j},
\end{align}
where in the second step we used that the integral over the unitary group can be computed exactly (see, e.g.,~\cite{collins2022weingarten}): the states appearing in there are the ``pairing states'' in Eq.~\eqref{eq:permutation_states} and the matrix of the coefficients, $[{\rm Wg}]_{i,j}$, is known as Weingarten matrix (see Romain's lecture for more details). 

Note that, after averaging, the effective local dimension passes from $d^{2n}$ to $n!$~\footnote{At least for $d\geq n$, otherwise the pairing states start to be linearly dependent.}. This means that for $n=2$ the local space is effectively $2$-dimensional and the treatment becomes much easier. Indeed, Eq.~\eqref{eq:weingarten} implies 
\be
\label{eq:basiclocalrel}
\fineq[-0.8ex][.9][1]{
	\roundgate[1][1][1][topright][gray][2]
	\cstatep[.5][1.5][]
	\sqrstate[1.5][1.5]
} =
\frac{d}{d^2+1}
\fineq[-0.8ex][.9][1]{
	\draw[thick] (.5,1.5)--++(0,-1);
	\draw[thick] (1.5,1.5)--++(0,-1);
	\sqrstate[1.5][1.5] 
	\sqrstate[.5][1.5]
}
+
\frac{d}{d^2+1}
\fineq[-0.8ex][.9][1]{
	\draw[thick] (.5,1.5)--++(0,-1);
	\draw[thick] (1.5,1.5)--++(0,-1);
	\cstatep[1.5][1.5][] 
	\cstatep[.5][1.5][] 
},
\ee
where I used the specific values of the Weingarten matrix to find the coefficients~\footnote{The states appearing in the expansion, however, are set by Eq.~\eqref{eq:weingarten}.}. This equation can be used to directly contract the diagram starting from the top: it ensures that the state on the $m$-th horizontal cut ($m=0,\ldots, 2t$) is a linear combination of ``domain wall states'' of the form 
\be
\ket{\ell_1,\ell_2,\ell_3}=\ket*{\overbrace{\mcirc  \ldots \mcirc }^{\ell_1} \overbrace{\msqr \ldots \msqr}^{\ell_2} \overbrace{\mcirc \ldots \mcirc}^{\ell_3}}, \qquad \ell_1 + \ell_2+ \ell_3 = |A| + 2m, 
\ee 
where the coefficients can be easily determined recursively (see, e.g.,~\cite{klobas2024translation}). The treatment further simplifies for initial states in product form and $|A|>2t$: in this case Eq.~\eqref{eq:weingarten} directly gives
\be
\mathbb E[\tr[\rho^2_A(t)]] = \left(\frac{2d}{d^2+1}\right)^4 \mathbb E[\tr[\rho^2_A(t-{1}/{2})]].
\ee
Noting that $\mathbb E[\tr[\rho^2_A(0)]]=1$ we then have 
\be
\mathbb E[\tr[\rho^2_A(t)]] =  \left(\frac{2d}{d^2+1}\right)^{4 t}\,.
\ee
Similar exact calculations are possible for other quantities involving two replicas of the time-evolution, e.g., out-of-time-order correlators~\cite{potter2022entanglement, fisher2023random}. 

For $n>2$ the analogue of Eq.~\eqref{eq:weingarten} produces more terms: this complicates the treatment severely enough to prevent an exact contraction of the diagram. As shown in by~\cite{zhou2019emergent}, however, can still analyse it in the limit of large $d$ and obtain 
\be
\mathbb E[\tr[\rho^n_A(t)]] \simeq \exp[- 4 t (n-1) \log d + O((\log d)^0)].
\ee
Random unitary circuits can be made ``less random'' by restricting the integration measure in Eq.~\eqref{eq:globalaverage}. For instance, one can consider ``Floquet random unitaries''~\cite{chan2018solution} where $\{U^{(x)}\}_{x\in\mathbb Z_{2L}/2}$ are chosen at random for $t=1$ but then are kept fixed in time, or even ``space-time translational invariant''~\cite{chan2022many} random circuits where one chooses one unitary matrix $U\in U(d^2)$ at random and then sets $U^{(x)}=U$ keeping this choice for all times. Doing this, however, makes these systems increasingly less solvable and exact results are restricted to the large-$d$ limit. 

\subsection{Dual-unitary circuits}
\label{sec:DUentanglement}

A possible motivation for approach (ii) --- imposing more constraints --- comes from the symmetry of BQCs under the exchange of space and time. The basic observation is that, since BQCs involve local interactions and evolve in discrete steps, the exchange of space and time in a BQC produces another system with the structure of a BQC. For instance, considering the diagram in Eq.~\eqref{eq:tracerhonfolded} and exchanging the roles of space and time, i.e., rotating it by 90 degrees, we obtain a diagram representing the evolution of a BQC (albeit with strange boundary conditions)  --- this is made more transparent by the choice of representing the gates' legs ``diagonally'' that I made in these lectures. Generically, however, the BQC-like system obtained by such a ``swap'' of space and time is not described by unitary gates. Therefore the symmetry is ``broken'': space-like and time-like dynamics are not equivalent. This suggests that the case in which this symmetry is instead preserved should be \emph{special} and perhaps \emph{solvable}. At the same time, this symmetry does not have an obvious relation with integrability. We are then left with the fascinating prospect of finding a family of solvable non-integrable models --- this is certainly a good enough reason for considering this case. 

Let us begin by defining more specifically what we mean by space-time swap. For a given local gate $U$ we define its space-time swapped counterpart, $\tilde U$, through the following reshuffling of matrix elements 
\be
\mel{ij}{\tilde U}{k\ell} = \fineq[-0.8ex][1][1]{
		\tsfmatV[0][-0.5][r][1][][][red6]
		\node at (-0.5,-0.2) {$i$};
		\node at (0.5,-0.2) {$k$};
		\node at (-0.5,1.2) {$j$};
		\node at (0.5,1.25) {$\ell$};} = \mel{j\ell}{U}{ik}\,.
\ee
The meaning of $\tilde U$ is now clear: while the four-legged tensor in the above equation acts like $U$ going from bottom to top, it acts as $\tilde U$ going from left to right. Requiring unitarity in both the temporal (upwards) and spatial (rightwards) directions leads to the following constraints
\be
U U^\dag = \mathds{1} = \tilde U \tilde U^\dag, 
\ee
which in the compact diagrammatic representation of the previous section are expressed as 
\begin{align}
	\!\!\!\!\!\!\!\!\!\fineq[-0.8ex][1][1]{
	\roundgate[1][1][1][topright][orange][n]
	\cstatep[1.5][1.5][\sigma]
	\cstatep[.5][1.5][\sigma]
}&=
\fineq[-0.8ex][1][1]{
	\draw[thick] (.5,1.5)--++(0,-1);
	\draw[thick] (1.5,1.5)--++(0,-1);
	\cstatep[1.5][1.5][\sigma] 
	\cstatep[.5][1.5][\sigma] },\quad\!\!
\fineq[-0.8ex][1][1]{
	\roundgate[1][1][1][topright][orange][n]
	\cstatep[1.5][.5][\sigma] 
	\cstatep[.5][.5][\sigma]
	 }=\fineq[-0.8ex][1][1]{
	\draw[thick] (.5,1.5)--++(0,-1);
	\draw[thick] (1.5,1.5)--++(0,-1);
	\cstatep[.5][.5][\sigma] 
	\cstatep[1.5][.5][\sigma]
	 },\label{eq:unitaritynfoldeddiagram}\\
\!\!\!\!\!\!\!\!\!\fineq[-0.8ex][1][1]{
	\roundgate[1][1][1][topright][orange][n]
	\cstatep[1.5][1.5][\sigma] 
	\cstatep[1.5][.5][\sigma]
}&=
\fineq[-0.8ex][1][1]{
	\draw[thick] (1.5,1.5)--++(-1,0);
	\draw[thick] (1.5,.5)--++(-1,0);
	\cstatep[1.5][1.5][\sigma] 
	\cstatep[1.5][.5][\sigma]
	 },\qquad\!\!
\fineq[-0.8ex][1][1]{
	\roundgate[1][1][1][topright][orange][n]		
	\cstatep[.5][1.5][\sigma]
	\cstatep[.5][.5][\sigma]
	 }=\fineq[-0.8ex][1][1]{
	\draw[thick] (1.5,.5)--++(-1,0);
	\draw[thick] (1.5,1.5)--++(-1,0);
	\cstatep[.5][1.5][\sigma]
	\cstatep[.5][.5][\sigma] 
	}.
	\label{eq:spaceunitaritynfoldeddiagram}
\end{align}
These constraints --- which turn out to have non-trivial solutions for all $d\geq 2$ --- define the family of so-called \emph{dual-unitary} gates~\cite{bertini2019exact, gopalakrishnan2019unitary}. An interesting subfamily is given by 
\be
\label{eq:parameterisation}
U = (u_1 \otimes u_2) \cdot S \cdot e^{i J S_3\otimes S_3} \cdot (u_3 \otimes u_4),\qquad u_j\in U(d), 
\ee
where $S$ is the SWAP gate and $S_3$ is the $d$-dimensional representation of the third generator of $SU(2)$. This family is exhaustive in $d=2$ while for $d>2$ larger (but non-exhaustive) families are known~\cite{bertini2025exactly}.  

The BQCs built of dual-unitary gates --- called dual unitary circuits --- are generically strongly interacting and generate non-integrable dynamics (although they do have integrable points), nevertheless, many of their non-equilibrium properties can be characterised exactly. In this and the next lecture I will present a few of these exact results, however, for a comprehensive account I direct you to~\cite{bertini2025exactly}. The review also discusses how the dual-unitarity condition can be systematically relaxed while retaining some solvability --- this research direction is still ongoing and I will not discuss it here. 

Considering again the evolution of R\'enyi entropies, let us specialise the treatment to a dual-unitary circuit prepared in an initial state that is compatible with the spatial unitarity condition, i.e., that provides a boundary condition preserving unitarity. For example, consider a product of Bell pairs of the form 
\be
\label{eq:Bellpairstate}
\ket{\Psi_0}=\bigotimes_{x=1}^L \ket{\psi_0}, \qquad \qquad   \ket{\psi_0} \equiv \sum_{i=1}^d \frac{\ket{i}_x\otimes \ket{i}_{x+1/2}}{\sqrt d}. 
\ee
Representing it diagrammatically 
\be
\ket{\Psi_0}= \frac{1}{d^{L/2}} \fineq[1.6ex][0.8][1]{
\foreach \i in {0, ..., 3}{\bellpair[2*\i][0]}
\draw [decorate, thick, decoration = {brace}]   (6,0.1)--++(-6,.0);
\node[scale=1.25] at (3,-.25) {$L$};	
},
\ee
we see that from the point of view of the space evolution the state in Eq.~\eqref{eq:Bellpairstate} implements (unitarity conserving) open boundary conditions at $t=0$. Writing Eq.~\eqref{eq:tracerhonsimplified} for this state we have 
\be
\label{eq:tracerhonsimplified}
\hspace{-1cm} \tr[\rho^n_A(t)] \!=  \frac{1}{d^{n(|A|+2t+1)}} \hspace{-1.25cm}
\fineq[-0.8ex][.7][1]{
\draw[thick] (6.5,-1.25)--++(0,-.25);
\draw[thick] (-6.5,-1.25)--++(0,-.25);
\cstate[6.5][-1.25]
\cstate[-6.5][-1.25]
\foreach \j [evaluate=\j as \jplus using {\j+1}] in {1,...,2}
{	
\foreach \i in {-\j,...,\j}
{
\roundgate[2*\i][4-2*\j][1][bottomright][orange][n]
}
\foreach \i  [evaluate=\i as \ieval using {\i+.5}] in {-\jplus,...,\j}
{
\roundgate[2*\ieval][4-2*\j-1][1][bottomright][orange][n]
}
}
\foreach \i in {0,...,3}{
\cstate[\i-5.5][\i-.5]
\cstate[-\i+5.5][\i-.5]
}
\foreach \i in {-1,...,2}{
\sqrstate[\i-.5][2.5]
}
\foreach \i in {-3, ..., 3}{\bellpair[2*\i][-2]}
\draw [decorate, thick, decoration = {brace}]   (-1.5,2.75)--++(3,.0);
\node[scale=1.25] at (.5,3.25) {$2 |A|$};		
\draw [decorate, thick, decoration = {brace}]   (2.65,2.75)--++(3,-3);
\node[scale=1.25] at (4.5,1.5) {$2 t$};
} \hspace{-1cm}. \hspace{-1.25cm}
\ee
By repeated application of Eqs.~\eqref{eq:unitaritynfoldeddiagram} and \eqref{eq:spaceunitaritynfoldeddiagram} we then immediately find  
\begin{align}
\tr[\rho^n_A(t)] &= \frac{ \braket{\msqr}{\mcirc}^{\min(4t+2,2|A|)} \braket{\mcirc}{\mcirc}^{\max(0,|A|-2t-1)} \braket{\msqr}{\msqr}^{\max(0,2t+1-|A|)}}{d^{n (|A|+2t+1)}}\notag\\
&= d^{-2(n-1) (\min(2t+1,|A|))}\,,
\label{eq:npurity}
\end{align}
where in the second step we used $\braket{\mcirc}{\mcirc}=\braket{\msqr}{\msqr}=d^n$ and $\braket{\msqr}{\mcirc}=d$. Specialising this equation to $n=2$, we see that in dual unitary circuits the purity decays faster than in random unitary circuits because 
\be
 \left(\frac{2d}{d^2+1}\right)> \frac{1}{d}\,.
\ee
Moreover, plugging Eq.~\eqref{eq:npurity} into Eq.~\eqref{eq:renyientropies} we find 
\be
\label{eq:resultrenyi}
S^{(n)}_A(t)  = 2 \min(2t+1,|A|) \log d\,,\qquad  \forall n=2,3,\ldots,
\ee
and in fact this result can be extended to any $n\in\mathbb R^+$~\cite{bertini2025exactly}. The entanglement growth described by this equation, i.e., $2\log d$ per time step, is the maximal possible growth achievable in a BQC (this can be seen using the minimal cut bound~\cite{bertini2019entanglement}). In fact, \cite{zhou2022maximal} have shown the converse implication: if a gate generates the maximal increase of entanglement then it is dual-unitary. This is a manifestation of the special nature of dual-unitary circuits: because of their space-time swap symmetry correlations (and entanglement) propagate maximally fast in these systems.  

The exact calculation above can be repeated for all states preserving the unitarity of the space evolution, which have been introduced and classified by~\cite{piroli2020exact}. In particular, these states can always be written as ``two-site'' matrix-product states of the form  
\begin{align}
\!\!\ket{\Psi_0(\mathcal{M})} & = \frac{1}{d^{L/2}} \!\!\! \sum_{i_1,i_2,\ldots=1}^d \!\!\!\!\tra{\mathcal{M}^{i_1,i_2}\mathcal{M}^{i_3,i_4}\!\!\!\!\ldots}\!\!\ket{i_1,i_2,\ldots, i_{2L}}\notag\\
& =\frac{1}{d^{L/2}} \fineq[-0.8ex][1][1]{
\foreach \i in {0, ..., 2}{
\MPSinitialstate[2*\i][0][red6][topright]}
\draw[very thick] (0,-0.05) -- (4,-0.05);
\draw[very thick, dashed] (-1,-0.05) -- (5,-0.05);
},
\label{eq:twositeMPS}
\end{align}
where the matrix 
\be
\mel{ij}{\Gamma}{k\ell}= \fineq[-0.8ex][1][1]{
\MPSinitialstate[0][0][red6][topright]
\node at (-0.65,-0.2) {$i$};
\node at (0.65,-0.2) {$k$};
\node at (-0.65,.65) {$j$};
\node at (0.65,.65) {$\ell$};}
\ee
is unitary. Moreover, \cite{foligno2023growth} have shown that --- even when the initial state is not of this form --- for almost all dual-unitary circuits the entanglement growth eventually approaches the maximal rate at large times.

\chapter{}

In the previous lecture we saw that BQCs can be used to characterise universal properties of quantum many-body dynamics, in some cases even exactly. In this lecture we will discuss how to use BQCs to distinguish different kinds of quantum many-body dynamics based on their \emph{complexity}. 

The basic question is how to define complexity or \emph{chaos} in the quantum realm and whether one can establish in this setting something analogous to the ergodic hierarchy occurring in classical systems. This question is famously hard because of the inherent differences between quantum and classical mechanics. For instance, in quantum mechanics there are no trajectories --- ruling out the intuitive notion of chaos in terms of sensitivity to the initial conditions that we all heard about in pop science. In extreme summary, the many different approaches proposed to address this problem in the last fifty years can be organised in two main groups  
\begin{itemize}
\item[1.] ``Spectral'' (historical): Chaos is indirectly linked to some features of the system's spectrum (set of eigenvalues of the Hamiltonian\slash evolution operator): chaos $\iff$ spectrum similar to that of a random matrix. 
\item[2.] ``Dynamical'' (modern): Chaos is directly linked to features of the system's dynamics. The most utilised ones are the ability of the dynamics to \emph{scramble} quantum information (make local (in space) information increasingly more non-local) or to the \emph{computational complexity} of their classical simulation. 
\end{itemize}

Remarkably, BQCs have allowed us to make key progress in both these directions. In this final lecture I will discuss two examples, showing that in the case of dual-unitary circuits one can even obtain exact results. 

\section{Dynamical chaos via local operator entanglement}

Let us begin by discussing the dynamical approach. I this lecture I will consider a ``practical'' definition of dynamical quantum chaos based the hardness of classical simulations. 

As we saw in the previous lecture the entanglement of quantum states grows very quickly after quantum quenches in almost all scenarios. Therefore, it does not provide a good criterion to characterise the complexity of the dynamics. One can, however, ask what happens if instead of states we evolve operators, i.e., we adopt the Heisenberg picture. For example, let us consider a local operator $\mathcal O_{L/2}$ acting non-trivially only at position $x=L/2$ and evolve it in time 
\be
\mathcal O_{L/2}(t) = \mathbb U^{-t} \mathcal O_{L/2} \mathbb U^t,
\ee
and imagine to expand it in a basis of operators ``in product form'' containing $ \mathcal O$ --- for example for $d=2$ one can imagine to take $\mathcal O=X$ and expand in the Pauli basis. If the dynamics is particularly simple --- for example a BQC composed only of SWAP gates --- this expansion will only contain one term. Conversely, for more complicated dynamics one can expect that the number of terms in the expansion will grow very fast with time as the support of the operator grows. \cite{prosen2007is, prosen2007operator} proposed to quantify this process using an appropriate entanglement measure.

To define it let us look at this operator as a state of a larger (doubled) system --- formally this is achieved by an operator-to-state mapping acting on an operator basis as   
\be
\ketbra{n_1,\ldots,n_{2L}}{m_1,\ldots,m_{2L}} \longmapsto \ket{n_1,m_1,\ldots,n_{2L},m_{2L}}, 
\ee
where $\{\ket{n}\}_{n=0}^{d}$ is a given basis of $\mathbb C^{d}$ and the symbols for state placed next to each other represent tensor products $\ket{n,m}=\ket{n}\otimes \ket{m}$. In fact this mapping is nothing but the folding mapping discussed in the previous lecture (cf.~Sec.~\ref{sec:folding}) specialised to the one-replica case ($n=1$). Under this mapping the local operator of interest transforms as 
\be
\mathcal O_{L/2}(t) \longmapsto \ket*{\mathcal O_{L/2}(t)} = (\mathcal O_{L/2}(t) \otimes \mathds{1}) \ket{\mcirc}^{\otimes L},
\ee
where the state $\ket{\mcirc}^{\otimes L}$ is known as ``Choi state'' in quantum information theory. Diagrammatically we have 
\be
\hspace{-.5cm}\mathcal O_{L/2}(t) = \fineq[-0.8ex][.6][1]{
\foreach \j in {0,...,3}{
\foreach \i in {0,...,2}{
\roundgate[2*\i][2*\j][1][bottomright][blue6][-1]
\roundgate[2*\i+1][1+2*\j][1][bottomright][blue6][-1]}}
\foreach \j in {0,...,3}{
\foreach \i in {0,...,2}{
\roundgate[2*\i][-2*\j-1][1][topright][red6][-1]
\roundgate[2*\i+1][-1-2*\j-1][1][topright][red6][-1]}}
 \draw[thick, fill = black]  (2.5,-.5) circle (2pt);
\node[scale=1.25] at (2.85,-0.5) {$\mathcal O$};
}
\hspace{-.5cm}\longmapsto 
\ket*{\mathcal O_{L/2}(t)} = \hspace{-.125cm} \fineq[-0.8ex][.6][1]{
\foreach \j in {0,...,3}{
\foreach \i in {0,...,2}{
\roundgate[2*\i][2*\j][1][bottomright][orange][1]
\roundgate[2*\i+1][1+2*\j][1][bottomright][orange][1]}}
\foreach \i in {0,...,5}{
\cstate[\i-.5][-.5]}
 \draw[thick, fill = black]  (2.5,-.5) circle (2.5pt);
\node[scale=1.25] at (2.75,-0.75) {$\mathcal O$};
}\,, 
\ee
where we introduced the symbol 
\be
 \fineq[-0.8ex][.65][1]{
\draw[thick] (0,0)--++(0,-.5);
\cstate[0][0]
 \draw[thick, fill = black]  (0,0) circle (2.5pt);
\node[scale=1.25] at (.5,0) {$\mathcal O$};
}=\ket{\mathcal O} = {\mathcal O}\otimes \mathds{1} \ket{\mcirc},
\label{eq:foldedlocaloperator}
\ee
to denote the folded version of $\mathcal O$. 

Now that we have $\ket*{\mathcal O_{L/2}(t)}$, following \cite{prosen2007is, prosen2007operator}, we compute its entanglement over an arbitrary (but contiguous) spatial bipartition $A\cup \bar A$. For instance, proceeding as in the previous lecture the second R\'enyi entropy can be expressed as
\be
S^{(2)}_{\mathcal O, A}(t) = -\log\tr[\rho^2_{\mathcal O, A}(t)],
\ee
where we introduced the operatorial density matrix
\be
\hspace{-1cm}\rho_{\mathcal O, A}(t) = \fineq[-0.8ex][.6][1]{
\foreach \j in {0,...,3}{
\foreach \i in {0,...,2}{
\roundgate[2*\i][2*\j+2.5][1][bottomright][orange][1]
\roundgate[2*\i+1][1+2*\j+2.5][1][bottomright][orange][1]}}
\foreach \j in {0,...,3}{
\foreach \i in {0,...,2}{
\begin{scope}[rotate around={180:(2*\i,-2*\j-1)}]
\roundgate[2*\i][-2*\j-1][1][bottomright][orange][1]
\end{scope}
\begin{scope}[rotate around={180:(2*\i+1,-1-2*\j-1)}]
\roundgate[2*\i+1][-1-2*\j-1][1][bottomright][orange][1]
\end{scope}}}
\foreach \i in {0,...,5}{
\cstate[\i-.5][2]}
 \draw[thick, fill = black]  (2.5,2) circle (2.5pt);
\node[scale=1.25] at (2.75,1.75) {$\mathcal O$};
\foreach \i in {0,...,5}{
\cstate[\i-.5][-0.5]}
 \draw[thick, fill = black]  (2.5,-0.5) circle (2.5pt);
\node[scale=1.25] at (2.75,-0.2) {$\mathcal O$};
\foreach \i in {0,...,4}{	
\draw[thick] (-2.5+\i+3,10) -- (-2.5+\i+3,10.5+\i*0.25) -- (-3.5-\i*0.25+2.5,10.5+\i*0.25) -- (-3.5-\i*0.25+2.5, -9-\i*0.25) -- (-2.5+\i+3,-9-\i*0.25) -- (-2.5+\i+3, 1-9.5);}
}
\hspace{-1cm}\longmapsto 
\ket*{\rho_{\mathcal O, A}(t)} = \hspace{-.125cm} \fineq[-0.8ex][.6][1]{
\foreach \j in {0,...,3}{
\foreach \i in {0,...,2}{
\roundgate[2*\i][2*\j][1][bottomright][orange][2]
\roundgate[2*\i+1][1+2*\j][1][bottomright][orange][2]}}
\foreach \i in {0,...,5}{
\cstate[\i-.5][-.5]}
 \draw[thick, fill = black]  (2.5,-.5) circle (2.5pt);
\node[scale=1.25] at (2.75,-0.75) {$\mathcal O$};
\foreach \i in {-1,...,2}{
\sqrstate[\i+1.5][7.5]
}
}\,. 
\ee
Folding again, we can represent the operator purity as
\be
\tr[\rho^2_{\mathcal O, A}(t)] = \frac{1}{d^{4L}}\fineq[-0.8ex][.7][1]{
\foreach \j in {0,...,3}{	
\foreach \i in {-1,...,3}{
\roundgate[2*\i][2*\j+1][1][bottomright][orange][4]
\roundgate[2*\i+1][2*\j][1][bottomright][orange][4]}}
\foreach \i in {-1,1,3}{
\trigstater[\i-0.5][7.5]}
\foreach \i in {-2,0,2}{
\trigstatel[\i-0.5][7.5]}
\foreach \i in {4,..., 7}{
\crossstate[\i-0.5][7.5]}
\foreach \i in {-1,...,8}{
\cstate[\i-.5][-.5]}
 \draw[thick, fill = black]  (1.5,-.5) circle (2.5pt);
\node[scale=1.25] at (1.2,-0.75) {$\mathcal O$};
}\,.
\ee
In the diagram above I introduced two special paring states 
\begin{align}
\ket{\mtrig} &= \sum_{i_1,\ldots,i_4 =1}^d \ket{i_1,i_{2},i_{2},i_{1},i_3,i_{4},i_{4},i_{3}} \equiv \fineq{
\draw[thick] (0,0)--++(0,-0.5);\trigstate[0][0]},\\
\ket{\mcross} & = \sum_{i_1,\ldots,i_4 =1}^d \ket{i_1,i_{2},i_{3},i_{4},i_4,i_{3},i_{2},i_{1}} \equiv \fineq{
\draw[thick] (0,0)--++(0,-0.5);\crossstate[0][0]}, 
\end{align}
and used the diagram in Eq.~\eqref{eq:foldedlocaloperator} to denote the replicated version of the folded local operator (${\mathcal O}\otimes \mathds{1}\mapsto ({\mathcal O}\otimes \mathds{1} \otimes {\mathcal O}^\dag \otimes \mathds{1} )^{\otimes n=2}$). 

The entanglement of $\ket*{\mathcal O_{L/2}(t)}$ is often succinctly referred to as local operator entanglement (LOE). Besides quantifying the complexity of implementing the Heisenberg evolution of local operators with tensor network methods, LOE is indeed believed to distinguish integrable and non-integrable dynamics. All the examples that we could characterise so far have shown (or suggested) that for integrable systems the growth of LOE entropies is sub-linear (bounded by a term $\propto \log t$) while for non-integrable systems is linear. 

As we will now see, BQCs (and especially dual-unitary ones) can be used to provide analytical support to these statements. 

\subsection{Local operator entanglement in dual-unitary circuits}

Let us now embark on a more explicit calculation of local operator entanglement. For simplicity, consider the second R\'enyi entropy of LOE (the discussion, however, generalises directly to all $n\geq 2$). 

We begin by using the unitarity relations (cf.~Eq.~\eqref{eq:unitaritynfoldeddiagram}) fulfilled by the local gates to reduce the diagram for the operator purity to the causal light cone. Conveniently rotating it by 45 degrees we have 
\be
\label{eq:lightcone}
\tr[\rho^2_{\mathcal O, A}(t)]  =\frac{1}{d^{2t}} \fineq[-2.4ex][.7][1]{
\begin{scope}[rotate around={-45:(0,0)}]
\foreach \i in {0,...,4}{
\foreach \j in {2,...,4}{
\roundgate[\i+\j][\i-\j][1][bottomright][orange][4]
}
}
\foreach \i in {0,...,4}{
\trigstatel[\i+1.5][\i-1.5]	
\cstate[\i+4.5][\i-4.5]	
}
\foreach \i in {2,...,4}{
\cstate[\i-.5][-\i-.5]	
\crossstate[\i+4.5][-\i+4.5]
}
\charge[.5+4][-0.5-4][black]
\end{scope}
\draw [decorate, thick, decoration = {brace}]   (-0.35,-2)--++(6.25,0);
\node[scale=1.5] at (2.75,-1.5) {$t+x$};
\draw [decorate, thick, decoration = {brace}]   (6.55,-2.5)--++(0,-3.35);
\node[scale=1.5] at (7.5,-4) {$t-x$};
\node[scale=1.5] at (.5,-6.5){$\mathcal O$};
},
\ee 
where I denoted by $x$ the distance between $L/2$ (which we assumed to be integer) and last site belonging to the subsystem $A$. Note that --- differently from what happened in last lecture's  state entanglement calculation (see Sec.~\ref{sec:DUentanglement}) --- even when specialising the treatment to dual-unitary circuits, and hence acquiring the additional relations in Eq.~\eqref{eq:spaceunitaritynfoldeddiagram}, the diagram cannot be directly contracted. We should have expected this as the dual-unitarity relations are fulfilled both by integrable (e.g.\ the SWAP gate) and non-integrable gates --- see the parameterisation in Eq.~\eqref{eq:parameterisation} --- while we believe the operator purity to behave differently in the two cases.

To make progress, one can note that the diagram above can be expressed as follows
\begin{align}
\tr[\rho^2_{\mathcal O, A}(t)]  &= \mel*{\underbrace{\mcirc \ldots \mcirc}_{x_-} }{R_{\bcirc,x_-} (R_{x_-})^{x_+-1}}{\underbrace{\mcross \ldots \mcross}_{x_-}} \\
&= \mel*{\bcirc\underbrace{\mcirc \ldots \mcirc}_{x_+-1}}{(L_{x_+})^{x_-}}{\underbrace{\mtrig \ldots \mtrig}_{x_+} },
\label{eq:rhoopTM}
\end{align}
where I introduced the ``light cone coordinates'' $x_\pm = t \pm x$ and the transfer matrices 
\begin{align}
R_x &= 
\fineq[-2.4ex][.65][1]{
\begin{scope}[rotate around={45:(0,0)}]
\foreach \j in {0,...,4}{
\roundgate[\j][-\j][1][bottomright][orange][4]
}
\trigstatel[-.5][.5]	
\cstate[4.5][-4.5]	
\end{scope}
\draw [decorate, thick, decoration = {brace}]   (-0.15,.75)--++(5.9,0);
\node[scale=1.5] at (2.8,1.15) {$x$};
},\\
R_{\bcirc,x} &= 
\fineq[-2.4ex][.65][1]{
\begin{scope}[rotate around={45:(0,0)}]
\foreach \j in {0,...,4}{
\roundgate[\j][-\j][1][bottomright][orange][4]
}
\trigstatel[-.5][.5]	
\cstate[4.5][-4.5][][black]	
\end{scope}
\draw [decorate, thick, decoration = {brace}]   (-0.15,.75)--++(5.9,0);
\node[scale=1.5] at (2.8,1.15) {$x$};
\node[scale=1.25] at (6.75,0) {$\mathcal O$};
}\!\!\!,\\
L_x &= 
\fineq[-2.4ex][.65][1]{
\begin{scope}[rotate around={-45:(0,0)}]
\foreach \j in {0,...,4}{
\roundgate[\j][\j][1][bottomright][orange][4]
}
\cstate[-.5][-.5]	
\crossstate[4.5][4.5]	
\end{scope}
\draw [decorate, thick, decoration = {brace}]   (-0.15,.75)--++(5.9,0);
\node[scale=1.5] at (2.8,1.15) {$x$};
}.
\end{align}

I now state and prove two simple lemmas that help characterising these matrices
\begin{lemma}
For all unitary local gates $\{U^{(x)}\}$ the transfer matrix $T_x=R_x, R_{\bcirc,x}, L_x$ fulfils 
\be
\| T_x\| \leq d^2. 
\ee
\end{lemma}
\begin{proof}
Without loss of generality I prove it for $R_x$. First note that, by appropriate unfolding, the matrix can be written as 
\be
R_x = 
\fineq[-2.4ex][.65][1]{
\begin{scope}[rotate around={45:(0,0)}]
\foreach \j in {0,...,4}{
\roundgate[\j][-\j][1][bottomright][orange][2]
}
\foreach \j in {5,...,9}{
\begin{scope}[rotate around={180:(\j,-\j)}]
\roundgate[\j][-\j][1][bottomright][green1][2]
\end{scope}
}
\cstate[-.5][.5]	
\cstate[9.5][-9.5]	
\end{scope}
\draw [decorate, thick, decoration = {brace}]   (-0.15,.75)--++(5.9,0);
\node[scale=1.5] at (2.8,1.15) {$x$};
\draw [decorate, thick, decoration = {brace}]   (7-0.15,.75)--++(5.9,0);
\node[scale=1.5] at (7+2.8,1.15) {$x$};
},
\ee
where in green I depicted the hermitian conjugate of the gate, i.e.,  
\be
\fineq{\roundgate[0][0][1][topright][green1][n]} =\fineq{
\foreach \i in {0,...,-2}
{\roundgate[\i*0.3][\i*0.15][1][topright][red6]
\roundgate[\i*0.3-.15][\i*0.15-.075][1][topright][blue6]
}
\draw [decorate, decoration = {brace}]   (-1.35,0.1)--++(1*0.9,.5*0.9);
\node[scale=1.25] at (-1.135,.5) {${}_{2n}$};
}= \left(U^*\otimes U\right)^{\otimes n}.
\label{eq:conjfoldedgatepicture}
\ee
Next, using the sub-multiplicativity of the operator norm and the unitarity of the gates we have
\be
\|R_x\| \leq \|\fineq[-0.8ex][.65][1]{
\begin{scope}[rotate around={0:(0,0)}]
\roundgate[0][0][1][bottomright][orange][2]
\end{scope}
}\|^x \|\fineq[-0.8ex][.65][1]{
\begin{scope}[rotate around={180:(0,0)}]
\roundgate[0][0][1][bottomright][green1][2]
\end{scope}
}\|^x \braket{\mcirc} = d^2.
\ee
\end{proof}

\begin{lemma}
\label{lemma:maximalev}
For dual-unitary local gates $\{U^{(x)}\}$ there are at least $x+1$ eigenvectors (left and right) of $R_x$ and $L_x$ corresponding to eigenvalue $d^2$. They are explicitly given by
\begin{align}
\{\ket{r_y} &=\ket*{\underbrace{\mtrig\cdots\mtrig}_{y}\underbrace{\mcirc\cdots\mcirc}_{x-y}}, \quad y=0,\ldots,x\}  & & \text{for $R_x$},\\
\{\ket{\ell_y} &=\ket*{\underbrace{\mcirc\cdots\mcirc}_{y}\underbrace{\mcross\cdots\mcross}_{x-y}}, \quad y=0,\ldots,x\} & & \text{for $L_x$}. 
\end{align}
When the local gates are not dual unitary $\ket{r_0}$ and $\ket{\ell_0}$ continue to be right eigenvectors and $\bra{r_x}$ and $\bra{\ell_x}$ left eigenvectors. 
\end{lemma}
\begin{proof}
The proof follows immediately from the (dual) unitarity relations in Eqs.~\eqref{eq:unitaritynfoldeddiagram} and~\eqref{eq:spaceunitaritynfoldeddiagram} specialised to $\ket{\sigma}=\ket{\mcirc}, \ket{\mtrig}, \ket{\mcross}$. 
\end{proof}

These lemmas imply that the leading eigenvalue of both $R_x$ and $L_x$ is $d^2$ for all unitary gates~\footnote{The operator norm bounds the spectral radius from above.}, and provide a characterisation of the corresponding eigenspace. Therefore, we can hope to use them to at least evaluate Eq.~\eqref{eq:rhoopTM} in the limit where one of $x_\pm$ is kept fixed and the other is large --- which means focussing on bipartitions $A\bar A$ where the boundary between $A$ and $\bar A$ is close to either the left or right edge of the causal light cone. Indeed, this limit can be characterised by making the replacement 
\be
\label{eq:leadingpartofTM}
T^x \longmapsto d^{2x} \sum_{y} \ketbra{v_y}{w_y}, 
\ee
where $\{v_y\}$ and $\{w_y\}$ are bi-orthonormal bases of the leading eigenvalue eigenspace of $T=R,L$.  

The question then is whether the eigenvectors identified in Lemma~\ref{lemma:maximalev} exhaust the leading-eigenvalue eigenspace. They do not do so in all cases: for example, when all the local gates are SWAPs the matrices $R_x$ and $L_x$ become proportional to the identity matrix that has $d^{8x}$ ``leading'' eigenvectors. More generally, a similar exponential scaling (but with a smaller exponent) is expected in the case of integrable dual-unitary circuits with one-local charge densities~\cite{foligno2025nonequilibrium}. However, for generic enough dual unitary circuits --- for example when $U$ is picked at random from the family in Eq.~\eqref{eq:parameterisation} --- numerics indicates that the family of Lemma~\ref{lemma:maximalev} is indeed exhaustive. This statement is still unproven but I suspect that a proof can be found reasoning along the lines of \cite{foligno2024entanglement}, where it is shown that for randomly picked dual unitary gates the matrix 
\be
\label{eq:TMn1}
\frac{1}{d} \fineq[-2.4ex][.65][1]{
\begin{scope}[rotate around={45:(0,0)}]
\foreach \j in {0,...,4}{
\roundgate[\j][-\j][1][bottomright][orange][1]
}
\cstate[-.5][.5]	
\cstate[4.5][-4.5]	
\end{scope}
\draw [decorate, thick, decoration = {brace}]   (-0.15,.75)--++(5.9,0);
\node[scale=1.5] at (2.8,1.15) {$x$};
},
\ee
has a unique leading eigenvalue for all $x$. 

The above observations motivate us to formulate the following definition. 
\begin{definition}
A family of dual-unitary local gates $\{U^{(x)}\}$ is \emph{completely chaotic} if the only eigenvectors of $R_x$ and $L_x$ corresponding to eigenvalue $d^2$ are those listed in Lemma~\ref{lemma:maximalev}. The corresponding circuit is a \emph{completely chaotic dual-unitary circuit}. 
\end{definition}

Using Eq.~\eqref{eq:leadingpartofTM} we then have that for completely chaotic dual-unitary circuits the operator purity fulfils 
\begin{align}
\lim_{x_-\to\infty}\tr[\rho^2_{\mathcal O, A}(t)] & = \frac{1}{d^{2 x_+}}\sum_{y=0}^{x_+} \braket*{\bcirc\underbrace{\mcirc \cdots \mcirc}_{x_+-1}}{\tilde \ell_y}\braket*{\tilde\ell_y}{\underbrace{\mtrig \cdots \mtrig}_{x_+}} \\
& = \frac{d^{2-2x_+}}{d^2-1}\,,
\label{eq:LOEleftlightcone}
\end{align}
where I introduced the orthonormalised family of vectors 
\be
\label{eq:orthonormalfamily}
\ket*{\tilde \ell_0} = \frac{1}{d^{2x}} \ket*{ \ell_0}, \qquad \ket*{\tilde \ell_{y\geq 1}} =  \frac{d^2 \ket*{ \ell_y}-\ket*{ \ell_{y-1}}}{d^{2x} \sqrt{d^2-1}}, 
\ee
used $\braket*{\bcirc{\mcirc \cdots \mcirc}}{\tilde \ell_{y}}=\delta_{y,0}$, and $\braket*{\tilde \ell_0}{\mtrig \cdots \mtrig}=1$. 

The result in Eq.~\eqref{eq:LOEleftlightcone} shows that --- up to the order one contribution $d^2/(d^2-1)$ --- the purity is minimal in this limit. In other words, the operator state $\ket*{\mathcal O_{L/2}(t)}$ almost maximally entangled if one takes a bipartition that is close to the left edge of the causal light cone (see Fig.~\ref{fig:LOElightcone}). The fact that the entanglement is not exactly maximal can be understood from the behaviour of dynamical two-point functions in dual-unitary circuits~\cite{bertini2025exactly}. Indeed, for an operator prepared in an integer site, the correlation is non-zero only if computed exactly at the left edge of the causal light cone. This means that expanding $\mathcal O_{L/2}(t)$ in a product-operator basis there can be no contribution with the identity operator at the left edge of the light cone. This constraint limits the entanglement (increasing the purity by a factor $d^2/(d^2-1)$). I also remark that if one uses the same logic --- assuming that the eigenvector of Lemma~\ref{lemma:maximalev} are exhaustive --- also for non-dual-unitary circuits, one immediately has $\tr\smash{[\rho^2_{\mathcal O, A}(t)]}=1$. This result is boring but makes sense: for non-dual-unitary circuits the spreading of correlations is slower and the regions close to the light cone edge are not entangled. To find a non-trivial result in this case one would need to place the cut deeper into the light cone.  

An analogous (but slightly more involved) calculation gives the following result for the other edge of the light cone 
\begin{align}
\lim_{x_+\to\infty}&\tr[\rho^2_{\mathcal O, A}(t)]  = \frac{1}{d^{2 x_-}}\sum_{y=0}^{x_-} \mel*{\underbrace{\mcirc \ldots \mcirc}_{x_-} }{R_{\bcirc,x_-}}{\tilde r_y}\braket*{\tilde r_y}{\underbrace{\mcross \ldots \mcross}_{x_-}} \\
& = \sum_{y=1}^{x_-} (c_k(\mathcal O)^2-c_{k-1}(\mathcal O)^2) d^{2-2k}+ c_0(\mathcal O)^2\,,
\label{eq:LOErightlightcone}
\end{align}
where I again introduced the orthonormalised family of eigenvectors (same definition as in Eq.~\eqref{eq:orthonormalfamily} with $\ell \mapsto r$) and set 
\be
\hspace{-0.25cm}c_k (\mathcal O) = \frac{1}{d^{2(x-k)}}
\fineq[-2.4ex][.65][1]{
\begin{scope}[rotate around={45:(0,0)}]
\foreach \j in {0,...,3}{
\roundgate[\j][-\j][1][bottomright][orange][1]
}	
\foreach \j in {4,...,7}{
\begin{scope}[rotate around={180:(\j,-\j)}]
\roundgate[\j][-\j][1][bottomright][green1][1]
\end{scope}
}
\cstate[-.5][.5][][black]	
\cstate[7.5][-7.5][][black]	
\end{scope}
\draw [decorate, thick, decoration = {brace}]   (-0.15,.9)--++(4.5,0);
\node[scale=1.5] at (2,1.3) {$x-k$};
\draw [decorate, thick, decoration = {brace}]   (5.7-0.15,.9)--++(4.5,0);
\node[scale=1.5] at (5.7+2,1.3) {$x-k$};
\node[scale=1.25] at (10.75,0.5) {$\mathcal O$};
\node[scale=1.25] at (-.75,0.5) {$\mathcal O$};
\foreach \i in{0,..., 7}{
\cstate[1.41421*\i][.65]	
\cstate[1.41421*\i][-.65]}
}.
\ee
When $x_-$ becomes large this result simplifies further to
\be
\label{eq:LOErightlightconeasy}
\lim_{x_+\to\infty} \tr[\rho^2_{\mathcal O, A}(t)]  \simeq 
\begin{cases}
|\lambda|^{4 x_-} & |\lambda| \geq d^{-1/2} \\
d^{-2 x_-} & |\lambda| \leq d^{-1/2}
\end{cases},
\ee
where $\lambda$ is the leading sub-leading eigenvalue of the matrix in Eq.~\eqref{eq:TMn1} for $x=1$. Interestingly, in this case the state is maximally mixed only if $|\lambda|$ is small enough. This phenomenon can again be connected to the behaviour of two-point functions, indeed, it turns out that $|\lambda|$ is the factor controlling their exponential decay~\cite{bertini2025exactly}. Therefore, Eq.~\eqref{eq:LOErightlightcone} implies that for to cuts close to the right light-cone edge $\ket*{\mathcal O_{L/2}(t)}$ is maximally mixed only when correlations decay sufficiently fast (cf.~Fig.~\ref{fig:LOElightcone}). 

\begin{figure} [t]
\centering
    \begin{tikzpicture}[baseline={([yshift=-0.6ex]current bounding box.center)},scale=1]
      \draw[thick,black,->] (0,0) -- (0,4) node at (0.5,3.75) {$t$}; 
       \draw[color=gray, fill = gray!60, opacity=0.5, thick, domain = -3.5: 3.5]    plot (\x,{abs(\x)});
     \draw[thick,black,->] (-4,0) -- (4,0)  node at (3.5,-.5) {$x$};
     \node at (0,-.5) {$\mathcal O$};
      \draw[color=gray, fill = gray] (2,3.5) -- (3.5,3.5) -- (3.4,3.4) -- (2,3.4) -- cycle; 
      \draw[color=gray, fill = gray] (-2,3.5) -- (-3.5,3.5) -- (-3.4,3.4) -- (-2,3.4) -- cycle; 
      \draw[thick,black,->] (-3.75,2.7) -- (-3.25,3.3); 
      \draw[thick,black,->] (3.75,2.7) -- (3.25,3.3); 
      \node at (-4.25,2.5) {maximally mixed};
      \node at (4,2.5) {generically not}; 
      \node at (4,2) {maximally mixed};
  \end{tikzpicture}
  \caption{Causal light cone.}
  \label{fig:LOElightcone}
\end{figure}
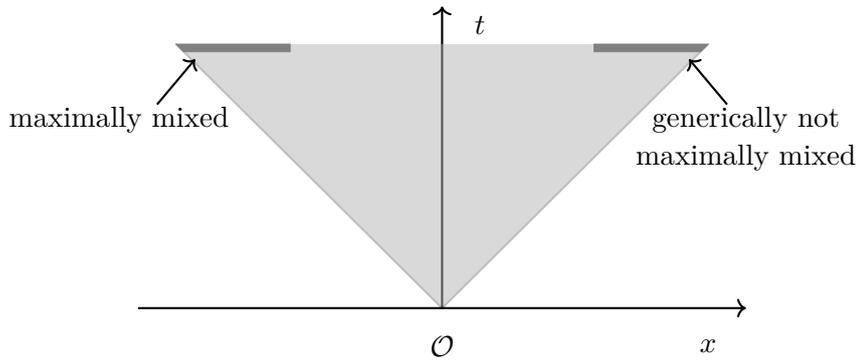

Putting all together, we that the R\'enyi-2 entropy of LOE fulfils 
\be
S^{(2)}_{\mathcal O, A}(t) \simeq 
\begin{cases}
2x_+ \log d & x_-\to\infty \\
4 x_- \log(1/|\lambda|) &x_+\to\infty\,\, \&\,\, \lambda \geq d^{-1/2}\\
2x_- \log d &x_+\to\infty\,\, \&\,\,\lambda \leq d^{-1/2}
\end{cases}.
\ee
Namely it has volume law scaling in the limits considered. In fact, \cite{bertini2020operatori} have argued that Eqs.~\eqref{eq:LOEleftlightcone} and~\eqref{eq:LOErightlightconeasy} describe the leading order scaling of $S^{(2)}_{\mathcal O, A}(t)$ also away from these limits. More precisely, one has  
\be
S^{(2)}_{\mathcal O, A}(t)  \simeq -\log[{\rm Eq.}~\eqref{eq:LOEleftlightcone}+{\rm Eq.}~\eqref{eq:LOErightlightconeasy}] = O(t). 
\ee
This volume law scaling should be compared to the result for integrable dual-unitary circuits, for which \cite{bertini2020operatorii} have shown that $S^{(2)}_{\mathcal O, A}(t)$ is bounded by a constant. 

An analogous picture is observed for all LOE R\'enyi entropies with $n\geq 2$. The only difference is that the change of behaviour in Eq.~\eqref{eq:LOErightlightconeasy} occurs for $\lambda=d^{(1-n)/n}$. This means that the LOE entanglement entropy ($n\to1$) is always maximal~\cite{bertini2020operatori}.

\section{Spectral chaos via spectral form factor}

The first investigations of ``quantum chaos'' date back to the end of the 1970s with an intense research effort continuing throughout 1980s. In those days, by studying numerous specific examples, researchers observed that, in the semiclassical limit, few particle systems that are classically chaotic show strong correlations among the eigenvalues of their evolution operator. These correlations were seen to coincide with those occurring in a random matrix with the same anti-unitary symmetries (e.g.\ time-reversal). At the same time, it was noted that the semiclassical spectrum of few particle classically integrable models showed no correlations, i.e., it was a Poisson process. These behaviours were then taken to be the defining feature of integrability and chaoticity in the quantum realm~\cite{casati1980on, bohigas1984characterisation, berry1977level} and their occurrence in the semiclassical regime was partially explained by means of periodic orbit theory~\cite{muller2004semiclassical, muller2004semiclassical}. 

With the improvement of computers, however, numerical investigations started to show that such a sharp difference in the spectral correlations is also present between integrable and non-integrable quantum many-body systems far from any semiclassical limit --- like spin-$1/2$ chains. In this case the origin of this difference remained long unexplained. As a final part of this lecture I want to show how BQCs allow us to resolve the impasse and give an analytical characterisation to this behaviour.

I will consider a simple measure of ``spectral correlations'' --- known as \emph{spectral form factor} (SFF) --- which is defined as  
\be
K(t) =  \mathbb E\left[\sum_{j,j'=1}^{\mathcal N} e^{i (\varphi_j-\varphi_{j'}) t}\right]= \mathbb E\!\left[\,|{\rm tr}\, \mathbb U^t|^2\,\right]\,,
\label{eq:SFF}
\ee
where $\mathcal N$ is the size of the evolution operator (dimension of the Hilbert space). The phases $\{\varphi_j\}$ appearing in the above equation are the eigenvalues of the time evolution operator $\mathbb U$ (sometimes called quasi-energies), while the average $\mathbb E[\cdot]$ is over an ensemble of similar systems, or a moving time average. The explicit average is needed to extract the universal part of the SFF as the latter is not self-averaging~\cite{prange1997the}. 

\begin{figure}[tb!]
\centering
  \begin{tikzpicture}[baseline={([yshift=-0.6ex]current bounding box.center)},scale=1]
      \draw[thick,black,->] (0,-0.25) -- (0,5) node at (-0.75,4.5) {$K(t)$}; 
      \draw[black, dashed] (3.5,-0.1) -- (3.5,3.5); 
       \draw[color=green1, thick, domain = 0: 10]    plot (\x,3.5);
       \draw[color=red, thick, domain = 0: 3.5]    plot (\x,{\x});
       \draw[color=red, thick, domain = 3.5: 10]    plot (\x,3.5);
        \draw[color=blue1, thick, domain = 0: 1.25, samples=100]    plot (\x,{2*\x});
        \draw[color=blue1, thick, domain = 1.25: 10, samples=100]    plot (\x,{3.5-(1.25/\x)^2});
     \draw[thick,black,->] (-.5,0) -- (10,0)  node at (9.5,-.25) {$t$};
      \draw[thick,black,->] (2.5,1.5) -- (2,2); 
      \draw[thick,black,->] (2,3.85) -- (2,3.55); 
       \draw[thick,black,->] (4.65,2.85) -- (4.65,3.4); 
      \node at (2.65,1.25) {CUE}; 
      \node at (4.65,2.65) {COE}; 
      \node at (2,4.15) {Poisson}; 
      \node at (3.5,-.25) {$\mathcal N$}; 
      \node at (-.25,3.5) {$\mathcal N$}; 
  \end{tikzpicture}
\caption{Schematic behaviour of the SFF for $\mathcal N \times \mathcal N$ random unitary matrices in two of the circular ensembles (COE/CUE) and a collection of $\mathcal N$ uncorrelated levels (Poisson).
\label{fig:SFF}}
\end{figure}
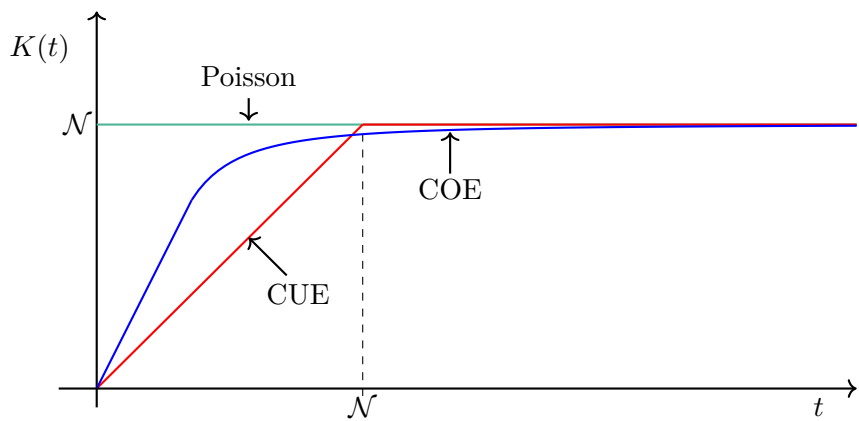

The SFF can be directly computed when the average is taken over the standard ensembles of random unitary matrices~\cite{mehta2004random}. In particular, I report three relevant examples in Fig.~\ref{fig:SFF}. The red curve corresponds to the SFF averaged over random unitary matrices distributed according to the flat Haar measure --- this ensemble is not time reversal invariant and is known as \emph{circular unitary ensemble} (CUE). For the blue curve the average is taken over symmetric unitary matrices $U^T U$, with $U$ distributed according to the Haar measure --- this is a time reversal invariant ensemble known as circular orthogonal ensemble (COE). Finally, the green line corresponds to an average over uncorrelated quasi-energies (Poisson process). The main feature signalling spectral correlations is the presence of the initial slope (``the ramp'') in the first two curves --- this is a hallmark of ``level repulsion''. Specifically, for large $\mathcal N$ we have 
\be
K(t)=
\begin{cases}
t + O(t/\mathcal N),   &{\rm CUE}    \\
2t + O(t/\mathcal N), &{\rm COE}    \\
\mathcal N,  &{\rm Poisson}
\end{cases}. 
\ee

To compare with this result, let us now compute the SFF in a BQC. Considering a $\mathbb U$ of the form in Eq.~\eqref{eq:floquetoperator} we can write Eq.~\eqref{eq:SFF} as  
\begin{align}
\!\!\!\!K(t) \!=  \mathbb E\!\!\left[\fineq[-1.8ex][.7][1]{
\foreach \i in {1,...,5}{
\draw[thick, dotted] (2*\i+2-12.5+0.255,-1.75-0.1) -- (2*\i+2-12.5+0.255,4.25-0.1);
\draw[thick, dotted] (2*\i+2-11.5-0.255,-1.75-0.1) -- (2*\i+2-11.5-0.255,4.25-0.1);}
\foreach \i in {1,...,5}{
\draw[thick] (2*\i+2-11.5,4) arc (-45:175:0.15);
\draw[thick] (2*\i+2-11.5,-2) arc (315:180:0.15);
\draw[thick] (2*\i+2-0.5-12,-2) arc (-135:0:0.15);
}
\foreach \i in {2,...,6}{
\draw[thick] (2*\i+2-2.5-12,4) arc (225:0:0.15);
}
\foreach \j in {0,1,2}{
\foreach \i in{0.5}
{
\draw[ thick] (0.5,1+2*\i-0.5-3.5+2*\j) arc (-45:90:0.17);
\draw[ thick] (-10+0.5,1+2*\i-0.5-3.5+2*\j) arc (270:90:0.15);
}
\foreach \i in{1.5}{
\draw[thick] (0.5,2*\i-0.5-3.5+2*\j) arc (45:-90:0.17);
\draw[thick] (-10+0.5+0,2*\i-0.5-3.5+2*\j) arc (90:270:0.15);
}
\roundgate[0][-.5+2*\j][1][bottomright][yellow1][1]
\roundgate[-2][-.5+2*\j][1][bottomright][yellow2][1]
\roundgate[-4][-.5+2*\j][1][bottomright][yellow3][1]
\roundgate[-6][-.5+2*\j][1][bottomright][yellow4][1]
\roundgate[-8][-.5+2*\j][1][bottomright][yellow5][1]
\roundgate[-1][-1.5+2*\j][1][bottomright][yellow6][1]
\roundgate[-3][-1.5+2*\j][1][bottomright][yellow7][1]
\roundgate[-5][-1.5+2*\j][1][bottomright][yellow8][1]
\roundgate[-7][-1.5+2*\j][1][bottomright][yellow9][1]
\roundgate[-9][-1.5+2*\j][1][bottomright][yellow10][1]
}
\draw [decorate, thick, decoration = {brace}]   (-8.5,4.35)--++(9,0);
\node[scale=1.5] at (-4,5) {$2L$};
\draw [decorate, thick, decoration = {brace}]   (.75,4.25)--++(0,-6.25);
\node[scale=1.5] at (1.25,1.2) {$2t$};
}\right]\!\!\!, \label{eq:SFFdiag}
\end{align}
where I am again using the folded representation discussed in the previous lecture (Sec.~\ref{sec:folding}) and I conveniently defined the average $\mathbb E[\cdot]$ to be over \emph{random} local gates $\{U^{(x)}\}$ (not necessarily distributed according to the Haar measure) with two important properties
\begin{itemize}
\item[a.] Gates at different spatial sites are independent; 
\item[b.] Gates are correlated in time, i.e., each vertical column has \emph{the same matrix} (same shade of colour in the above diagram);   
\end{itemize}
The second property ensures that this random BQC has a well defined time-evolution operator. The first, instead, implies that the average factorises in space and we can write the SFF as 
\be
K(t) = \tr[\mathbb T^L], 
\ee
where I introduced the left-to-right operating space transfer matrix 
\be
\label{eq:TM}
\mathbb T = \mathbb E\!\!\left[\fineq[-0.6ex][.7][1]{
\foreach \i in {5}{
\draw[thick, dotted] (2*\i+2-12.5+0.255,-1.75-0.1) -- (2*\i+2-12.5+0.255,4.25-0.1);
\draw[thick, dotted] (2*\i+2-11.5-0.255,-1.75-0.1) -- (2*\i+2-11.5-0.255,4.25-0.1);}
\foreach \i in {5}{
\draw[thick] (2*\i+2-11.5,4) arc (-45:175:0.15);
\draw[thick] (2*\i+2-11.5,-2) arc (315:180:0.15);
\draw[thick] (2*\i+2-0.5-12,-2) arc (-135:0:0.15);
}
\foreach \i in {6}{
\draw[thick] (2*\i+2-2.5-12,4) arc (225:0:0.15);
}
\foreach \j in {0,1,2}{
\foreach \i in {0}{
\roundgate[-2*\i][-.5+2*\j][1][bottomright][yellow10][1]
}
\foreach \i in {1}
{
\roundgate[-2*\i+1][-1.5+2*\j][1][bottomright][yellow6][1]
}}
}\right]. 
\ee
The problem of computing the SFF is then mapped into that of characterising the spectrum of $\mathbb T$. This was done for different choices of random gates in the large $d$ limit (see, e.g.,~\cite{chan2018solution, chan2018spectral}) and exactly for dual unitary circuits~\cite{bertini2018exact, bertini2021random}. In the following I will briefly discuss the latter case.

\subsection{Spectral form factor in dual-unitary circuits}

In the case of dual unitary circuits one considers random local gates of the form  
\be
\fineq[-0.6ex][.85][1]{
 \tsfmatV[0][-0.5][r][1][][][red6]
  \circgate[-.35][.85][blue6][left]
  \circgate[-.35][0.15][blue7][right]	} = (u^{(x)} \otimes \mathds{1})\, W (w^{(x)} \otimes \mathds{1}), 
\label{eq:SFFDUgates}
\ee
where $W$ is an arbitrary dual-unitary gate parameterised as in Eq.~\eqref{eq:parameterisation}, while $u^{(x)}$ and $w^{(x)}$ are random and independently distributed according to some smooth, but arbitrary distribution in $U(q)$. The dual-unitary gate $W$ can be taken to be position dependent but, for simplicity, in the following I will consider the spatially homogeneous case. 

For this choice \cite{bertini2021random} proved the following Theorem. 
\begin{theorem}
For any (smooth) distribution of $u_x,w_x$ and any dual-unitary gate $W$ in Eq.~\eqref{eq:parameterisation} with $J\neq 0$ 
\be
\lim_{L\to\infty} K(t) = t\,.
\ee
\end{theorem}

This theorem guarantees that, in the thermodynamic limit, the dual-unitary BQC recovers the result obtained for non-T-symmetric random unitary matrices (CUE) of infinite size. The thermodynamic limit is considered to aid the analytical derivation (see below) while the fact that we recover the CUE result makes sense since the circuit defined by the gates in Eq.~\eqref{eq:SFFDUgates} is not T-symmetric. In fact, \cite{bertini2018exact} proved an analogous theorem for a particular example of T-symmetric dual-unitary circuit (the self-dual kicked Ising model). In that case $t \to 2t$ in agreement with the COE result. One can also constrain the setting described above to generate a $T$-symmetric circuit (cf.~\cite{bertini2021random}) but for this setting the analogue theorem is still unproven (see Remark 1.\ below).  

The proof of the above theorem is rather convoluted and I will not discuss it in detail here. However, I will conclude this lecture giving an idea of the proof strategy focussing on the case $d=2$. In essence, one should proceed in three main steps. \\

\noindent {\bf Step 1.} Bound the spectral radius~\footnote{The maximal magnitude of the eigenvalues.} of $\mathbb T$. This is done by writing the matrix in Eq.~\eqref{eq:TM} explicitly as 
\be
\mathbb T = (\tilde{\mathbb W} \otimes_r \tilde{\mathbb W}^*) \mathbb O^\dag (\Pi \tilde{\mathbb W} \Pi^\dag \otimes_r  \Pi^* \tilde{\mathbb W}^* \Pi^T)  \mathbb O^{\phantom{\dag}},
\label{eq:TMexplicit}
\ee
where $\Pi$ is the one-site shift operator of the $2t$-site temporal lattice, $\otimes_r$ is the tensor product between forward and backward time sheet and I set 
\be
\tilde{\mathbb{W}} = \tilde{W}^{\otimes t}, \qquad \mathbb O = \mathbb E\left[ (v \otimes w)^{\otimes t}\otimes_r (v^* \otimes w^*)^{\otimes t} \right]. 
\ee
Writing $v$ and $w$ in terms of the $SU(2)$ generators, $\mathbb O$ can be expressed explicitly as 
\be
\mathbb O \!=\!\int\!\!{\rm d}^{3}\boldsymbol{\theta}{\rm d}^{3}\boldsymbol{\varphi}\,f(\boldsymbol{\theta})g(\boldsymbol{\varphi})
e^{i\boldsymbol\theta\cdot (\boldsymbol{M}_{0}\otimes \mathds{1}-\mathds{1}\otimes \boldsymbol{M}^*_{0})} e^{i\boldsymbol\varphi\cdot (\boldsymbol{M}_{1}\otimes \mathds{1}-\mathds{1}\otimes \boldsymbol{M}^*_{1})}\!,
\ee
where I introduced the temporal sub-lattice magnetisations 
\begin{align}
[\boldsymbol M_{\iota}]_a \equiv  M_{a, \iota}=\sum_{\tau=0}^{t-1}\!\sigma^a_{\tau+\frac{\iota}{2}},
\end{align}
with $\{\sigma^a\}$ Pauli matrices, $\iota=0$ denoting to the integer sublattice, and $\iota=1$ the half-odd-integer sublattice. From the form in Eq.~\eqref{eq:TMexplicit} we can then conclude the spectral radius of $\mathbb T$ is bounded by one. \\

\noindent {\bf Step 2.} Reformulate the problem. Show that maximal-eigenvalue problem for $\mathbb T$, i.e., 
\be
\mathbb T\ket{A} =e^{i \alpha} \ket{A}, 
\label{eq:eigprob}
\ee
can be recast into a set of commutation identities. Namely, show that $\ket{A}\in \mathbb C^{d^{4t}}$ solving Eq.~\eqref{eq:eigprob} is in 1-1 correspondence with $A\in {\rm End}(\mathbb C^{d^{2t}})$ solving 
\be
[A, M_{a,\iota}]=0, \quad [A, M_{ab,\iota}]=0\,, \quad a,b\in\{x,y,z\}, \; \iota\in\{0, 1\}\,,
\ee
where I introduced  the temporal sub-lattice double magnetisations
\be
M_{ab,\iota}=\sum_{\tau=0}^{t-1}\!\sigma^a_{\tau+\frac{\iota}{2}}\sigma^b_{\tau+\frac{\iota+1}{2}}\,. 
\ee
In doing so one also shows that Eq.~\eqref{eq:eigprob} has solutions only for $\alpha=0$. \\

\noindent {\bf Third.} Show that $\{M_{a,\iota}, M_{ab,\iota}\}$ generate the full algebra of $2$-site shit invariant operators on the temporal lattice. Therefore, the only operators commuting with them are the powers of $\Pi^2$, i.e.   
\be
\Pi^{2n}, \qquad n=0,\ldots,t-1. 
\ee
This shows that we indeed have $t$ eigenvalues. \\ 

\noindent Let me conclude with two final remarks:
\begin{itemize}
\item[1.] One can repeat the first two steps for a T-symmetric circuit, showing that $M_{a,\iota}$ and $M_{ab,\iota}$ are replaced by $M_{a,\iota}+RM_{a,\iota}R$ and $M_{ab,\iota}+RM_{ab,\iota}R$, where $R$ is the operator implementing a reflection of the temporal lattice about its centre. The third step, however, is still open. 
\item[2.] A similar approach can be used to compute the higher moments of the SFF, i.e.
\be
K_n(t) = \mathbb E\!\left[\,|{\rm tr}\, \mathbb U^t|^{2n}\,\right].
\ee
which also show a universal behaviour in agreement with random matrix theory. In this case one needs to find the commutant of $\mathcal K^{\otimes n}$ (where $\mathcal K$ algebra generated by $\{M_{a,\iota}, M_{ab,\iota}\}$). This problem is also still open. 
\end{itemize}

\bibliographystyle{apacite}
\bibliography{./bibliography}

\begin{thebibliography}{}

\bibitem [\protect \citeauthoryear {%
Alba%
\ \BBA {} Calabrese%
}{%
Alba%
\ \BBA {} Calabrese%
}{%
{\protect \APACyear {2017}}%
}]{%
alba2017entanglement}
\APACinsertmetastar {%
alba2017entanglement}%
\begin{APACrefauthors}%
Alba, V.%
\BCBT {}\ \BBA {} Calabrese, P.%
\end{APACrefauthors}%
\unskip\
\newblock
\APACrefYearMonthDay{2017}{}{}.
\newblock
{\BBOQ}\APACrefatitle {Entanglement and thermodynamics after a quantum quench
  in integrable systems} {Entanglement and thermodynamics after a quantum
  quench in integrable systems}.{\BBCQ}
\newblock
\APACjournalVolNumPages{PNAS}{114}{30}{7947--7951}.
\newblock
\begin{APACrefDOI} \doi{10.1073/pnas.1703516114} \end{APACrefDOI}
\PrintBackRefs{\CurrentBib}

\bibitem [\protect \citeauthoryear {%
Anderson%
}{%
Anderson%
}{%
{\protect \APACyear {1972}}%
}]{%
anderson1972more}
\APACinsertmetastar {%
anderson1972more}%
\begin{APACrefauthors}%
Anderson, P\BPBI W.%
\end{APACrefauthors}%
\unskip\
\newblock
\APACrefYearMonthDay{1972}{}{}.
\newblock
{\BBOQ}\APACrefatitle {More Is Different} {More is different}.{\BBCQ}
\newblock
\APACjournalVolNumPages{Science}{177}{4047}{393-396}.
\newblock
\begin{APACrefURL}
  \url{https://www.science.org/doi/abs/10.1126/science.177.4047.393}
  \end{APACrefURL}
\newblock
\begin{APACrefDOI} \doi{10.1126/science.177.4047.393} \end{APACrefDOI}
\PrintBackRefs{\CurrentBib}

\bibitem [\protect \citeauthoryear {%
Berry%
\ \BBA {} Tabor%
}{%
Berry%
\ \BBA {} Tabor%
}{%
{\protect \APACyear {1977}}%
}]{%
berry1977level}
\APACinsertmetastar {%
berry1977level}%
\begin{APACrefauthors}%
Berry, M\BPBI V.%
\BCBT {}\ \BBA {} Tabor, M.%
\end{APACrefauthors}%
\unskip\
\newblock
\APACrefYearMonthDay{1977}{}{}.
\newblock
{\BBOQ}\APACrefatitle {Level clustering in the regular spectrum} {Level
  clustering in the regular spectrum}.{\BBCQ}
\newblock
\APACjournalVolNumPages{Proc R Soc Lond A Math Phys Sci}{356}{1686}{375--394}.
\newblock
\begin{APACrefDOI} \doi{10.1098/rspa.1977.0140} \end{APACrefDOI}
\PrintBackRefs{\CurrentBib}

\bibitem [\protect \citeauthoryear {%
Bertini%
, Claeys%
\BCBL {}\ \BBA {} Prosen%
}{%
Bertini%
\ \protect \BOthers {.}}{%
{\protect \APACyear {2025}}%
}]{%
bertini2025exactly}
\APACinsertmetastar {%
bertini2025exactly}%
\begin{APACrefauthors}%
Bertini, B.%
, Claeys, P\BPBI W.%
\BCBL {}\ \BBA {} Prosen, T.%
\end{APACrefauthors}%
\unskip\
\newblock
\APACrefYearMonthDay{2025}{}{}.
\newblock
\APACrefbtitle {Exactly solvable many-body dynamics from space-time duality.}
  {Exactly solvable many-body dynamics from space-time duality.}
\newblock
\begin{APACrefURL} \url{https://arxiv.org/abs/2505.11489} \end{APACrefURL}
\PrintBackRefs{\CurrentBib}

\bibitem [\protect \citeauthoryear {%
Bertini%
, Klobas%
, Alba%
, Lagnese%
\BCBL {}\ \BBA {} Calabrese%
}{%
Bertini%
\ \protect \BOthers {.}}{%
{\protect \APACyear {2022}}%
}]{%
bertini2022growth}
\APACinsertmetastar {%
bertini2022growth}%
\begin{APACrefauthors}%
Bertini, B.%
, Klobas, K.%
, Alba, V.%
, Lagnese, G.%
\BCBL {}\ \BBA {} Calabrese, P.%
\end{APACrefauthors}%
\unskip\
\newblock
\APACrefYearMonthDay{2022}{Jul}{}.
\newblock
{\BBOQ}\APACrefatitle {Growth of {R\'enyi} Entropies in Interacting Integrable
  Models and the Breakdown of the Quasiparticle Picture} {Growth of {R\'enyi}
  entropies in interacting integrable models and the breakdown of the
  quasiparticle picture}.{\BBCQ}
\newblock
\APACjournalVolNumPages{Phys. Rev. X}{12}{}{031016}.
\newblock
\begin{APACrefURL} \url{https://link.aps.org/doi/10.1103/PhysRevX.12.031016}
  \end{APACrefURL}
\newblock
\begin{APACrefDOI} \doi{10.1103/PhysRevX.12.031016} \end{APACrefDOI}
\PrintBackRefs{\CurrentBib}

\bibitem [\protect \citeauthoryear {%
Bertini%
, Kos%
\BCBL {}\ \BBA {} Prosen%
}{%
Bertini%
\ \protect \BOthers {.}}{%
{\protect \APACyear {2018}}%
}]{%
bertini2018exact}
\APACinsertmetastar {%
bertini2018exact}%
\begin{APACrefauthors}%
Bertini, B.%
, Kos, P.%
\BCBL {}\ \BBA {} Prosen, T.%
\end{APACrefauthors}%
\unskip\
\newblock
\APACrefYearMonthDay{2018}{}{}.
\newblock
{\BBOQ}\APACrefatitle {Exact {Spectral} {Form} {Factor} in a {Minimal} {Model}
  of {Many}-{Body} {Quantum} {Chaos}} {Exact {Spectral} {Form} {Factor} in a
  {Minimal} {Model} of {Many}-{Body} {Quantum} {Chaos}}.{\BBCQ}
\newblock
\APACjournalVolNumPages{Phys. Rev. Lett.}{121}{26}{264101}.
\newblock
\begin{APACrefDOI} \doi{10.1103/PhysRevLett.121.264101} \end{APACrefDOI}
\PrintBackRefs{\CurrentBib}

\bibitem [\protect \citeauthoryear {%
Bertini%
, Kos%
\BCBL {}\ \BBA {} Prosen%
}{%
Bertini%
\ \protect \BOthers {.}}{%
{\protect \APACyear {2019}}%
{\protect \APACexlab {{\protect \BCnt {1}}}}}]{%
bertini2019entanglement}
\APACinsertmetastar {%
bertini2019entanglement}%
\begin{APACrefauthors}%
Bertini, B.%
, Kos, P.%
\BCBL {}\ \BBA {} Prosen, T.%
\end{APACrefauthors}%
\unskip\
\newblock
\APACrefYearMonthDay{2019{\protect \BCnt {1}}}{May}{}.
\newblock
{\BBOQ}\APACrefatitle {Entanglement Spreading in a Minimal Model of Maximal
  Many-Body Quantum Chaos} {Entanglement spreading in a minimal model of
  maximal many-body quantum chaos}.{\BBCQ}
\newblock
\APACjournalVolNumPages{Phys. Rev. X}{9}{2}{021033}.
\newblock
\begin{APACrefDOI} \doi{10.1103/PhysRevX.9.021033} \end{APACrefDOI}
\PrintBackRefs{\CurrentBib}

\bibitem [\protect \citeauthoryear {%
Bertini%
, Kos%
\BCBL {}\ \BBA {} Prosen%
}{%
Bertini%
\ \protect \BOthers {.}}{%
{\protect \APACyear {2019}}%
{\protect \APACexlab {{\protect \BCnt {2}}}}}]{%
bertini2019exact}
\APACinsertmetastar {%
bertini2019exact}%
\begin{APACrefauthors}%
Bertini, B.%
, Kos, P.%
\BCBL {}\ \BBA {} Prosen, T.%
\end{APACrefauthors}%
\unskip\
\newblock
\APACrefYearMonthDay{2019{\protect \BCnt {2}}}{Nov}{}.
\newblock
{\BBOQ}\APACrefatitle {Exact Correlation Functions for Dual-Unitary Lattice
  Models in $1+1$ Dimensions} {Exact correlation functions for dual-unitary
  lattice models in $1+1$ dimensions}.{\BBCQ}
\newblock
\APACjournalVolNumPages{Phys. Rev. Lett.}{123}{}{210601}.
\newblock
\begin{APACrefDOI} \doi{10.1103/PhysRevLett.123.210601} \end{APACrefDOI}
\PrintBackRefs{\CurrentBib}

\bibitem [\protect \citeauthoryear {%
Bertini%
, Kos%
\BCBL {}\ \BBA {} Prosen%
}{%
Bertini%
\ \protect \BOthers {.}}{%
{\protect \APACyear {2020}}%
{\protect \APACexlab {{\protect \BCnt {1}}}}}]{%
bertini2020operatori}
\APACinsertmetastar {%
bertini2020operatori}%
\begin{APACrefauthors}%
Bertini, B.%
, Kos, P.%
\BCBL {}\ \BBA {} Prosen, T.%
\end{APACrefauthors}%
\unskip\
\newblock
\APACrefYearMonthDay{2020{\protect \BCnt {1}}}{}{}.
\newblock
{\BBOQ}\APACrefatitle {Operator {Entanglement} in {Local} {Quantum} {Circuits}
  {I}: {Chaotic} {Dual}-{Unitary} {Circuits}} {Operator {Entanglement} in
  {Local} {Quantum} {Circuits} {I}: {Chaotic} {Dual}-{Unitary}
  {Circuits}}.{\BBCQ}
\newblock
\APACjournalVolNumPages{SciPost Phys.}{8}{4}{067}.
\newblock
\begin{APACrefDOI} \doi{10.21468/SciPostPhys.8.4.067} \end{APACrefDOI}
\PrintBackRefs{\CurrentBib}

\bibitem [\protect \citeauthoryear {%
Bertini%
, Kos%
\BCBL {}\ \BBA {} Prosen%
}{%
Bertini%
\ \protect \BOthers {.}}{%
{\protect \APACyear {2020}}%
{\protect \APACexlab {{\protect \BCnt {2}}}}}]{%
bertini2020operatorii}
\APACinsertmetastar {%
bertini2020operatorii}%
\begin{APACrefauthors}%
Bertini, B.%
, Kos, P.%
\BCBL {}\ \BBA {} Prosen, T.%
\end{APACrefauthors}%
\unskip\
\newblock
\APACrefYearMonthDay{2020{\protect \BCnt {2}}}{}{}.
\newblock
{\BBOQ}\APACrefatitle {Operator {Entanglement} in {Local} {Quantum} {Circuits}
  {II}: {Solitons} in {Chains} of {Qubits}} {Operator {Entanglement} in {Local}
  {Quantum} {Circuits} {II}: {Solitons} in {Chains} of {Qubits}}.{\BBCQ}
\newblock
\APACjournalVolNumPages{SciPost Phys.}{8}{4}{068}.
\newblock
\begin{APACrefDOI} \doi{10.21468/SciPostPhys.8.4.068} \end{APACrefDOI}
\PrintBackRefs{\CurrentBib}

\bibitem [\protect \citeauthoryear {%
Bertini%
, Kos%
\BCBL {}\ \BBA {} Prosen%
}{%
Bertini%
\ \protect \BOthers {.}}{%
{\protect \APACyear {2021}}%
}]{%
bertini2021random}
\APACinsertmetastar {%
bertini2021random}%
\begin{APACrefauthors}%
Bertini, B.%
, Kos, P.%
\BCBL {}\ \BBA {} Prosen, T.%
\end{APACrefauthors}%
\unskip\
\newblock
\APACrefYearMonthDay{2021}{}{}.
\newblock
{\BBOQ}\APACrefatitle {Random matrix spectral form factor of dual-unitary
  quantum circuits} {Random matrix spectral form factor of dual-unitary quantum
  circuits}.{\BBCQ}
\newblock
\APACjournalVolNumPages{Comm. Math. Phys.}{387}{1}{597--620}.
\newblock
\begin{APACrefDOI} \doi{10.1007/s00220-021-04139-2} \end{APACrefDOI}
\PrintBackRefs{\CurrentBib}

\bibitem [\protect \citeauthoryear {%
Bohigas%
, Giannoni%
\BCBL {}\ \BBA {} Schmit%
}{%
Bohigas%
\ \protect \BOthers {.}}{%
{\protect \APACyear {1984}}%
}]{%
bohigas1984characterisation}
\APACinsertmetastar {%
bohigas1984characterisation}%
\begin{APACrefauthors}%
Bohigas, O.%
, Giannoni, M\BPBI J.%
\BCBL {}\ \BBA {} Schmit, C.%
\end{APACrefauthors}%
\unskip\
\newblock
\APACrefYearMonthDay{1984}{{\APACmonth{01}}}{}.
\newblock
{\BBOQ}\APACrefatitle {{Characterization of Chaotic Quantum Spectra and
  Universality of Level Fluctuation Laws}} {{Characterization of Chaotic
  Quantum Spectra and Universality of Level Fluctuation Laws}}.{\BBCQ}
\newblock
\APACjournalVolNumPages{Phys. Rev. Lett.}{52}{}{1--4}.
\newblock
\begin{APACrefURL} \url{https://link.aps.org/doi/10.1103/PhysRevLett.52.1}
  \end{APACrefURL}
\newblock
\begin{APACrefDOI} \doi{10.1103/PhysRevLett.52.1} \end{APACrefDOI}
\PrintBackRefs{\CurrentBib}

\bibitem [\protect \citeauthoryear {%
Bravyi%
, Hastings%
\BCBL {}\ \BBA {} Verstraete%
}{%
Bravyi%
\ \protect \BOthers {.}}{%
{\protect \APACyear {2006}}%
}]{%
bravyi2006lieb}
\APACinsertmetastar {%
bravyi2006lieb}%
\begin{APACrefauthors}%
Bravyi, S.%
, Hastings, M\BPBI B.%
\BCBL {}\ \BBA {} Verstraete, F.%
\end{APACrefauthors}%
\unskip\
\newblock
\APACrefYearMonthDay{2006}{Jul}{}.
\newblock
{\BBOQ}\APACrefatitle {Lieb-Robinson Bounds and the Generation of Correlations
  and Topological Quantum Order} {Lieb-robinson bounds and the generation of
  correlations and topological quantum order}.{\BBCQ}
\newblock
\APACjournalVolNumPages{Phys. Rev. Lett.}{97}{}{050401}.
\newblock
\begin{APACrefURL} \url{https://link.aps.org/doi/10.1103/PhysRevLett.97.050401}
  \end{APACrefURL}
\newblock
\begin{APACrefDOI} \doi{10.1103/PhysRevLett.97.050401} \end{APACrefDOI}
\PrintBackRefs{\CurrentBib}

\bibitem [\protect \citeauthoryear {%
Calabrese%
}{%
Calabrese%
}{%
{\protect \APACyear {2020}}%
}]{%
calabrese2020entanglement}
\APACinsertmetastar {%
calabrese2020entanglement}%
\begin{APACrefauthors}%
Calabrese, P.%
\end{APACrefauthors}%
\unskip\
\newblock
\APACrefYearMonthDay{2020}{}{}.
\newblock
{\BBOQ}\APACrefatitle {{Entanglement spreading in non-equilibrium integrable
  systems}} {{Entanglement spreading in non-equilibrium integrable
  systems}}.{\BBCQ}
\newblock
\APACjournalVolNumPages{SciPost Phys. Lect. Notes}{}{}{20}.
\newblock
\begin{APACrefURL} \url{https://scipost.org/10.21468/SciPostPhysLectNotes.20}
  \end{APACrefURL}
\newblock
\begin{APACrefDOI} \doi{10.21468/SciPostPhysLectNotes.20} \end{APACrefDOI}
\PrintBackRefs{\CurrentBib}

\bibitem [\protect \citeauthoryear {%
Calabrese%
\ \BBA {} Cardy%
}{%
Calabrese%
\ \BBA {} Cardy%
}{%
{\protect \APACyear {2005}}%
}]{%
calabrese2005evolution}
\APACinsertmetastar {%
calabrese2005evolution}%
\begin{APACrefauthors}%
Calabrese, P.%
\BCBT {}\ \BBA {} Cardy, J.%
\end{APACrefauthors}%
\unskip\
\newblock
\APACrefYearMonthDay{2005}{apr}{}.
\newblock
{\BBOQ}\APACrefatitle {Evolution of entanglement entropy in one-dimensional
  systems} {Evolution of entanglement entropy in one-dimensional
  systems}.{\BBCQ}
\newblock
\APACjournalVolNumPages{J. Stat. Mech.}{2005}{04}{P04010}.
\newblock
\begin{APACrefDOI} \doi{10.1088/1742-5468/2005/04/p04010} \end{APACrefDOI}
\PrintBackRefs{\CurrentBib}

\bibitem [\protect \citeauthoryear {%
Calabrese%
\ \BBA {} Cardy%
}{%
Calabrese%
\ \BBA {} Cardy%
}{%
{\protect \APACyear {2006}}%
}]{%
calabrese2006time}
\APACinsertmetastar {%
calabrese2006time}%
\begin{APACrefauthors}%
Calabrese, P.%
\BCBT {}\ \BBA {} Cardy, J.%
\end{APACrefauthors}%
\unskip\
\newblock
\APACrefYearMonthDay{2006}{}{}.
\newblock
{\BBOQ}\APACrefatitle {Time Dependence of Correlation Functions Following a
  Quantum Quench} {Time dependence of correlation functions following a quantum
  quench}.{\BBCQ}
\newblock
\APACjournalVolNumPages{Phys. Rev. Lett.}{96}{}{136801}.
\newblock
\begin{APACrefDOI} \doi{10.1103/PhysRevLett.96.136801} \end{APACrefDOI}
\PrintBackRefs{\CurrentBib}

\bibitem [\protect \citeauthoryear {%
Casati%
, Valz-Gris%
\BCBL {}\ \BBA {} Guarnieri%
}{%
Casati%
\ \protect \BOthers {.}}{%
{\protect \APACyear {1980}}%
}]{%
casati1980on}
\APACinsertmetastar {%
casati1980on}%
\begin{APACrefauthors}%
Casati, G.%
, Valz-Gris, F.%
\BCBL {}\ \BBA {} Guarnieri, I.%
\end{APACrefauthors}%
\unskip\
\newblock
\APACrefYearMonthDay{1980}{}{}.
\newblock
{\BBOQ}\APACrefatitle {On the connection between quantization of nonintegrable
  systems and statistical theory of spectra} {On the connection between
  quantization of nonintegrable systems and statistical theory of
  spectra}.{\BBCQ}
\newblock
\APACjournalVolNumPages{Lettere al Nuovo Cimento (1971-1985)}{28}{}{279--282}.
\newblock
\begin{APACrefURL} \url{https://link.springer.com/article/10.1007/BF02798790}
  \end{APACrefURL}
\PrintBackRefs{\CurrentBib}

\bibitem [\protect \citeauthoryear {%
Chan%
, De~Luca%
\BCBL {}\ \BBA {} Chalker%
}{%
Chan%
\ \protect \BOthers {.}}{%
{\protect \APACyear {2018}}%
{\protect \APACexlab {{\protect \BCnt {1}}}}}]{%
chan2018solution}
\APACinsertmetastar {%
chan2018solution}%
\begin{APACrefauthors}%
Chan, A.%
, De~Luca, A.%
\BCBL {}\ \BBA {} Chalker, J\BPBI T.%
\end{APACrefauthors}%
\unskip\
\newblock
\APACrefYearMonthDay{2018{\protect \BCnt {1}}}{Nov}{}.
\newblock
{\BBOQ}\APACrefatitle {Solution of a Minimal Model for Many-Body Quantum Chaos}
  {Solution of a minimal model for many-body quantum chaos}.{\BBCQ}
\newblock
\APACjournalVolNumPages{Phys. Rev. X}{8}{4}{041019}.
\newblock
\begin{APACrefDOI} \doi{10.1103/PhysRevX.8.041019} \end{APACrefDOI}
\PrintBackRefs{\CurrentBib}

\bibitem [\protect \citeauthoryear {%
Chan%
, De~Luca%
\BCBL {}\ \BBA {} Chalker%
}{%
Chan%
\ \protect \BOthers {.}}{%
{\protect \APACyear {2018}}%
{\protect \APACexlab {{\protect \BCnt {2}}}}}]{%
chan2018spectral}
\APACinsertmetastar {%
chan2018spectral}%
\begin{APACrefauthors}%
Chan, A.%
, De~Luca, A.%
\BCBL {}\ \BBA {} Chalker, J\BPBI T.%
\end{APACrefauthors}%
\unskip\
\newblock
\APACrefYearMonthDay{2018{\protect \BCnt {2}}}{{\APACmonth{08}}}{}.
\newblock
{\BBOQ}\APACrefatitle {{Spectral Statistics in Spatially Extended Chaotic
  Quantum Many-Body Systems}} {{Spectral Statistics in Spatially Extended
  Chaotic Quantum Many-Body Systems}}.{\BBCQ}
\newblock
\APACjournalVolNumPages{Phys. Rev. Lett.}{121}{}{060601}.
\newblock
\begin{APACrefURL}
  \url{https://link.aps.org/doi/10.1103/PhysRevLett.121.060601}
  \end{APACrefURL}
\newblock
\begin{APACrefDOI} \doi{10.1103/PhysRevLett.121.060601} \end{APACrefDOI}
\PrintBackRefs{\CurrentBib}

\bibitem [\protect \citeauthoryear {%
Chan%
, Shivam%
, Huse%
\BCBL {}\ \BBA {} De~Luca%
}{%
Chan%
\ \protect \BOthers {.}}{%
{\protect \APACyear {2022}}%
}]{%
chan2022many}
\APACinsertmetastar {%
chan2022many}%
\begin{APACrefauthors}%
Chan, A.%
, Shivam, S.%
, Huse, D\BPBI A.%
\BCBL {}\ \BBA {} De~Luca, A.%
\end{APACrefauthors}%
\unskip\
\newblock
\APACrefYearMonthDay{2022}{}{}.
\newblock
{\BBOQ}\APACrefatitle {Many-body quantum chaos and space-time translational
  invariance} {Many-body quantum chaos and space-time translational
  invariance}.{\BBCQ}
\newblock
\APACjournalVolNumPages{Nat. Comm.}{13}{1}{7484}.
\newblock
\begin{APACrefDOI} \doi{10.1038/s41467-022-34318-1} \end{APACrefDOI}
\PrintBackRefs{\CurrentBib}

\bibitem [\protect \citeauthoryear {%
Collins%
, Matsumoto%
\BCBL {}\ \BBA {} Novak%
}{%
Collins%
\ \protect \BOthers {.}}{%
{\protect \APACyear {2022}}%
}]{%
collins2022weingarten}
\APACinsertmetastar {%
collins2022weingarten}%
\begin{APACrefauthors}%
Collins, B.%
, Matsumoto, S.%
\BCBL {}\ \BBA {} Novak, J.%
\end{APACrefauthors}%
\unskip\
\newblock
\APACrefYearMonthDay{2022}{}{}.
\newblock
{\BBOQ}\APACrefatitle {The weingarten calculus} {The weingarten
  calculus}.{\BBCQ}
\newblock
\APACjournalVolNumPages{Notices of the American Mathematical
  Society}{69}{05}{1}.
\newblock
\begin{APACrefDOI} \doi{10.1090/noti2474} \end{APACrefDOI}
\PrintBackRefs{\CurrentBib}

\bibitem [\protect \citeauthoryear {%
Cuiper%
, Wiesiolek%
\BCBL {}\ \BBA {} Verstraete%
}{%
Cuiper%
\ \protect \BOthers {.}}{%
{\protect \APACyear {2026}}%
}]{%
cuiper2026leshouches}
\APACinsertmetastar {%
cuiper2026leshouches}%
\begin{APACrefauthors}%
Cuiper, B\BPBI V\BHBI D.%
, Wiesiolek, W.%
\BCBL {}\ \BBA {} Verstraete, F.%
\end{APACrefauthors}%
\unskip\
\newblock
\APACrefYearMonthDay{2026}{}{}.
\newblock
\APACrefbtitle {Les Houches Lecture Notes on Tensor Networks.} {Les houches
  lecture notes on tensor networks.}
\newblock
\begin{APACrefURL} \url{https://arxiv.org/abs/2512.24390} \end{APACrefURL}
\PrintBackRefs{\CurrentBib}

\bibitem [\protect \citeauthoryear {%
Essler%
\ \BBA {} Fagotti%
}{%
Essler%
\ \BBA {} Fagotti%
}{%
{\protect \APACyear {2016}}%
}]{%
essler2016quench}
\APACinsertmetastar {%
essler2016quench}%
\begin{APACrefauthors}%
Essler, F\BPBI H\BPBI L.%
\BCBT {}\ \BBA {} Fagotti, M.%
\end{APACrefauthors}%
\unskip\
\newblock
\APACrefYearMonthDay{2016}{jun}{}.
\newblock
{\BBOQ}\APACrefatitle {Quench dynamics and relaxation in isolated integrable
  quantum spin chains} {Quench dynamics and relaxation in isolated integrable
  quantum spin chains}.{\BBCQ}
\newblock
\APACjournalVolNumPages{Journal of Statistical Mechanics: Theory and
  Experiment}{2016}{6}{064002}.
\newblock
\begin{APACrefURL} \url{https://dx.doi.org/10.1088/1742-5468/2016/06/064002}
  \end{APACrefURL}
\newblock
\begin{APACrefDOI} \doi{10.1088/1742-5468/2016/06/064002} \end{APACrefDOI}
\PrintBackRefs{\CurrentBib}

\bibitem [\protect \citeauthoryear {%
Fagotti%
\ \BBA {} Calabrese%
}{%
Fagotti%
\ \BBA {} Calabrese%
}{%
{\protect \APACyear {2008}}%
}]{%
fagotti2008evolution}
\APACinsertmetastar {%
fagotti2008evolution}%
\begin{APACrefauthors}%
Fagotti, M.%
\BCBT {}\ \BBA {} Calabrese, P.%
\end{APACrefauthors}%
\unskip\
\newblock
\APACrefYearMonthDay{2008}{Jul}{}.
\newblock
{\BBOQ}\APACrefatitle {Evolution of entanglement entropy following a quantum
  quench: Analytic results for the {XY} chain in a transverse magnetic field}
  {Evolution of entanglement entropy following a quantum quench: Analytic
  results for the {XY} chain in a transverse magnetic field}.{\BBCQ}
\newblock
\APACjournalVolNumPages{Phys. Rev. A}{78}{}{010306}.
\newblock
\begin{APACrefDOI} \doi{10.1103/PhysRevA.78.010306} \end{APACrefDOI}
\PrintBackRefs{\CurrentBib}

\bibitem [\protect \citeauthoryear {%
Feynman%
}{%
Feynman%
}{%
{\protect \APACyear {2018}}%
}]{%
feynman2018simulating}
\APACinsertmetastar {%
feynman2018simulating}%
\begin{APACrefauthors}%
Feynman, R\BPBI P.%
\end{APACrefauthors}%
\unskip\
\newblock
\APACrefYearMonthDay{2018}{}{}.
\newblock
{\BBOQ}\APACrefatitle {Simulating physics with computers} {Simulating physics
  with computers}.{\BBCQ}
\newblock
\BIn{} \APACrefbtitle {Feynman and computation} {Feynman and computation}\
  (\BPGS\ 133--153).
\newblock
\APACaddressPublisher{}{cRc Press}.
\PrintBackRefs{\CurrentBib}

\bibitem [\protect \citeauthoryear {%
Fisher%
, Khemani%
, Nahum%
\BCBL {}\ \BBA {} Vijay%
}{%
Fisher%
\ \protect \BOthers {.}}{%
{\protect \APACyear {2023}}%
}]{%
fisher2023random}
\APACinsertmetastar {%
fisher2023random}%
\begin{APACrefauthors}%
Fisher, M\BPBI P.%
, Khemani, V.%
, Nahum, A.%
\BCBL {}\ \BBA {} Vijay, S.%
\end{APACrefauthors}%
\unskip\
\newblock
\APACrefYearMonthDay{2023}{}{}.
\newblock
{\BBOQ}\APACrefatitle {Random {Quantum} {Circuits}} {Random {Quantum}
  {Circuits}}.{\BBCQ}
\newblock
\APACjournalVolNumPages{Annu. Rev. Conden. Ma. P.}{14}{1}{335-379}.
\newblock
\begin{APACrefURL}
  [{2023-01-11}]\url{https://www.annualreviews.org/content/journals/10.1146/annurev-conmatphys-031720-030658}
  \end{APACrefURL}
\newblock
\begin{APACrefDOI} \doi{10.1146/annurev-conmatphys-031720-030658}
  \end{APACrefDOI}
\PrintBackRefs{\CurrentBib}

\bibitem [\protect \citeauthoryear {%
Foligno%
\ \BBA {} Bertini%
}{%
Foligno%
\ \BBA {} Bertini%
}{%
{\protect \APACyear {2023}}%
}]{%
foligno2023growth}
\APACinsertmetastar {%
foligno2023growth}%
\begin{APACrefauthors}%
Foligno, A.%
\BCBT {}\ \BBA {} Bertini, B.%
\end{APACrefauthors}%
\unskip\
\newblock
\APACrefYearMonthDay{2023}{May}{}.
\newblock
{\BBOQ}\APACrefatitle {Growth of entanglement of generic states under
  dual-unitary dynamics} {Growth of entanglement of generic states under
  dual-unitary dynamics}.{\BBCQ}
\newblock
\APACjournalVolNumPages{Phys. Rev. B}{107}{}{174311}.
\newblock
\begin{APACrefURL} \url{https://link.aps.org/doi/10.1103/PhysRevB.107.174311}
  \end{APACrefURL}
\newblock
\begin{APACrefDOI} \doi{10.1103/PhysRevB.107.174311} \end{APACrefDOI}
\PrintBackRefs{\CurrentBib}

\bibitem [\protect \citeauthoryear {%
Foligno%
\ \BBA {} Bertini%
}{%
Foligno%
\ \BBA {} Bertini%
}{%
{\protect \APACyear {2025}}%
}]{%
foligno2024entanglement}
\APACinsertmetastar {%
foligno2024entanglement}%
\begin{APACrefauthors}%
Foligno, A.%
\BCBT {}\ \BBA {} Bertini, B.%
\end{APACrefauthors}%
\unskip\
\newblock
\APACrefYearMonthDay{2025}{}{}.
\newblock
{\BBOQ}\APACrefatitle {Entanglement of {D}isjoint {I}ntervals in
  {D}ual-{U}nitary {C}ircuits: {E}xact {R}esults} {Entanglement of {D}isjoint
  {I}ntervals in {D}ual-{U}nitary {C}ircuits: {E}xact {R}esults}.{\BBCQ}
\newblock
\APACjournalVolNumPages{{Quantum}}{9}{}{1678}.
\newblock
\begin{APACrefURL} \url{https://doi.org/10.22331/q-2025-03-26-1678}
  \end{APACrefURL}
\newblock
\begin{APACrefDOI} \doi{10.22331/q-2025-03-26-1678} \end{APACrefDOI}
\PrintBackRefs{\CurrentBib}

\bibitem [\protect \citeauthoryear {%
Foligno%
, Calabrese%
\BCBL {}\ \BBA {} Bertini%
}{%
Foligno%
\ \protect \BOthers {.}}{%
{\protect \APACyear {2025}}%
}]{%
foligno2025nonequilibrium}
\APACinsertmetastar {%
foligno2025nonequilibrium}%
\begin{APACrefauthors}%
Foligno, A.%
, Calabrese, P.%
\BCBL {}\ \BBA {} Bertini, B.%
\end{APACrefauthors}%
\unskip\
\newblock
\APACrefYearMonthDay{2025}{Feb}{}.
\newblock
{\BBOQ}\APACrefatitle {Nonequilibrium Dynamics of Charged Dual-Unitary
  Circuits} {Nonequilibrium dynamics of charged dual-unitary circuits}.{\BBCQ}
\newblock
\APACjournalVolNumPages{PRX Quantum}{6}{}{010324}.
\newblock
\begin{APACrefURL} \url{https://link.aps.org/doi/10.1103/PRXQuantum.6.010324}
  \end{APACrefURL}
\newblock
\begin{APACrefDOI} \doi{10.1103/PRXQuantum.6.010324} \end{APACrefDOI}
\PrintBackRefs{\CurrentBib}

\bibitem [\protect \citeauthoryear {%
Gopalakrishnan%
\ \BBA {} Lamacraft%
}{%
Gopalakrishnan%
\ \BBA {} Lamacraft%
}{%
{\protect \APACyear {2019}}%
}]{%
gopalakrishnan2019unitary}
\APACinsertmetastar {%
gopalakrishnan2019unitary}%
\begin{APACrefauthors}%
Gopalakrishnan, S.%
\BCBT {}\ \BBA {} Lamacraft, A.%
\end{APACrefauthors}%
\unskip\
\newblock
\APACrefYearMonthDay{2019}{Aug}{}.
\newblock
{\BBOQ}\APACrefatitle {Unitary circuits of finite depth and infinite width from
  quantum channels} {Unitary circuits of finite depth and infinite width from
  quantum channels}.{\BBCQ}
\newblock
\APACjournalVolNumPages{Phys. Rev. B}{100}{}{064309}.
\newblock
\begin{APACrefDOI} \doi{10.1103/PhysRevB.100.064309} \end{APACrefDOI}
\PrintBackRefs{\CurrentBib}

\bibitem [\protect \citeauthoryear {%
Haah%
, Hastings%
, Kothari%
\BCBL {}\ \BBA {} Low%
}{%
Haah%
\ \protect \BOthers {.}}{%
{\protect \APACyear {2023}}%
}]{%
haah2023quantum}
\APACinsertmetastar {%
haah2023quantum}%
\begin{APACrefauthors}%
Haah, J.%
, Hastings, M\BPBI B.%
, Kothari, R.%
\BCBL {}\ \BBA {} Low, G\BPBI H.%
\end{APACrefauthors}%
\unskip\
\newblock
\APACrefYearMonthDay{2023}{}{}.
\newblock
{\BBOQ}\APACrefatitle {Quantum Algorithm for Simulating Real Time Evolution of
  Lattice Hamiltonians} {Quantum algorithm for simulating real time evolution
  of lattice hamiltonians}.{\BBCQ}
\newblock
\APACjournalVolNumPages{SIAM Journal on
  Computing}{52}{6}{FOCS18-250-FOCS18-284}.
\newblock
\begin{APACrefURL} \url{https://doi.org/10.1137/18M1231511} \end{APACrefURL}
\newblock
\begin{APACrefDOI} \doi{10.1137/18M1231511} \end{APACrefDOI}
\PrintBackRefs{\CurrentBib}

\bibitem [\protect \citeauthoryear {%
Harrow%
\ \BBA {} Low%
}{%
Harrow%
\ \BBA {} Low%
}{%
{\protect \APACyear {2009}}%
}]{%
harrow2009random}
\APACinsertmetastar {%
harrow2009random}%
\begin{APACrefauthors}%
Harrow, A\BPBI W.%
\BCBT {}\ \BBA {} Low, R\BPBI A.%
\end{APACrefauthors}%
\unskip\
\newblock
\APACrefYearMonthDay{2009}{}{}.
\newblock
{\BBOQ}\APACrefatitle {Random quantum circuits are approximate 2-designs}
  {Random quantum circuits are approximate 2-designs}.{\BBCQ}
\newblock
\APACjournalVolNumPages{Communications in Mathematical
  Physics}{291}{1}{257--302}.
\newblock
\begin{APACrefDOI} \doi{10.1007/s00220-009-0873-6} \end{APACrefDOI}
\PrintBackRefs{\CurrentBib}

\bibitem [\protect \citeauthoryear {%
Huang%
}{%
Huang%
}{%
{\protect \APACyear {2020}}%
}]{%
huang2020dynamics}
\APACinsertmetastar {%
huang2020dynamics}%
\begin{APACrefauthors}%
Huang, Y.%
\end{APACrefauthors}%
\unskip\
\newblock
\APACrefYearMonthDay{2020}{}{}.
\newblock
{\BBOQ}\APACrefatitle {Dynamics of {R\'enyi} entanglement entropy in diffusive
  qudit systems} {Dynamics of {R\'enyi} entanglement entropy in diffusive qudit
  systems}.{\BBCQ}
\newblock
\APACjournalVolNumPages{IOP SciNotes}{1}{3}{035205}.
\newblock
\begin{APACrefURL} \url{https://dx.doi.org/10.1088/2633-1357/abd1e2}
  \end{APACrefURL}
\newblock
\begin{APACrefDOI} \doi{10.1088/2633-1357/abd1e2} \end{APACrefDOI}
\PrintBackRefs{\CurrentBib}

\bibitem [\protect \citeauthoryear {%
Ilievski%
, Medenjak%
, Prosen%
\BCBL {}\ \BBA {} Zadnik%
}{%
Ilievski%
\ \protect \BOthers {.}}{%
{\protect \APACyear {2016}}%
}]{%
ilievski2016quasi}
\APACinsertmetastar {%
ilievski2016quasi}%
\begin{APACrefauthors}%
Ilievski, E.%
, Medenjak, M.%
, Prosen, T.%
\BCBL {}\ \BBA {} Zadnik, L.%
\end{APACrefauthors}%
\unskip\
\newblock
\APACrefYearMonthDay{2016}{jun}{}.
\newblock
{\BBOQ}\APACrefatitle {Quasilocal charges in integrable lattice systems}
  {Quasilocal charges in integrable lattice systems}.{\BBCQ}
\newblock
\APACjournalVolNumPages{Journal of Statistical Mechanics: Theory and
  Experiment}{2016}{6}{064008}.
\newblock
\begin{APACrefURL} \url{https://dx.doi.org/10.1088/1742-5468/2016/06/064008}
  \end{APACrefURL}
\newblock
\begin{APACrefDOI} \doi{10.1088/1742-5468/2016/06/064008} \end{APACrefDOI}
\PrintBackRefs{\CurrentBib}

\bibitem [\protect \citeauthoryear {%
Kim%
\ \BBA {} Huse%
}{%
Kim%
\ \BBA {} Huse%
}{%
{\protect \APACyear {2013}}%
}]{%
kim2013ballistic}
\APACinsertmetastar {%
kim2013ballistic}%
\begin{APACrefauthors}%
Kim, H.%
\BCBT {}\ \BBA {} Huse, D\BPBI A.%
\end{APACrefauthors}%
\unskip\
\newblock
\APACrefYearMonthDay{2013}{Sep}{}.
\newblock
{\BBOQ}\APACrefatitle {Ballistic Spreading of Entanglement in a Diffusive
  Nonintegrable System} {Ballistic spreading of entanglement in a diffusive
  nonintegrable system}.{\BBCQ}
\newblock
\APACjournalVolNumPages{Phys. Rev. Lett.}{111}{}{127205}.
\newblock
\begin{APACrefDOI} \doi{10.1103/PhysRevLett.111.127205} \end{APACrefDOI}
\PrintBackRefs{\CurrentBib}

\bibitem [\protect \citeauthoryear {%
Klobas%
, Rylands%
\BCBL {}\ \BBA {} Bertini%
}{%
Klobas%
\ \protect \BOthers {.}}{%
{\protect \APACyear {2025}}%
}]{%
klobas2024translation}
\APACinsertmetastar {%
klobas2024translation}%
\begin{APACrefauthors}%
Klobas, K.%
, Rylands, C.%
\BCBL {}\ \BBA {} Bertini, B.%
\end{APACrefauthors}%
\unskip\
\newblock
\APACrefYearMonthDay{2025}{Apr}{}.
\newblock
{\BBOQ}\APACrefatitle {Translation symmetry restoration under random unitary
  dynamics} {Translation symmetry restoration under random unitary
  dynamics}.{\BBCQ}
\newblock
\APACjournalVolNumPages{Phys. Rev. B}{111}{}{L140304}.
\newblock
\begin{APACrefURL} \url{https://link.aps.org/doi/10.1103/PhysRevB.111.L140304}
  \end{APACrefURL}
\newblock
\begin{APACrefDOI} \doi{10.1103/PhysRevB.111.L140304} \end{APACrefDOI}
\PrintBackRefs{\CurrentBib}

\bibitem [\protect \citeauthoryear {%
Kuwahara%
\ \BBA {} Saito%
}{%
Kuwahara%
\ \BBA {} Saito%
}{%
{\protect \APACyear {2021}}%
}]{%
kuwahara2021lieb}
\APACinsertmetastar {%
kuwahara2021lieb}%
\begin{APACrefauthors}%
Kuwahara, T.%
\BCBT {}\ \BBA {} Saito, K.%
\end{APACrefauthors}%
\unskip\
\newblock
\APACrefYearMonthDay{2021}{Aug}{}.
\newblock
{\BBOQ}\APACrefatitle {Lieb-Robinson Bound and Almost-Linear Light Cone in
  Interacting Boson Systems} {Lieb-robinson bound and almost-linear light cone
  in interacting boson systems}.{\BBCQ}
\newblock
\APACjournalVolNumPages{Phys. Rev. Lett.}{127}{}{070403}.
\newblock
\begin{APACrefURL}
  \url{https://link.aps.org/doi/10.1103/PhysRevLett.127.070403}
  \end{APACrefURL}
\newblock
\begin{APACrefDOI} \doi{10.1103/PhysRevLett.127.070403} \end{APACrefDOI}
\PrintBackRefs{\CurrentBib}

\bibitem [\protect \citeauthoryear {%
L\"auchli%
\ \BBA {} Kollath%
}{%
L\"auchli%
\ \BBA {} Kollath%
}{%
{\protect \APACyear {2008}}%
}]{%
laeuchli2008spreading}
\APACinsertmetastar {%
laeuchli2008spreading}%
\begin{APACrefauthors}%
L\"auchli, A\BPBI M.%
\BCBT {}\ \BBA {} Kollath, C.%
\end{APACrefauthors}%
\unskip\
\newblock
\APACrefYearMonthDay{2008}{may}{}.
\newblock
{\BBOQ}\APACrefatitle {Spreading of correlations and entanglement after a
  quench in the one-dimensional Bose--Hubbard model} {Spreading of correlations
  and entanglement after a quench in the one-dimensional bose--hubbard
  model}.{\BBCQ}
\newblock
\APACjournalVolNumPages{J. Stat. Mech.}{2008}{05}{P05018}.
\newblock
\begin{APACrefDOI} \doi{10.1088/1742-5468/2008/05/p05018} \end{APACrefDOI}
\PrintBackRefs{\CurrentBib}

\bibitem [\protect \citeauthoryear {%
Lieb%
\ \BBA {} Robinson%
}{%
Lieb%
\ \BBA {} Robinson%
}{%
{\protect \APACyear {1972}}%
}]{%
lieb1972finite}
\APACinsertmetastar {%
lieb1972finite}%
\begin{APACrefauthors}%
Lieb, E\BPBI H.%
\BCBT {}\ \BBA {} Robinson, D\BPBI W.%
\end{APACrefauthors}%
\unskip\
\newblock
\APACrefYearMonthDay{1972}{}{}.
\newblock
{\BBOQ}\APACrefatitle {The finite group velocity of quantum spin systems} {The
  finite group velocity of quantum spin systems}.{\BBCQ}
\newblock
\APACjournalVolNumPages{Communications in mathematical
  physics}{28}{3}{251--257}.
\newblock
\begin{APACrefDOI} \doi{10.1007/BF01645779} \end{APACrefDOI}
\PrintBackRefs{\CurrentBib}

\bibitem [\protect \citeauthoryear {%
Liu%
\ \BBA {} Suh%
}{%
Liu%
\ \BBA {} Suh%
}{%
{\protect \APACyear {2014}}%
}]{%
liu2014entanglement}
\APACinsertmetastar {%
liu2014entanglement}%
\begin{APACrefauthors}%
Liu, H.%
\BCBT {}\ \BBA {} Suh, S\BPBI J.%
\end{APACrefauthors}%
\unskip\
\newblock
\APACrefYearMonthDay{2014}{Jan}{}.
\newblock
{\BBOQ}\APACrefatitle {Entanglement Tsunami: Universal Scaling in Holographic
  Thermalization} {Entanglement tsunami: Universal scaling in holographic
  thermalization}.{\BBCQ}
\newblock
\APACjournalVolNumPages{Phys. Rev. Lett.}{112}{}{011601}.
\newblock
\begin{APACrefURL}
  \url{https://link.aps.org/doi/10.1103/PhysRevLett.112.011601}
  \end{APACrefURL}
\newblock
\begin{APACrefDOI} \doi{10.1103/PhysRevLett.112.011601} \end{APACrefDOI}
\PrintBackRefs{\CurrentBib}

\bibitem [\protect \citeauthoryear {%
Mehta%
}{%
Mehta%
}{%
{\protect \APACyear {2004}}%
}]{%
mehta2004random}
\APACinsertmetastar {%
mehta2004random}%
\begin{APACrefauthors}%
Mehta, M.%
\end{APACrefauthors}%
\unskip\
\newblock
\APACrefYear{2004}.
\newblock
\APACrefbtitle {Random Matrices} {Random matrices}.
\newblock
\APACaddressPublisher{}{Academic Press}.
\newblock
\begin{APACrefURL} \url{https://books.google.co.uk/books?id=Kp3Nx03_gMwC}
  \end{APACrefURL}
\PrintBackRefs{\CurrentBib}

\bibitem [\protect \citeauthoryear {%
M\"uller%
, Heusler%
, Braun%
, Haake%
\BCBL {}\ \BBA {} Altland%
}{%
M\"uller%
\ \protect \BOthers {.}}{%
{\protect \APACyear {2004}}%
}]{%
muller2004semiclassical}
\APACinsertmetastar {%
muller2004semiclassical}%
\begin{APACrefauthors}%
M\"uller, S.%
, Heusler, S.%
, Braun, P.%
, Haake, F.%
\BCBL {}\ \BBA {} Altland, A.%
\end{APACrefauthors}%
\unskip\
\newblock
\APACrefYearMonthDay{2004}{{\APACmonth{07}}}{}.
\newblock
{\BBOQ}\APACrefatitle {Semiclassical Foundation of Universality in Quantum
  Chaos} {Semiclassical foundation of universality in quantum chaos}.{\BBCQ}
\newblock
\APACjournalVolNumPages{Phys. Rev. Lett.}{93}{}{014103}.
\newblock
\begin{APACrefURL} \url{https://link.aps.org/doi/10.1103/PhysRevLett.93.014103}
  \end{APACrefURL}
\newblock
\begin{APACrefDOI} \doi{10.1103/PhysRevLett.93.014103} \end{APACrefDOI}
\PrintBackRefs{\CurrentBib}

\bibitem [\protect \citeauthoryear {%
Nahum%
, Ruhman%
, Vijay%
\BCBL {}\ \BBA {} Haah%
}{%
Nahum%
\ \protect \BOthers {.}}{%
{\protect \APACyear {2017}}%
}]{%
nahum2017quantum}
\APACinsertmetastar {%
nahum2017quantum}%
\begin{APACrefauthors}%
Nahum, A.%
, Ruhman, J.%
, Vijay, S.%
\BCBL {}\ \BBA {} Haah, J.%
\end{APACrefauthors}%
\unskip\
\newblock
\APACrefYearMonthDay{2017}{Jul}{}.
\newblock
{\BBOQ}\APACrefatitle {Quantum Entanglement Growth under Random Unitary
  Dynamics} {Quantum entanglement growth under random unitary dynamics}.{\BBCQ}
\newblock
\APACjournalVolNumPages{Phys. Rev. X}{7}{}{031016}.
\newblock
\begin{APACrefDOI} \doi{10.1103/PhysRevX.7.031016} \end{APACrefDOI}
\PrintBackRefs{\CurrentBib}

\bibitem [\protect \citeauthoryear {%
Osborne%
}{%
Osborne%
}{%
{\protect \APACyear {2006}}%
}]{%
osborne2006efficient}
\APACinsertmetastar {%
osborne2006efficient}%
\begin{APACrefauthors}%
Osborne, T\BPBI J.%
\end{APACrefauthors}%
\unskip\
\newblock
\APACrefYearMonthDay{2006}{Oct}{}.
\newblock
{\BBOQ}\APACrefatitle {Efficient Approximation of the Dynamics of
  One-Dimensional Quantum Spin Systems} {Efficient approximation of the
  dynamics of one-dimensional quantum spin systems}.{\BBCQ}
\newblock
\APACjournalVolNumPages{Phys. Rev. Lett.}{97}{}{157202}.
\newblock
\begin{APACrefURL} \url{https://link.aps.org/doi/10.1103/PhysRevLett.97.157202}
  \end{APACrefURL}
\newblock
\begin{APACrefDOI} \doi{10.1103/PhysRevLett.97.157202} \end{APACrefDOI}
\PrintBackRefs{\CurrentBib}

\bibitem [\protect \citeauthoryear {%
Piroli%
, Bertini%
, Cirac%
\BCBL {}\ \BBA {} Prosen%
}{%
Piroli%
\ \protect \BOthers {.}}{%
{\protect \APACyear {2020}}%
}]{%
piroli2020exact}
\APACinsertmetastar {%
piroli2020exact}%
\begin{APACrefauthors}%
Piroli, L.%
, Bertini, B.%
, Cirac, J\BPBI I.%
\BCBL {}\ \BBA {} Prosen, T.%
\end{APACrefauthors}%
\unskip\
\newblock
\APACrefYearMonthDay{2020}{Mar}{}.
\newblock
{\BBOQ}\APACrefatitle {Exact dynamics in dual-unitary quantum circuits} {Exact
  dynamics in dual-unitary quantum circuits}.{\BBCQ}
\newblock
\APACjournalVolNumPages{Phys. Rev. B}{101}{}{094304}.
\newblock
\begin{APACrefDOI} \doi{10.1103/PhysRevB.101.094304} \end{APACrefDOI}
\PrintBackRefs{\CurrentBib}

\bibitem [\protect \citeauthoryear {%
Potter%
\ \BBA {} Vasseur%
}{%
Potter%
\ \BBA {} Vasseur%
}{%
{\protect \APACyear {2022}}%
}]{%
potter2022entanglement}
\APACinsertmetastar {%
potter2022entanglement}%
\begin{APACrefauthors}%
Potter, A\BPBI C.%
\BCBT {}\ \BBA {} Vasseur, R.%
\end{APACrefauthors}%
\unskip\
\newblock
\APACrefYearMonthDay{2022}{}{}.
\newblock
{\BBOQ}\APACrefatitle {Entanglement Dynamics in Hybrid Quantum Circuits}
  {Entanglement dynamics in hybrid quantum circuits}.{\BBCQ}
\newblock
\BIn{} A.~Bayat, S.~Bose\BCBL {}\ \BBA {} H.~Johannesson\ (\BEDS),
  \APACrefbtitle {Entanglement in Spin Chains: From Theory to Quantum
  Technology Applications} {Entanglement in spin chains: From theory to quantum
  technology applications}\ (\BPGS\ 211--249).
\newblock
\APACaddressPublisher{Cham}{Springer International Publishing}.
\newblock
\begin{APACrefURL} \url{https://doi.org/10.1007/978-3-031-03998-0_9}
  \end{APACrefURL}
\newblock
\begin{APACrefDOI} \doi{10.1007/978-3-031-03998-0_9} \end{APACrefDOI}
\PrintBackRefs{\CurrentBib}

\bibitem [\protect \citeauthoryear {%
Prange%
}{%
Prange%
}{%
{\protect \APACyear {1997}}%
}]{%
prange1997the}
\APACinsertmetastar {%
prange1997the}%
\begin{APACrefauthors}%
Prange, R\BPBI E.%
\end{APACrefauthors}%
\unskip\
\newblock
\APACrefYearMonthDay{1997}{{\APACmonth{03}}}{}.
\newblock
{\BBOQ}\APACrefatitle {The Spectral Form Factor Is Not Self-Averaging} {The
  spectral form factor is not self-averaging}.{\BBCQ}
\newblock
\APACjournalVolNumPages{Phys. Rev. Lett.}{78}{}{2280--2283}.
\newblock
\begin{APACrefURL} \url{https://link.aps.org/doi/10.1103/PhysRevLett.78.2280}
  \end{APACrefURL}
\newblock
\begin{APACrefDOI} \doi{10.1103/PhysRevLett.78.2280} \end{APACrefDOI}
\PrintBackRefs{\CurrentBib}

\bibitem [\protect \citeauthoryear {%
Prosen%
\ \BBA {} Pi{\v z}orn%
}{%
Prosen%
\ \BBA {} Pi{\v z}orn%
}{%
{\protect \APACyear {2007}}%
}]{%
prosen2007operator}
\APACinsertmetastar {%
prosen2007operator}%
\begin{APACrefauthors}%
Prosen, T.%
\BCBT {}\ \BBA {} Pi{\v z}orn, I.%
\end{APACrefauthors}%
\unskip\
\newblock
\APACrefYearMonthDay{2007}{}{}.
\newblock
{\BBOQ}\APACrefatitle {Operator space entanglement entropy in a transverse
  {Ising} chain} {Operator space entanglement entropy in a transverse {Ising}
  chain}.{\BBCQ}
\newblock
\APACjournalVolNumPages{Phys. Rev. A}{76}{3}{032316}.
\newblock
\begin{APACrefDOI} \doi{10.1103/physreva.76.032316} \end{APACrefDOI}
\PrintBackRefs{\CurrentBib}

\bibitem [\protect \citeauthoryear {%
Prosen%
\ \BBA {} {\v Z}nidari{\v c}%
}{%
Prosen%
\ \BBA {} {\v Z}nidari{\v c}%
}{%
{\protect \APACyear {2007}}%
}]{%
prosen2007is}
\APACinsertmetastar {%
prosen2007is}%
\begin{APACrefauthors}%
Prosen, T.%
\BCBT {}\ \BBA {} {\v Z}nidari{\v c}, M.%
\end{APACrefauthors}%
\unskip\
\newblock
\APACrefYearMonthDay{2007}{}{}.
\newblock
{\BBOQ}\APACrefatitle {Is the efficiency of classical simulations of quantum
  dynamics related to integrability?} {Is the efficiency of classical
  simulations of quantum dynamics related to integrability?}{\BBCQ}
\newblock
\APACjournalVolNumPages{Phys. Rev. E}{75}{}{015202}.
\newblock
\begin{APACrefDOI} \doi{10.1103/PhysRevE.75.015202} \end{APACrefDOI}
\PrintBackRefs{\CurrentBib}

\bibitem [\protect \citeauthoryear {%
Rakovszky%
, Pollmann%
\BCBL {}\ \BBA {} von Keyserlingk%
}{%
Rakovszky%
\ \protect \BOthers {.}}{%
{\protect \APACyear {2019}}%
}]{%
rakovszky2019sub}
\APACinsertmetastar {%
rakovszky2019sub}%
\begin{APACrefauthors}%
Rakovszky, T.%
, Pollmann, F.%
\BCBL {}\ \BBA {} von Keyserlingk, C\BPBI W.%
\end{APACrefauthors}%
\unskip\
\newblock
\APACrefYearMonthDay{2019}{Jun}{}.
\newblock
{\BBOQ}\APACrefatitle {Sub-ballistic Growth of {R}\'enyi Entropies due to
  Diffusion} {Sub-ballistic growth of {R}\'enyi entropies due to
  diffusion}.{\BBCQ}
\newblock
\APACjournalVolNumPages{Phys. Rev. Lett.}{122}{25}{250602}.
\newblock
\begin{APACrefDOI} \doi{10.1103/PhysRevLett.122.250602} \end{APACrefDOI}
\PrintBackRefs{\CurrentBib}

\bibitem [\protect \citeauthoryear {%
Suzuki%
}{%
Suzuki%
}{%
{\protect \APACyear {1990}}%
}]{%
suzuki1990fractal}
\APACinsertmetastar {%
suzuki1990fractal}%
\begin{APACrefauthors}%
Suzuki, M.%
\end{APACrefauthors}%
\unskip\
\newblock
\APACrefYearMonthDay{1990}{}{}.
\newblock
{\BBOQ}\APACrefatitle {Fractal decomposition of exponential operators with
  applications to many-body theories and Monte Carlo simulations} {Fractal
  decomposition of exponential operators with applications to many-body
  theories and monte carlo simulations}.{\BBCQ}
\newblock
\APACjournalVolNumPages{Physics Letters A}{146}{6}{319-323}.
\newblock
\begin{APACrefURL}
  \url{https://www.sciencedirect.com/science/article/pii/037596019090962N}
  \end{APACrefURL}
\newblock
\begin{APACrefDOI} \doi{https://doi.org/10.1016/0375-9601(90)90962-N}
  \end{APACrefDOI}
\PrintBackRefs{\CurrentBib}

\bibitem [\protect \citeauthoryear {%
Trotter%
}{%
Trotter%
}{%
{\protect \APACyear {1959}}%
}]{%
trotter1959product}
\APACinsertmetastar {%
trotter1959product}%
\begin{APACrefauthors}%
Trotter, H\BPBI F.%
\end{APACrefauthors}%
\unskip\
\newblock
\APACrefYearMonthDay{1959}{}{}.
\newblock
{\BBOQ}\APACrefatitle {On the product of semi-groups of operators} {On the
  product of semi-groups of operators}.{\BBCQ}
\newblock
\APACjournalVolNumPages{Proceedings of the American Mathematical
  Society}{10}{4}{545--551}.
\PrintBackRefs{\CurrentBib}

\bibitem [\protect \citeauthoryear {%
Vasseur%
}{%
Vasseur%
}{%
{\protect \APACyear {2026}}%
}]{%
vasseur2026leshouches}
\APACinsertmetastar {%
vasseur2026leshouches}%
\begin{APACrefauthors}%
Vasseur, R.%
\end{APACrefauthors}%
\unskip\
\newblock
\APACrefYearMonthDay{2026}{}{}.
\newblock
\APACrefbtitle {Les Houches lectures on random quantum circuits and monitored
  quantum dynamics.} {Les houches lectures on random quantum circuits and
  monitored quantum dynamics.}
\newblock
\begin{APACrefURL} \url{https://arxiv.org/abs/2602.17258} \end{APACrefURL}
\PrintBackRefs{\CurrentBib}

\bibitem [\protect \citeauthoryear {%
Vernier%
, Bertini%
, Giudici%
\BCBL {}\ \BBA {} Piroli%
}{%
Vernier%
\ \protect \BOthers {.}}{%
{\protect \APACyear {2023}}%
}]{%
vernier2023integrable}
\APACinsertmetastar {%
vernier2023integrable}%
\begin{APACrefauthors}%
Vernier, E.%
, Bertini, B.%
, Giudici, G.%
\BCBL {}\ \BBA {} Piroli, L.%
\end{APACrefauthors}%
\unskip\
\newblock
\APACrefYearMonthDay{2023}{Jun}{}.
\newblock
{\BBOQ}\APACrefatitle {Integrable Digital Quantum Simulation: Generalized Gibbs
  Ensembles and Trotter Transitions} {Integrable digital quantum simulation:
  Generalized gibbs ensembles and trotter transitions}.{\BBCQ}
\newblock
\APACjournalVolNumPages{Phys. Rev. Lett.}{130}{}{260401}.
\newblock
\begin{APACrefURL}
  \url{https://link.aps.org/doi/10.1103/PhysRevLett.130.260401}
  \end{APACrefURL}
\newblock
\begin{APACrefDOI} \doi{10.1103/PhysRevLett.130.260401} \end{APACrefDOI}
\PrintBackRefs{\CurrentBib}

\bibitem [\protect \citeauthoryear {%
Zhou%
\ \BBA {} Harrow%
}{%
Zhou%
\ \BBA {} Harrow%
}{%
{\protect \APACyear {2022}}%
}]{%
zhou2022maximal}
\APACinsertmetastar {%
zhou2022maximal}%
\begin{APACrefauthors}%
Zhou, T.%
\BCBT {}\ \BBA {} Harrow, A\BPBI W.%
\end{APACrefauthors}%
\unskip\
\newblock
\APACrefYearMonthDay{2022}{Nov}{}.
\newblock
{\BBOQ}\APACrefatitle {Maximal entanglement velocity implies dual unitarity}
  {Maximal entanglement velocity implies dual unitarity}.{\BBCQ}
\newblock
\APACjournalVolNumPages{Phys. Rev. B}{106}{}{L201104}.
\newblock
\begin{APACrefURL} \url{https://link.aps.org/doi/10.1103/PhysRevB.106.L201104}
  \end{APACrefURL}
\newblock
\begin{APACrefDOI} \doi{10.1103/PhysRevB.106.L201104} \end{APACrefDOI}
\PrintBackRefs{\CurrentBib}

\bibitem [\protect \citeauthoryear {%
Zhou%
\ \BBA {} Nahum%
}{%
Zhou%
\ \BBA {} Nahum%
}{%
{\protect \APACyear {2019}}%
}]{%
zhou2019emergent}
\APACinsertmetastar {%
zhou2019emergent}%
\begin{APACrefauthors}%
Zhou, T.%
\BCBT {}\ \BBA {} Nahum, A.%
\end{APACrefauthors}%
\unskip\
\newblock
\APACrefYearMonthDay{2019}{May}{}.
\newblock
{\BBOQ}\APACrefatitle {Emergent statistical mechanics of entanglement in random
  unitary circuits} {Emergent statistical mechanics of entanglement in random
  unitary circuits}.{\BBCQ}
\newblock
\APACjournalVolNumPages{Phys. Rev. B}{99}{}{174205}.
\newblock
\begin{APACrefURL} \url{https://link.aps.org/doi/10.1103/PhysRevB.99.174205}
  \end{APACrefURL}
\newblock
\begin{APACrefDOI} \doi{10.1103/PhysRevB.99.174205} \end{APACrefDOI}
\PrintBackRefs{\CurrentBib}

\bibitem [\protect \citeauthoryear {%
Zhou%
\ \BBA {} Nahum%
}{%
Zhou%
\ \BBA {} Nahum%
}{%
{\protect \APACyear {2020}}%
}]{%
zhou2020entanglement}
\APACinsertmetastar {%
zhou2020entanglement}%
\begin{APACrefauthors}%
Zhou, T.%
\BCBT {}\ \BBA {} Nahum, A.%
\end{APACrefauthors}%
\unskip\
\newblock
\APACrefYearMonthDay{2020}{Sep}{}.
\newblock
{\BBOQ}\APACrefatitle {Entanglement Membrane in Chaotic Many-Body Systems}
  {Entanglement membrane in chaotic many-body systems}.{\BBCQ}
\newblock
\APACjournalVolNumPages{Phys. Rev. X}{10}{}{031066}.
\newblock
\begin{APACrefDOI} \doi{10.1103/PhysRevX.10.031066} \end{APACrefDOI}
\PrintBackRefs{\CurrentBib}

\bibitem [\protect \citeauthoryear {%
{\v Z}nidari{\v c}%
}{%
{\v Z}nidari{\v c}%
}{%
{\protect \APACyear {2008}}%
}]{%
znidaric2008exact}
\APACinsertmetastar {%
znidaric2008exact}%
\begin{APACrefauthors}%
{\v Z}nidari{\v c}, M.%
\end{APACrefauthors}%
\unskip\
\newblock
\APACrefYearMonthDay{2008}{Sep}{}.
\newblock
{\BBOQ}\APACrefatitle {Exact convergence times for generation of random
  bipartite entanglement} {Exact convergence times for generation of random
  bipartite entanglement}.{\BBCQ}
\newblock
\APACjournalVolNumPages{Phys. Rev. A}{78}{}{032324}.
\newblock
\begin{APACrefURL} \url{https://link.aps.org/doi/10.1103/PhysRevA.78.032324}
  \end{APACrefURL}
\newblock
\begin{APACrefDOI} \doi{10.1103/PhysRevA.78.032324} \end{APACrefDOI}
\PrintBackRefs{\CurrentBib}

\end{thebibliography}

\end{document}